\documentclass[12pt]{article}
\usepackage{amsthm,amsmath,amsfonts,amssymb,enumitem,aliascnt,MnSymbol}
\usepackage{comment}
\usepackage{graphicx}
\RequirePackage[dvipsnames]{xcolor}
\RequirePackage{natbib}
\RequirePackage[colorlinks,citecolor=blue,urlcolor=blue]{hyperref}
\usepackage{url} % not crucial - just used below for the URL 

\usepackage{xr}
\makeatletter

\newcommand*{\addFileDependency}[1]{% argument=file name and extension
\typeout{(#1)}% latexmk will find this if $recorder=0
\@addtofilelist{#1}
\IfFileExists{#1}{}{\typeout{No file #1.}}
}\makeatother

\newcommand*{\myexternaldocument}[1]{%
\externaldocument{#1}%
\addFileDependency{#1.tex}%
\addFileDependency{#1.aux}%
}
\myexternaldocument{supplement_JASA_Oct2_2025}

\newcommand{\blind}{1}

\allowdisplaybreaks

\theoremstyle{plain}

\newtheorem{theorem}{Theorem}

\newaliascnt{corollary}{theorem}
\newtheorem{corollary}[corollary]{Corollary}
\aliascntresetthe{corollary}

\newaliascnt{lemma}{theorem}
\newtheorem{lemma}[lemma]{Lemma}
\aliascntresetthe{lemma}

\newaliascnt{proposition}{theorem}

\aliascntresetthe{proposition}

\theoremstyle{definition}
\newaliascnt{definition}{theorem}
\newtheorem{definition}[definition]{Definition}
\aliascntresetthe{definition}

\theoremstyle{remark}
\newaliascnt{remark}{theorem}
\newtheorem{remark}[remark]{Remark}
\aliascntresetthe{remark}

\theoremstyle{plain}
\newtheorem{stheorem}{Theorem}

\newaliascnt{slemma}{stheorem}
\newtheorem{slemma}[slemma]{Lemma}
\aliascntresetthe{slemma}

\newaliascnt{scorollary}{stheorem}

\aliascntresetthe{scorollary}

\renewcommand{\P}{\text{P}}
\newcommand{\E}{\text{E}}
\newcommand{\Cov}{\text{Cov}}
\newcommand{\Var}{\text{Var}}
\newcommand{\bSigma}{\boldsymbol{\Sigma}}
\newcommand{\bmu}{\boldsymbol{\mu}}
\newcommand{\bdelta}{\boldsymbol{\delta}}
\newcommand{\bGamma}{\boldsymbol{\Gamma}}
\newcommand{\tr}{\mathrm{tr}}

\newcommand{\bS}{\mathbf{S}}
\newcommand{\bT}{\mathbf{T}}
\newcommand{\bU}{\mathbf{U}}

\newcommand{\bW}{\mathbf{W}}
\newcommand{\bX}{\mathbf{X}}
\newcommand{\bY}{\mathbf{Y}}
\newcommand{\bZ}{\mathbf{Z}}
\newcommand{\be}{\mathbf{e}}
\newcommand{\bbf}{\mathbf{f}} % \bf is already defined (boldface command)
\newcommand{\bg}{\mathbf{g}}
\newcommand{\bh}{\mathbf{h}}

\newcommand{\bt}{\mathbf{t}}
\newcommand{\bu}{\mathbf{u}}
\newcommand{\bv}{\mathbf{v}}
\newcommand{\bw}{\mathbf{w}}
\newcommand{\bx}{\mathbf{x}}
\newcommand{\by}{\mathbf{y}}
\newcommand{\bz}{\mathbf{z}}
\newcommand{\bA}{\mathbf{A}}
\newcommand{\bB}{\mathbf{B}}

\newcommand{\ba}{\mathbf{a}}
\newcommand{\bb}{\mathbf{b}}
\newcommand{\bc}{\mathbf{c}}

\usepackage{xcolor}

\newcounter{supplement}  
\makeatletter
\@addtoreset{equation}{supplement}
\@addtoreset{figure}{supplement}
\@addtoreset{table}{supplement}
\@addtoreset{section}{supplement}
\makeatother

\begin{document}

\def\spacingset#1{\renewcommand{\baselinestretch}%
{#1}\small\normalsize} \spacingset{1}

%%%%%%%%%%%%%%%%%%%%%%%%%%%%%%%%%%%%%%%%%%%%%%%%%%%%%%%%%%%%%%%%%%%%%%%%%%%%%%

\if1\blind
{
  \title{\bf Uniform-over-dimension location tests for multivariate and high-dimensional data}
  \author{
    Ritabrata Karmakar \\
    Indian Statistical Institute\\
    and \\
    {Joydeep Chowdhury}
    %\thanks{
    %The authors gratefully acknowledge \textit{please remember to list all relevant funding sources in the unblinded version}}\hspace{.2cm}\\
    \\
    King Abdullah University of Science and Technology\\
    and \\
    Subhajit Dutta\\
    Indian Statistical Institute and Indian Institute of Technology\\
    and \\
    Marc G. Genton\\
    King Abdullah University of Science and Technology}
  \maketitle
} \fi

\if 0 \blind
{
  \bigskip
  \bigskip
  \bigskip
  %\begin{center}
    \title{\bf Uniform-over-dimension location tests for multivariate and high-dimensional data}
  %\end{center}
  \author{}
  \maketitle
  \medskip
} \fi

\bigskip
\begin{abstract}
Asymptotic methods for hypothesis testing in high-dimensional data usually require the dimension of the observations to increase to infinity, often with an additional relationship between the dimension (say, $p$) and the sample size (say, $n$). On the other hand, multivariate asymptotic testing methods are valid for fixed dimension only and their implementations typically require the sample size to be large compared to the dimension to yield desirable results. In practical scenarios, it is usually not possible to determine whether the dimension of the data conform to the conditions required for the validity of the high-dimensional asymptotic methods for hypothesis testing, or whether the sample size is large enough compared to the dimension of the data. 
In this work, we first describe the notion of uniform-over-$p$ convergences and subsequently, develop a uniform-over-dimension central limit theorem.
An asymptotic test for the two-sample equality of locations is developed, which now holds uniformly over the dimension of the observations.
%This methodology attempts to unify the asymptotic results for fixed-dimensional multivariate data and high-dimensional data, and accounts for the effect of the dimension of the data on the performance of the hypothesis testing procedures.
%This hypothesis testing methodology thus can be applied to data of any dimension without concern. 
%We apply this  to the two-sample test for the equality of locations.
%is chosen to demonstrate this methodology.
%The test statistic proposed is unscaled by the sample covariance, similar to the usual tests for high-dimensional data. 
Using simulated and real data, it is demonstrated that the proposed test exhibits better performance compared to several popular tests in the literature for high-dimensional data as well as the usual scaled two-sample tests for multivariate data, including the Hotelling's $T^2$ test for multivariate Gaussian data. 
%Further, it is demonstrated in simulated models that the proposed test performs better than the usual scaled two-sample tests for multivariate data, including the Hotelling's $T^2$ test for multivariate Gaussian data.
\end{abstract}

\noindent%
\vfill

\newpage
\spacingset{1.9} % DON'T change the spacing!
\section{Introduction}
\label{introduction}

In diverse studies and scientific experiments, observations generated are high-dimensional in nature, i.e., the dimension of the observation vectors are higher than the number of observations (see, e.g., \cite{carvalho2008high,wang2008approaches,lange2014next,buhlmann2014high}).
There is an extensive literature on high-dimensional one-sample and two-sample tests of means, or more generally, of the locations/centers of the underlying distributions (see \cite{goeman2006testing,chen2010two,srivastava2013two,biswas2014nonparametric,cai2014two,wang2015high,javanmard2014confidence,chakraborty2017tests} and the recent review paper by \cite{MR4353841}).
In many asymptotic results involving high-dimensional data, authors make the assumption that the dimension $p \to \infty$ along with sample size $n \to \infty$ (see \cite{bai1996effect,chen2010two,zhang2020simple}). Otherwise, there are often conditions on the rate of growth of $p$, e.g., $p = o( n^\alpha )$ for some $\alpha>0$ (see \cite{hu2020pairwise}). However, these conditions on the dimension are not easy to verify, and it is often unclear whether they are satisfied in a particular situation or not. In a practical situation of hypothesis testing, an experimenter has only one sample of observations, which fixes the value of $n$ and~$p$. Based on just this one pair of $n$ and $p$, it is not possible to verify whether $p$ satisfies the particular rate with respect to $n$ required for the validity of those aforementioned results.

Multivariate asymptotic results are well-established in the literature for an extensive collection of statistics. These results are based on the assumption of a fixed value of $p$ and $n$ tends to infinity. However, in practice, except when $p$ is very small compared to $n$, it is sometimes observed that existing asymptotic results provide unsatisfactory approximations to the actual distributions of the test statistics for moderate values of $p$, which are still far lower compared to $n$.

In this work, we propose an asymptotic method of hypothesis testing, which is applicable uniformly over the dimension $p$.	
This eliminates one's concern over the validity of the asymptotic conditions on the dimension, which are generally found in the existing results.
Thus, this testing method can be applied to data with arbitrary dimensions, including usual multivariate as well as high-dimensional data.

In location tests, the statistic is typically univariate (regardless of the data dimension $p$). This motivates us to consider $d$-dimensional random vectors $\bX_{n,p} = f_p(\bZ_{1,p},\ldots,\bZ_{n,p})$, where $f_p$ maps the $p$-dimensional sample $\bZ_{1,p},\ldots,\bZ_{n,p}$ into $\mathbb{R}^d$ with $d$ a fixed positive integer. 
Here, $d = 1$ for the test statistic. 
However, we allow general values of $d$ and develop the notion of uniform-over-$p$ convergence of $\bX_{n,p}$ to some collection $\bX_p$, thereby, extending the classical weak convergence results such as the L\'{e}vy's continuity theorem and the Lindeberg-Feller central limit theorem (see, e.g., \cite{van1996weak,billingsley2013convergence,bogachev2018weak}) 
to this~setting.

%{\color{red}
The utility of these theoretical results is demonstrated on a two-sample test for equality of locations. Our test statistic is constructed using a kernel, without normalization by the sample covariance matrix (which is quite common for high-dimensional tests). Intuitively, the proposed test statistic stabilizes across different dimensions (once divided by a suitable scaling factor) and leads us to uniform-over-dimension convergence. In fact, this scaling factor does not need to be estimated in our implementations. When the dimension is treated to be fixed, the large sample distribution of this test statistic is a weighted sum of centered $\chi^2_1$ random variables. This further motivates the use of a parametric bootstrap procedure to estimate the cut-off of the unscaled test statistic from this large sample distribution, with natural data driven estimators of the weights. Under a set of regularity conditions, we show that the resulting testing procedure remains valid and consistent, uniformly across all dimensions as the sample size tends to infinity. 
%}

\begin{comment}
Suppose one carries out a test of hypothesis based on $n$ independent $p$-dimensional observations $\bZ_{1, p}, \ldots, \bZ_{n, p}$. The test statistic in any such hypothesis test, $S_{n,p} = g_p( \bZ_{1, p}, \ldots, \bZ_{n, p} )$ is always univariate, whatever the value of the dimension $p$ of the data might be.
This observation motivates the setup we consider for developing a uniform-over-dimension asymptotic theory of convergence for our hypothesis testing method. The ideas of uniform convergence of stochastic processes are well-studied (see, e.g., \cite{van1996weak,billingsley2013convergence,bogachev2018weak}).

We consider functions $\bX_p = f_p( \bZ_{n, p} )$ of $p$-dimensional random vectors $\bZ_{n, p}$, where $f_p : \mathbb{R}^p \to \mathbb{R}^d$ with $d$ being a fixed positive integer, and describe the uniform-over-$p$ convergence of $\bX_p$. For a test statistic, $d$ is 1, but we consider general integer values of $d$ while developing the theory, which helps in deriving subsequent results. 
%Often, a test statistic can be decomposed in two components, one of which is asymptotically negligible, and it is easier to derive the asymptotic distribution of the other component. 
To derive %implement similar techniques in deriving 
uniform-over-dimension asymptotic distributions, we first describe uniform-over-dimension convergence in distribution of functions of random vectors.
%, and state associated results concerning both uniform-over-dimension convergence in distribution and convergence in probability.
\end{comment}

A testing method was proposed in \cite{kim2024dimension}, where the authors used sample splitting and self-normalization to construct a modified version of a one-sample degenerate U-statistic. Its limiting null distribution is univariate Gaussian, regardless of the data dimension. Their dimension-agnostic convergence result is in fact a special case of our uniform-over-dimension convergence. In our framework, the limiting distribution may depend on the data dimension, making it more general. While their technique was used for one-sample mean and covariance testing problems, our theory is applied to construct a kernel-based two-sample location test that avoids data splitting.
Related work in this area includes the work on universal asymptotics for high-dimensional sign tests by \cite{paindaveine2016high}. Under a mild continuity assumption, tests based on center-outward ranks (see, e.g., \cite{MR4255122}) also have the dimension-agnostic property.

%Add some citations!
%A related testing method was proposed in \cite{kim2024dimension}, where the authors employed sample splitting and self-normalization to construct a modified version of a one-sample degenerate U-statistics whose limiting null distribution is univariate Gaussian, regardless of the data dimension. The dimension-agnostic convergence result in the paper by \cite{kim2024dimension} is a special case of our uniform-over-dimension convergence. In our framework, the limiting distribution may depend on the data dimension, making it more general. Their technique was used for one-sample mean and covariance testing problems. %, achieving minimax rate-optimal power under appropriate local alternatives. 
%However, our theory is applied to two-sample location testing problems based on kernel functions (see Section \ref{sec:multisample}), using a procedure that differs substantially from theirs as it avoids data splitting.

%The proposed methodology is demonstrated on a test for equality of locations and deriving the uniform-over-dimension asymptotic null distribution of the test statistic. Our test statistic is constructed based on a kernel, without using normalization by the sample covariance matrix. Normalization by the sample covariance is usually not used for high-dimensional tests, and we have followed a similar principle here. However, we propose the test to be applied without regard to the dimension of the data, whether high-dimensional, or standard multivariate situations. 

A natural question arises about the performance of the proposed testing method for multivariate data, without normalization by the sample covariance.
%, where normalization by sample covariance is usually employed. 
In our numerical work, we demonstrate that the proposed test outperforms the other tests both for low as well as high-dimensional data.
%, in Gaussian as well as non-Gaussian~models.
Our comparative study further highlights how the proposed test serves as a unified framework for two-sample location testing. It maintains appropriate empirical size and exhibits high power irrespective of whether the underlying data are Gaussian or heavy-tailed, or even high-dimensional. 
%Unlike traditional multivariate tests, it does not rely on the invertibility of the sample covariance matrix, and unlike many high-dimensional procedures, it remains valid and effective when $p$ is small. 
To summarize, our proposed test emerges as a \textit{dimension-agnostic}, \textit{distributionally robust} and \textit{computationally fast} alternative to existing approaches for testing equality of locations.

%Intuitively speaking, our proposed test statistic stabilizes across different dimensions once it is divided by a suitable scaling factor, which leads to a uniform-over-dimension convergence. In practice, the scaling factor does not need to be estimated. When the dimension is treated as fixed, the large sample distribution of this test statistic is a weighted sum of centered $\chi^2_1$ random variables. This gives us a motivation to use a parametric bootstrap procedure to estimate the cut-off of the unscaled test statistic from this large sample distribution with suitable data driven estimators of the weights. Under mild regularity conditions, we show that the resulting testing procedure continues to be valid and consistent, uniformly over all dimensions as the sample size grows to infinity. 
%We now give a structure of our paper.

The rest of this paper is organized as follows.
In Section \ref{sec:theory}, definitions and theorems related to uniform-over-dimension convergence in distribution for functions of random vectors are presented. In Section \ref{sec:twosample}, the proposed theory is employed to derive the asymptotic null distribution of a suggested test of equality of locations (which is valid uniformly over the dimension of the observations) and a parametric bootstrap implementation is described. Further, the asymptotic consistency of the test is also established (uniformly over the dimension of the observations). Using simulated and real data, it is demonstrated in Section~\ref{sec:data_analysis} that the proposed test equipped with its asymptotic null distribution is uniformly valid over the data dimension and outperforms other tests available for both usual multivariate as well as high-dimensional data. In Section \ref{sec:conclusion}, further work on the proposed theory and potential applications are discussed. Proofs of the mathematical results are divided between an Appendix and a Supplementary.

%\vspace*{-0.1in}
\section{Uniform-over-dimension theoretical results} \label{sec:theory}

This section introduces uniform-over-dimension convergence in distribution and in~probability, together with associated results like uniform-over-$p$ L\'{e}vy's continuity theorem and central limit theorem (CLT). All proofs of this section are deferred to the~Supplementary.

Fix the indices $n, p \in \mathbb{N}$. Let $\{ \bX_{n, p} \}$ be $d$-dimensional random vectors with probability measures $\{\mu_{n, p}\}$ and associated distribution functions $\{F_{n, p}\}$ on $(\mathbb{R}^d, \mathcal{R}^d)$, where $\mathcal{R}^d$ is the Borel sigma field on $\mathbb{R}^d$. Also, let $\bX_p$ be a $d$-dimensional random vector with probability measure $\mu_p$ and associated distribution function $F_p$ on $(\mathbb{R}^d, \mathcal{R}^d)$ for $p\in\mathbb{N}$.

\begin{definition}\label{definition1}
We say that 
%$\bX_{n, p} \stackrel{D}{\longrightarrow} \bX_p$}, 
$\bX_{n, p} \Longrightarrow \bX_p$, 
$\mu_{n, p} \Longrightarrow \mu_p$ or $F_{n, p} \Longrightarrow F_p$ uniformly-over-$p$ if for every bounded continuous function $f : \mathbb{R}^d \to \mathbb{R}$, we have
\begin{align*}
\lim\limits_{n \to \infty} \sup_{ p\in \mathbb{N}} \left| \int f \mathrm{d} \mu_{n, p} - \int f \mathrm{d} \mu_p \right| = 0 .
\end{align*}
%We write it as either of the following:.
\end{definition}

%Similarly, we define the uniform-over-$p$ convergence in probability.

\begin{definition}\label{definition2}
We say that $\bX_{n, p} \stackrel{\P}{\longrightarrow} \bX_{p}$ uniformly-over-$p$ if for every $\epsilon > 0$, we have
\begin{align*}
\lim\limits_{n \to \infty} \sup_{p \in \mathbb{N}} \P\left[ \left\| \bX_{n, p} - \bX_{p} \right\| > \epsilon \right] = 0 .
\end{align*}
\end{definition}
Using these notions, we now develop several results that parallel classical theorems for weak convergence of probability measures (see, e.g., \cite{billingsley2013convergence}). The uniform-over-$p$ version of the classical continuous mapping theorem (see \autoref{mappingthm} in the Supplementary) follows immediately from this definition. To establish other results, we impose two assumptions involving tightness and equi-continuity (see, e.g., \cite{billingsley2013convergence}). 

We first recall the definition of equi-continuity. The collection of distribution functions (dfs) $\{F_p\}_{p\in \mathbb{N}}$ is said to be equi-continuous at $\bx \in \mathbb{R}^d$ if for any $\epsilon > 0$ there exists $\delta_\bx > 0$ such that whenever $\|\by - \bx\| < \delta_\bx$, we have $\sup_p|F_p(\by) - F_p(\bx)| < \epsilon$. Now, we state the assumptions.
\begin{enumerate}[label = (A\arabic*), ref = (A\arabic*)]
\item The collection $\{F_p\}_{p \in \mathbb{N}}$ of probability distributions is tight. \label{assumption1}
\item The set of all equi-continuity points of the collection $\{F_p\}_{p \in \mathbb{N}}$ denoted by
$$C = \{\bx\in \mathbb{R}^d : \bx \,\, \text{is a point of equi-continuity for}\,\, \{F_p\}_{p \in \mathbb{N}}\}$$
is co-countable, i.e., is the complement of a countable set. 
\label{assumption2}
\end{enumerate}
Assumption \ref{assumption1} is standard in the weak convergence literature, while assumption \ref{assumption2} is a mild regularity condition. Under these assumptions, we extend both the Portmanteau theorem (see \autoref{portmanteau}) and Slutsky’s theorem (see \autoref{lemmaSlutsky} and \autoref{Slutsky}) to the uniform-over-$p$ setting. Complete statements and proofs are provided in the Supplementary.

We now state the uniform-over-$p$ L\'{e}vy's continuity theorem, which plays a central role in establishing the asymptotic results of the next section.

\begin{theorem}[Uniform-over-$p$ L\'{e}vy's continuity theorem]\label{levy}
Let $\{\mu_{n, p}\}$ and $\{ \mu_p \}$ be probability measures on $\mathbb{R}^d$ with characteristic functions $\{ \varphi_{n, p} \}$ and $\{ \varphi_p \}$ for $n\in\mathbb{N}$ and $p\in\mathbb{N}$. Under assumptions \ref{assumption1} and \ref{assumption2}, $ \mu_{n, p} \Longrightarrow \mu_p $ uniformly-over-$p$ if and only if
\begin{align*}
\lim\limits_{n \to \infty} \sup_p | \varphi_{n, p}( \bt ) - \varphi_p( \bt ) | = 0 \text{ for every } \bt \in \mathbb{R}^d.
\end{align*}
\end{theorem}

In our theoretical results of Section \ref{sec:twosample}, we will only consider a univariate test statistic (i.e., $d = 1$). 
Recall that the quantile function $Q: (0,1) \to \mathbb{R}$ for a df $F$ on $\mathbb{R}$ is defined~as $Q(\alpha) = \inf \{x \in \mathbb{R}: F(x) \geq \alpha\} \text{ for } \alpha \in (0,1)$. 
The following theorem establishes the uniform-over-$p$ convergence of the quantiles of the test statistic. Further, it also provides a sufficient condition under which assumptions \ref{assumption1} and \ref{assumption2} hold in this setting. 

\begin{theorem}\label{lemma_quantile}
Let $\{F_{n,p}\}_{n, p\in\mathbb{N}}$ be a family of dfs and $\{F_p\}_{p \in \mathbb{N}}$ and $F$ be continuous dfs on $\mathbb{R}$. Suppose that $F_p \stackrel{D}{\longrightarrow} F$ as $p \to \infty$. Then, $\{F_p\}_{p\in\mathbb{N}}$ satisfies assumptions \ref{assumption1} and \ref{assumption2}. Moreover, if $F_{n,p} \Longrightarrow F_p$ uniformly-over-$p$, then
$$\lim_{n \to \infty} \sup_{p \in \mathbb{N}} \sup_{x \in \mathbb{R}}|F_{n,p}(x) - F_p(x)| = 0.$$
Furthermore, let $Q_{n,p}, Q_p$ and $Q$ denote the quantile functions of $F_{n,p}, F_p$ and $F$, respectively. For a given $\alpha \in (0, 1)$, assume that  $F_p(Q_p(\alpha)) < F_p(Q_p(\alpha) + \epsilon)$ and $F(Q(\alpha)) < F(Q(\alpha) + \epsilon)$ for any $\epsilon > 0$ and $p \in \mathbb{N}$. Then, we have
$$\lim_{n \to \infty} \sup_{p \in \mathbb{N}} |Q_{n,p}(\alpha) - Q_p(\alpha)| = 0.$$
\end{theorem}

\begin{remark}
The conditions ensure that $Q_p$ for $p \in \mathbb{N}$ and $Q$ are continuous at the point~$\alpha$.
\end{remark}

\begin{remark}
This theorem provides a set of sufficient conditions for assumptions \ref{assumption1} and~\ref{assumption2}.
\end{remark}

Finally, we state a uniform-over-$p$ version of the classical Lindeberg-Feller central limit theorem (CLT).
To establish this result, the uniform-over-$p$ L\'{e}vy's continuity theorem is a useful tool to establish the uniform-over-$p$ convergence in distribution for a sequence of random variables which are functions of some $p$-variate random vectors.

\begin{theorem}[Uniform-over-$p$ CLT]\label{lindeberg}
Let $\{k_n\}_{n \in \mathbb{N}}$ be a sequence of natural numbers such that $k_n \to \infty$ as $n\to \infty$. Let $\bX_{n,p,1},\ldots,\bX_{n,p, k_n}$ be independent $d$-dimensional random vectors with mean zero and $\E[\|\bX_{n,p,r}\|^2] < \infty$ for $1 \leq r \leq k_n$ and $n, p\in \mathbb{N}$. Define $\bS_{n,p} = \bX_{n,p,1}+\cdots+\bX_{n,p,k_n}$. Let $\{\be_{i}: i = 1,\ldots, d\}$ denote the canonical basis of $\mathbb{R}^d$. Assume that there exists a matrix $\bSigma_p = (\sigma_{p,ij})_{d \times d}$ for each $p \in \mathbb{N}$ such that the following~holds:
\begin{enumerate}[label = (\alph*), ref = (\alph*)]
    \item $\lim_{n\to\infty}\sup_{p \in \mathbb{N}}\left|\be_i^\top\Cov(\bS_{n,p})\be_j - \sigma_{p,ij}\right| = 0$ for all $1\leq i,j\leq d$. 
    \item $\sup_{p \in \mathbb{N}} \sigma_{p,ii} < \infty$ for $1\leq i\leq d$.
    \item For any $\epsilon > 0$ and $1 \leq i \leq d$,
    $$\lim_{n\to\infty} \sup_{p \in \mathbb{N}} \sum_{r=1}^{k_n}\E\left[(\be_i^\top\bX_{n,p,r})^2 \mathbb{I}{\{|\be_i^\top\bX_{n,p,r}| > \epsilon\}}\right] = 0.$$
\end{enumerate}
Then, $\bS_{n,p} \Longrightarrow N_d(\mathbf{0}_d, \bSigma_p)$ uniformly-over-$p$ as $n \to \infty$.
\end{theorem}

\begin{corollary}[Uniform-over-$p$ Lyapunov's condition]\label{lyapunov}
The assertion in \autoref{lindeberg} continues to hold if condition (c) is modified as follows. There exists $\delta > 0$ such that 
$$\lim_{n\to\infty} \sup_{p \in \mathbb{N}} \sum_{r=1}^{k_n} \E\left[(\be_i^\top\bX_{n,p,r})^{2+\delta}\right] = 0.$$
\end{corollary}

\noindent
For example, consider a sequence of independent and identically distributed (i.i.d.) mean zero random vectors $\{\bX_{p, r}\}_{r \in \mathbb{N}}$ with $\Cov(\bX_{p, r}) = \bSigma_p$, where $\bSigma_p$ satisfies conditions (a) and (b) of \autoref{lindeberg}. Condition (c) is then automatically satisfied and $\sqrt{n}\bar{\bX}_p = \frac{1}{\sqrt{n}} \sum_{r=1}^n \bX_{p, r} \Longrightarrow N_d(\mathbf{0}_d, \bSigma_p)$ uniformly-over-$p$ as $n \to \infty$.

The results here are parallel to those in the unpublished short note of \cite{bengs2019uniform} and \cite{kasy2019uniformity}. The former requires uniform absolute continuity, whereas we impose weaker as well as more widely applicable conditions. On the other hand, \cite{kasy2019uniformity} focuses on the uniform validity of the delta method (which differs in scope from our work). Although they prove a uniform CLT, our result in \autoref{lindeberg} is more~general.

%The results in this section are parallel to those in the unpublished short note of \cite{bengs2019uniform} and \cite{kasy2019uniformity}, though based on different sets of assumptions. In \cite{bengs2019uniform}, the authors obtained similar results to ours under the stronger condition of uniform absolute continuity, which is difficult to verify in practice. In contrast, we do not impose such strong assumptions, making our results more widely applicable. \cite{kasy2019uniformity} focused on the uniform validity of the delta method, leading to application areas considerably different from ours. While they also proved a uniform central limit theorem, our result in \autoref{lindeberg} is more general.

%\vspace*{-0.1in}
\section{Two-sample location tests}\label{sec:twosample}

%In this section, we demonstrate how the preceding theory of convergence helps us derive asymptotic distributions of test statistics independent of any assumption on the dimension.

Let $\bX_{p,1}, \ldots, \bX_{p,n_1}$ and $\bY_{p,1}, \ldots, \bY_{p,n_2}$ be two i.i.d. samples of sizes $n_1$ and $n_2$ drawn from two probability distributions on $\mathbb{R}^p$ with location parameters ${\bmu}_{p,1}$ and $\bmu_{p,2}$ for $p \in \mathbb{N}$. 
We are interested in testing the following hypothesis
\begin{align} \label{h_0}
\text{H}_0 : \; \bmu_{p, 1} = \bmu_{p, 2} \quad\text{versus} \quad
\text{H}_1 : \; \bmu_{p, 1} \neq \bmu_{p, 2}.
\end{align}
Let $\bh_p( \cdot, \cdot ) : \mathbb{R}^p \times \mathbb{R}^p \to \mathbb{R}^p$ be a kernel function and define $\bdelta_p = \E[\bh_p( \bX_{p,1}, \bY_{p,1} )]$, which satisfies $\bdelta_p = \mathbf{0}_p$ under $\text{H}_0$. When the location parameters are the population means, a natural choice is $\bh_p( \bx, \by) = \bx - \by$. 
If $\bX_{p,1} \stackrel{D}{=} \bmu_{p,1} + \bW_p$ and $\bY_{p,1} \stackrel{D}{=} \bmu_{p,2} + \bW_p$ for some random vector $\bW_p$ (with possibly undefined moments), then $\bh_p( \bx, \by ) = ( \bx - \by ) / \| \bx - \by \|$ satisfies $\bdelta_p = \mathbf{0}_p$ under $\text{H}_0$. This is so because both $\bX_{p,1}$ and $\bY_{p,1}$ are identically distributed under $\text{H}_0$ and hence, we have $\E\left[\bh(\bX_{p,1}, \bY_{p,1})\right] = \E\left[\bh(\bY_{p,1}, \bX_{p,1})\right] = - \E\left[\bh(\bX_{p,1}, \bY_{p,1})\right]$. Clearly, this now implies that $\bdelta_p = \mathbf{0}_p$. 
Moreover, if $p \geq 2$ and the support of $\bW_p$ is not contained in a straight line, it follows from the proof of Theorem 2.6 in \cite{chowdhury2022multi} that $\|\bdelta_p\| > 0$ under $\text{H}_1$. Consider the following test statistic (with $n=n_1+n_2$):
\begin{multline*}
T_{n,p} = \frac{1}{nn_1n_2} \left[ \left\|\sum_{i=1}^{n_1} \sum_{j=1}^{n_2}\bh_p(\bX_{p,i}, \bY_{p,j})\right\|^2  - \sum_{i=1}^{n_1}\left\|\sum_{j=1}^{n_2}\bh_p(\bX_{p,i}, \bY_{p,j})\right\|^2\right.\\
\left.- \sum_{j=1}^{n_2}\left\|\sum_{i=1}^{n_1}\bh_p(\bX_{p,i}, \bY_{p,j})\right\|^2 + \sum_{i=1}^{n_1}\sum_{j=1}^{n_2}\left\|\bh_p(\bX_{p,i}, \bY_{p,j})\right\|^2  \right].  
\end{multline*}
Now, the null is to be rejected if the value of $T_{n,p}$ is large. An equivalent representation is $T_{n,p}= \frac{1}{n n_1 n_2} \sum_{i_1 \neq i_2} \sum_{j_1 \neq j_2} \bh_p(\bX_{p,i_1}, \bY_{p,j_1})^\top \bh_p(\bX_{p,i_2}, \bY_{p,j_2})$. When $\bh_p( \bx, \by ) = \bx - \by$, the resulting statistic is similar to that of \cite{chen2010two}. Related variants appear in \cite{zhang2020simple, zhang2021hetero}, where a Welch–Satterthwaite approximation was used for the implementation of the test.
%, but our procedure differs substantially from both. 
When $\bh_p( \bx, \by ) = ( \bx - \by ) / \| \bx - \by \|$, a closely related test was proposed by \cite{chowdhury2022multi}, but their observations are from a separable Hilbert space. Another spatial sign-based test was proposed by \cite{feng2016multivariate}, where normalization was done via a diagonal matrix and the asymptotic theory requires conditions on the growth of $p$ relative to $n$. For the derivation of the asymptotic distribution of $T_{n,p}$, we first introduce its centered version as follows:
$$T_{n,p, 0} = \frac{1}{n n_1 n_2} \sum_{i_1 \neq i_2} \sum_{j_1 \neq j_2} [\bh_p(\bX_{p,i_1}, \bY_{p,j_1}) -\bdelta_p]^\top [\bh_p(\bX_{p,i_2}, \bY_{p,j_2}) - \bdelta_p].$$
Assume the following standard condition.
\begin{enumerate}[label = (C\arabic*), ref = (C\arabic*)]
\vspace*{-0.125in}
\item As $n \to \infty$, we have $n_1/n \to \tau \in (0,1)$. \label{cond1}
\vspace*{-0.125in}
%where $n = n_1 + n_2$. 
\end{enumerate}
Let $\bSigma_{p,1} = \Var(\E[\bh_p(\bX_{p,1}, \bY_{p,1})\vert \bX_{p,1}])$ and $\bSigma_{p,2} = \Var(\E[\bh_p(\bX_{p,1}, \bY_{p,1})\vert \bY_{p,1}])$, and define $\bSigma_p = (1-\tau) \bSigma_{p,1} + \tau \bSigma_{p,2}$ for $p \in \mathbb{N}$. 
For any $1 \leq i \leq n_1$ and $1 \leq j \leq n_2$, define
\begin{align}
& \bar{\bh}_p( \bX_{p,i}, \bY_{p,j})
= \bh_p( \bX_{p,i}, \bY_{p,j}) 
- \E[\bh_p(\bX_{p,i},\bY_{p,j})\vert \bX_{p,i}]
- \E[\bh_p(\bX_{p,i}, \bY_{p,j}) \vert \bY_{p,j}] + \bdelta_p. \label{hbar}
\end{align}
We have $\E[\bar{\bh}_p( \bX_{p,i}, \bY_{p,j})] = \E[\bar{\bh}_p( \bX_{p,i}, \bY_{p,j}) \vert \bX_{p,i}] = \E[\bar{\bh}_p( \bX_{p,i}, \bY_{p,j}) \vert \bY_{p,j}] = \mathbf{0}_p$. Now, consider the following conditions. 
\begin{enumerate}[label = (C\arabic*), ref = (C\arabic*)]
\setcounter{enumi}{1}
\item $\displaystyle \sup_p\frac{\E[\|\bar{\bh}_p(\bX_{p,1}, \bY_{p,1})\|^4]}{\tr(\bSigma_p^2)}< \infty$ and  $\displaystyle \sup_p\frac{\E[\|\bar{\bh}_p(\bX_{p,i},\bY_{p,j})\|^4] [\tr(\bSigma_p)]^2}{[\tr(\bSigma_p^2)]^2} < \infty$. \label{cond2}
\item We have
$$\E[\bh_p(\bX_{p,i}, \bY_{p,j})\vert \bX_{p,i}] = \bdelta_p + \bGamma_{p,1} \mathbf{z}_{p,i1} \text{ for } i = 1, \ldots,n_1$$ and $$\E[\bh_p(\bX_{p,i}, \bY_{p,j})\vert \bY_{p,j}] = \bdelta_p + \bGamma_{p,2} \mathbf{z}_{p,j2} \text{ for } j = 1, \ldots,n_2,$$ 
where each $\bGamma_{p,j}$ is a $p \times m$ matrix such that $\bGamma_{p,j} \bGamma_{p,j}^\top = \bSigma_{p,j}$ and $\mathbf{z}_{p,ij}$ are i.i.d. $m$-dimensional vectors with $\E(\mathbf{z}_{p,ij}) = \mathbf{0}_m$ and $\Var(\mathbf{z}_{p,ij}) = \mathbf{I}_m$ (the $m \times m$ identity matrix), for $i = 1,\ldots, n_j$ and $j = 1,2$. Further, $\E(z_{p,ijk}^4) = 3 + \Delta$ with $z_{p,ijk}$ being the $k$-th component of $\mathbf{z}_{p,ij}$ and $\Delta$ is a fixed (independent of $p$) constant. Also, assume that
$$\E(z_{p,ijl_1}^{\alpha_1} \cdots z_{p,ijl_q}^{\alpha_q}) = \E(z_{p,ijl_1}^{\alpha_1}) \cdots \E(z_{p,ijl_q}^{\alpha_q})$$ 
for any positive integer $q$ such that $\sum_{\ell=1}^q \alpha_\ell \leq 4$ and $l_1 \ne \cdots \ne l_q$. \label{cond3}

\item Assume that $\mathrm{tr}(\bSigma_{p,i} \bSigma_{p,j} \bSigma_{p,l} \bSigma_{p,h}) = o\left\{ \mathrm{tr}^2\left[ (\bSigma_{p,1} + \bSigma_{p,2})^2 \right] \right\}$ as $p \to \infty$ for all $i,j,l,h = 1$ or $2$. Further, in \ref{cond3}, $\E(z_{p,ijl_1}^{\alpha_1} \cdots z_{p,ijl_q}^{\alpha_q}) = \E(z_{p,ijl_1}^{\alpha_1}) \cdots \E(z_{p,ijl_q}^{\alpha_q})$ holds for any positive integer $q$ such that $\sum_{\ell=1}^q \alpha_\ell \leq 8$ and $l_1 \ne \cdots \ne l_q$. \label{cond4}

\item 
%Define $\bSigma_p = (1-\tau) \bSigma_{p,1} + \tau \bSigma_{p,2}$ and 
Let $\lambda_{p,r}$ denote the eigenvalues of $\bSigma_p$ in decreasing order for $r = 1,\ldots,p$. Define
$\rho_{p,r} =  \lambda_{p,r} / \sqrt{\lambda_{p,1}^2 +\cdots+ \lambda_{p,p}^2} $
for all $r$. Assume that $\rho_{p,r} \to \rho_r$ as $p \to \infty$. Further, assume $\sup_{p > q} \sum_{r=q+1}^p \rho_{p,r}^2 \to 0$ and $\sum_{r=q+1}^\infty \rho_{r}^2 \to 0$ as $q \to \infty$ and $\rho_1 > 0$. \label{cond5}
\end{enumerate}

Assumption \ref{cond2} ensures that the error terms arising from the Hoeffding-type decomposition of the test statistic are negligible for large $n$ uniformly-over-$p$. When $\bh_p(\bx, \by) = (\bx - \by) / \|\bx - \by\|$, this condition holds provided $\inf_p\tr(\bSigma_p) > 0$. Suppose $\bh_p(\bx, \by) = \bx - \by$. Then, \ref{cond2} is trivially true and $\bSigma_p = (1-\tau) \Var(\bX_{p,1}) + \tau \Var(\bY_{p,1})$. In this setting, conditions \ref{cond3} and \ref{cond4} were considered by \cite{chen2010two} to specify a factor model for high-dimensional data. Condition \ref{cond5} generalizes a similar assumption considered by \cite{zhang2020simple, zhang2021hetero} and now includes the equi-correlation matrix as a special case.
%case when $\bSigma_p$ is an
%(which is not addressed by such earlier work).  

We now describe the asymptotic distribution of the test statistic $T_{n,p}$. For each $p$, define $\zeta_p = (\sum_{r=1}^p \rho_{p,r}(W_r - 1)) / \sqrt{2}$, where $\{ W_r \}_{r \in \mathbb{N}}$ are i.i.d. $\chi_1^2$ random variates.

\begin{theorem}\label{thm:1}
Suppose that \ref{cond1}, \ref{cond2} and \ref{cond3} hold, and that either \ref{cond4}, or \ref{cond5} is satisfied. Then, the following holds:
\begin{enumerate}[label = (\alph*), ref = (C\alph*)]
\item $\zeta_p \stackrel{D}{\longrightarrow} \zeta$ as $p \to \infty$, where $\zeta \sim N(0,1)$ under \ref{cond4}, and $\zeta$ has the same distribution as $(\sum_{r=1}^\infty \rho_{r}(W_r - 1)) / \sqrt{2}$ under \ref{cond5}.
\item  As $n , p \to \infty$, we have
%{\color{red}$T_{n,p,0} / \sqrt{2\,\mathrm{tr}(\bSigma_{p}^2)}$ converges in distribution to $\zeta$, i.e,}
$$T_{n,p,0} / \sqrt{2\,\mathrm{tr}(\bSigma_{p}^2)} \stackrel{D}{\longrightarrow} \zeta.$$
\item  As $n \to \infty$
$$T_{n,p,0} / \sqrt{2\,\mathrm{tr}(\bSigma_{p}^2)} \Longrightarrow \zeta_p
\quad \text{uniformly-over-$p$}.$$
\end{enumerate}
\end{theorem}

\begin{remark}
The assumptions of \autoref{lemma_quantile} hold here since each of the random variables $\{\zeta_p\}_{p \in \mathbb{N}}$ and $\zeta$ have a strictly increasing and continuous dfs on the support.
\end{remark}

%{\color{red}
%Unlike \cite{zhang2020simple, zhang2021hetero}, part (b) holds without any assumptions on the relative growth rates of $n$ and $p$.
%which implicitly assume such a relationship, 
%Fix $q \in \mathbb{N}$. While bounding a quantity similar to $|\varphi_{n,p}^{(q)}(t) - \varphi_{\tilde{\zeta}_p^{(q)}}(t)|$ (defined in the proof, see p.30 of the Appendix for more details) by an arbitrary $\epsilon > 0$ (by using the classical CLT). For large $n$, the bound must hold uniformly for all $p \geq q$, which is achieved using uniform-over-$p$ CLT (see \autoref{lindeberg}).
%}
\begin{remark}
Unlike \cite{zhang2020simple, zhang2021hetero}, part (b) holds without any assumption on the relative growth rates of $n$ and $p$. Fix $q \in \mathbb{N}$ and an arbitrary $\epsilon > 0$. While bounding a quantity similar to $|\varphi_{n,p}^{(q)}(t) - \varphi_{\tilde{\zeta}_p^{(q)}}(t)|$ (see p. 28 of the Appendix for the definition) by $\epsilon$, the authors had used the classical CLT, but this cannot provide a bound that holds uniformly over $p \geq q$. We achieve this uniformity by using the uniform-over-$p$ CLT (see \autoref{lindeberg}).
\end{remark}

When $\bh_p( \bx, \by ) = \bx - \by$, parts (b) and (c) of \autoref{thm:1} extend the results given in \cite{zhang2020simple, zhang2021hetero} to the uniform-over-$p$ case and also cover a larger class of eigen structures of $\bSigma_p$ like the spiked covariance (see, e.g., \cite{johnstone2018}). 
The choice $\bh_p( \bx, \by ) = \bx - \by$ requires the existence of the mean and covariance of the random vectors. This may not be the case for heavy-tailed distributions, e.g., the multivariate Cauchy distribution. On the other hand, the choice $\bh_p( \bx, \by ) = ( \bx - \by ) / \| \bx - \by \|$ ensures that all moments exist because the kernel $\bh_p( \bx, \by )$ itself is bounded and more widely~applicable. 
%This choice can be more widely applied without regard to the existence of moments of the underlying random~vectors.

We now describe a procedure for implementing our test. By \autoref{thm:1}, the asymptotic distribution of $T_{n,p}$ is same as that of $\sum_{r=1}^p \lambda_{p,r}(W_r - 1)$ for a fixed $p$. Motivated by this key observation, we use a simulation based approach to estimate the cut-off under the null hypothesis. We first take an estimator $\hat\bSigma_p$ of $\bSigma_p$ and compute its eigenvalues $\hat\lambda_{n,p,i}$ for $i=1,\ldots,p$. Fix $\alpha \in (0,1)$. To implement a level-$\alpha$ test, we first generate i.i.d. standard normal observations $Z_{i,j}$ for $i = 1,\ldots,p$ and compute $V_{n,p,j} = \sum_{i=1}^p \hat\lambda_{n,p,i} (Z_{ij}^2 - 1)$ for $j= 1,\ldots, M$, where $M$ is a large positive integer. Let $\hat c_{n,p}(\alpha)$ be the $(1 - \alpha)$-th sample quantile of the values $V_{n,p,1}, \ldots,V_{n,p,M}$. Finally, we reject the null hypothesis if $T_{n,p} > \hat c_{n,p}(\alpha)$. 

We shall now focus on two different estimators of $\bSigma_p$. The first one is defined as
\begin{align*}
\hat\bSigma_{p}^{(1)} &= \frac{1}{nn_1n_2} \sum_{i=1}^{n_1} \left(\sum_{j=1}^{n_2} \bh_p(\bX_{p,i}, \bY_{p,j})\right)\left(\sum_{j=1}^{n_2} \bh_p(\bX_{p,i}, \bY_{p,j})\right)^\top \\
&\;\;\;\;\;+ \frac{1}{nn_1n_2} \sum_{j=1}^{n_2} \left(\sum_{i=1}^{n_1} \bh_p(\bX_{p,i}, \bY_{p,j})\right)\left(\sum_{i=1}^{n_1} \bh_p(\bX_{p,i}, \bY_{p,j})\right)^\top - \hat{\bdelta}_p \hat{\bdelta}_p^\top,
\end{align*}
where $\hat{\bdelta}_p = (\sum_{i=1}^{n_1} \sum_{j=1}^{n_2} \bh_p(\bX_{p,i}, \bY_{p,j})) / (n_1 n_2)$. Let $\hat\bSigma_p = (\hat\sigma_{ij})_{p \times p}$.
Motivated by covariance structures like the AR(1), a second estimator based on tapering is defined as follows:
$$\hat \bSigma_{p}^{(2)} = (w_{ij}\hat\sigma_{ij})_{p \times p}$$
with $w_{ij} = 1$ if $|i - j| < k / 2$, $w_{ij} =1 - |i-j| / k$ if $k / 2 < |i-j| <k$ and $w_{ij} = 0$ otherwise, for $k = \min\{n^{1 / (2\beta + 2)}, p\}$. This tapering procedure was considered by \cite{ttcai2010optimalcov} for banded covariance matrices.

We now study the asymptotic level and power of the proposed test statistic. Consider the following assumptions.
\begin{enumerate}[label = (C\arabic*), ref = (C\arabic*)]
\setcounter{enumi}{5}
\item $\sup_p \left[\tr(\bSigma_p) / \sqrt{\tr(\bSigma_p^2)}\,\right] < \infty$. \label{cond6}
\item Let $\bSigma_p = (\sigma_{ij})_{p \times p}$. There exists constants $C$ and $\beta > 0$ (independent of $p$) such that
$|\sigma_{ij}| \leq C \lambda_1 |i-j|^{-(\beta + 1)}$
and $\sup_p \left[p \lambda_1^2 / \tr(\bSigma_p^2)\right] < \infty$. \label{cond7}
\end{enumerate}
Assumptions \ref{cond5} and \ref{cond6} are satisfied when $\bSigma_p$ is an equi-correlation, a spiked correlation matrix, or when the eigenvalues of $\bSigma_p$ decay very rapidly. On the other hand, when $\bSigma_p$ is an identity matrix, the autocorrelation matrix of some MA($k$) or AR($1$) process, or has a banded structure, then conditions \ref{cond4} and \ref{cond7} are satisfied. The following lemma is necessary to derive the asymptotic level and power of our test.

\vspace*{-0.1in}
\begin{lemma}\label{lemma_gaussian_estimation}
Suppose that the assumptions of \autoref{thm:1} hold. In addition, assume \ref{cond6} when using $\hat\bSigma_p^{(1)}$ and \ref{cond7} when using $\hat\bSigma_{p}^{(2)}$. Then, 
$$V_{n,p,1} / \sqrt{2\tr(\bSigma_{p}^2)} \Longrightarrow \zeta_p
\quad \text{uniformly-over-$p$}.$$
\end{lemma}

Using this important lemma, we now establish the uniform-over-$p$ asymptotic validity and consistency of our proposed test.

\vspace*{-0.1in}
\begin{theorem} \label{thm:2}
\begin{enumerate}[label = (\alph*), ref = (C\alph*)]
\item If the assumptions of \autoref{lemma_gaussian_estimation} hold under $\text{H}_0$, then the probability of Type 1 error of the proposed level-$\alpha$ test converges to $\alpha$ uniformly-over-$p$ as $n,M\to \infty$. %, $M  \to \infty$.

\item Under $\text{H}_1$, suppose either 
\begin{enumerate}[label = (\roman*)]
    \item $\inf_p \|\bdelta_p\| > 0$ and $\sup_p \sup_{\bx} \sup_{\by} \|\bh_p(\bx, \by)\| < \infty$, or
    \item $\inf_p\left( \|\bdelta_p\|^2 / \sqrt{\tr(\bSigma_p^2)}\;\right) > 0$ and the assumptions of \autoref{lemma_gaussian_estimation} hold.
\end{enumerate}
Then, the power of the proposed level-$\alpha$ test converges to $1$ uniformly-over-$p$ as $n \to \infty$.
\end{enumerate}
\end{theorem}

%\begin{remark}
Here, we do not assume anything on the relative growth rates of $n$ and $p$, which is in line with \autoref{thm:1}. 
%( M ?)
%\end{remark}
When $\bh_p(\bx, \by) = \bx - \by$, the required conditions hold provided $\inf_p \|\bdelta_p\| / \sqrt{\tr(\bSigma_p^2)} > 0$ under $\text{H}_1$ and the assumptions of \autoref{lemma_gaussian_estimation} are satisfied under $\text{H}_0$.  
When $\bh_p(\bx, \by) = (\bx - \by) / \|\bx - \by\|$ the required conditions for part (b) hold as long as $\inf_p \|\bdelta_p\| > 0$ under $\text{H}_1$. 

%In the next section, we compare the performances of these two tests (which we obtain as particular cases of the general kernel-based test) with other popular tests for the usual multivariate as well as high-dimensional data.

%\newpage
\vspace*{-0.25in}
\section{Comparison of performance} \label{sec:data_analysis}

In this section, we study the performance of the proposed test through an extensive simulation study and a real data analysis.
%comparisons with several well-known multivariate as well as high-dimensional two-sample tests. 
Our goal is to examine the finite-sample behavior of our test in terms of the empirical size and power under diverse settings, including both Gaussian and heavy-tailed distributions across varying covariance structures as well as~dimensions. 

%The section is organized as follows. We first motivate the need for a unified testing framework that operates seamlessly across low and high dimensions. We then discuss the role of covariance normalization and its limitations in high-dimensional inference. Subsequently, we compare the SS test with classical multivariate and high-dimensional competitors, describe the simulation models used, and present a detailed discussion of empirical sizes and power. The section concludes with a summary of findings and implications.

%\subsection{Motivation: Unifying Multivariate and High-Dimensional Testing}
Traditional multivariate procedures such as the Hotelling’s $T^2$ test are optimal when the dimension $p$ is small relative to the sample size. However, this test is not well-defined when $p > n$ due to the singularity of the sample covariance matrix. To address this, a variety of high-dimensional two-sample tests have been proposed including those by \cite{bai1996effect}, \cite{chen2010two} and \cite{cai2014two}, but they often perform poorly when $p$ is small or moderate. A dimension agnostic testing procedure should retain validity and good power in both low and high-dimensional settings. Our proposed test is designed with this goal and we empirically evaluate its performance across both~regimes.

When $p \ll n$, we first examine the performance of the proposed test in comparison with two classical procedures: the Hotelling’s $T^2$ test (say, HT2) and the \cite{choi1997approach} test (say, CM1997), along with the test proposed by \cite{zhang2020simple} (say, ZGZC2020). The HT2 test is the most powerful invariant test under multivariate normality, while the CM1997 test extends rank-based inference to multivariate settings. These tests achieve scale invariance by normalizing with the estimated covariance matrix which is not possible when $p > n$ due to its singularity. Attempts to regularize the covariance matrix may affect the asymptotic property of the test. Our test circumvents this issue by avoiding any such normalization. While this forfeits strict scale invariance, it offers a simple recipe, especially, in high-dimensional settings. The key question that we explore here is whether our proposed test can retain comparable (in fact, superior for some cases) performance even in scenarios where the normalized methods are theoretically~optimal.

In the high-dimensional regime, the proposed test and ZGZC2020 test are evaluated alongside several widely used tests from the literature.
%\cite{bai1996effect}, \cite{cai2014two}, \cite{chen2010two}, \cite{chen2014two} and \cite{srivastava2008test}, and 
We denote them as BS1996 (\cite{bai1996effect}), CLX2014 (\cite{cai2014two}), CQ2010 (\cite{chen2010two}), CLZ2014 (\cite{chen2014two,chen2019two}) and SD2008 (\cite{srivastava2008test}) tests.
These methods are implemented using the \texttt{highmean} package in \texttt{R} (see \cite{highmean}). %(Lin and Pan, 2016). 
Each method handles the high-dimensional setting differently (e.g., via projections, regularized sample covariance, etc.). Our goal is to situate the proposed test in this landscape and evaluate whether its simplicity also translates into good empirical performance.

We consider two versions of the proposed test: one with the identity kernel $\bh_p(\bx, \by) = \bx - \by$ (say, KCDG2025) and one with the spatial sign kernel $\bh_p(\bx, \by) = (\bx - \by)/\|\bx - \by\|$ (say, sKCDG2025). Each version was applied with two choices of estimated covariance matrix: non-tapered (say, 1) and tapered with $\beta=0.25$ (say, 2). 
So, we have four versions of our method, namely, KCDG2025$^1$ and KCDG2025$^2$, and sKCDG2025$^1$ and sKCDG2025$^2$.
R codes for the proposed method are available from \href{https://github.com/ritakarma/uniform-over-dim-2sample-location/}{this Github link}.

\begin{comment}
In the high-dimensional regime, the KCDG2025 and ZGZC2020 tests are evaluated alongside several widely used tests from the literature.
%\cite{bai1996effect}, \cite{cai2014two}, \cite{chen2010two}, \cite{chen2014two} and \cite{srivastava2008test}, and 
We denote them as BS1996 (\cite{bai1996effect}), CLX2014 (\cite{cai2014two}), CQ2010 (\cite{chen2010two}), CLZ2014 (\cite{chen2014two,chen2019two}) and SD2008 (\cite{srivastava2008test}) tests.
These methods are implemented using the \texttt{highmean} package in \texttt{R} (see \cite{}). %(Lin and Pan, 2016). 
Each test adopts a different strategy for handling the high-dimensional setting like projections, regularized versions of the sample covariance matrix, etc. Our goal is to situate the proposed KCDG2025 test in this landscape and evaluate whether its simplicity translates into good empirical performance.
We use two versions of the proposed test with the identity kernel (say, KCDG2025) and the spatial sign kernel (say, sKCDG2025), along with two choices of the estimated covariance matrix, namely, non-banded (say, $1$) and banded with the choice $\beta=0.25$ (say, $2$). 
R codes of our methods are available from \href{https://www.dropbox.com/scl/fo/w4e3xg0obyj2azfzyaju7/AO01uTUN2icH61iSPekyHKU?rlkey=jtg35f57hi1jar9pt0mw8mx82&dl=0}{this Dropbox link}.
\end{comment}
We consider three models involving both Gaussian and non-Gaussian distributions:
\begin{itemize}
    \item \textit{Model 1:} $\nu_{p,0} \equiv N_p({\bf 0}_p, \bSigma_p)$, the $p$-dimensional Gaussian distribution;
    \item \textit{Model 2:} $\nu_{p,0} \equiv t_{p,4}({\bf 0}_p, \bSigma_p)$, the heavy-tailed multivariate $t$ distribution with 4 degrees of freedom;~and
    \item \textit{Model 3:} $\nu_{p,0} \equiv t_{p,1}({\bf 0}_p, \bSigma_p)$, the multivariate Cauchy distribution.
\end{itemize}

\noindent
For the covariance matrix $\bSigma_p = (\sigma_{ij})$, we take three representative forms: (i) $\sigma_{ij} = 0.5 + 0.5 \mathbb{I}(i=j)$, (ii) $\sigma_{ij} = 0 + 1 \mathbb{I}(i=j)$ and (iii) $\sigma_{ij} = 0.75^{|i-j|}$ for $1 \leq i, j \leq p$.
We set the sample sizes to $n_1 = 40$ and $n_2 = 50$, and consider dimensions $p = 5, 25, 50$ and $100$. For each configuration, we define ${\bf a}_{p,1} = {\bf 0}_p$ and ${\bf a}_{p,2} = \delta {\bf h}_p$ with ${\bf h}_p = (1, 2, \ldots, p)^\top / \| (1, 2, \ldots, p)^\top \|$. By varying $\delta$, we generate different degrees of separation between the two population means, allowing the estimation of empirical power curves as a function of $\delta$. When $\delta = 0$ (the null hypothesis is true), the corresponding proportions estimate the empirical sizes of the~tests.

\begin{comment}
The distributions in Models 1 and 2 are either identical or very similar to the distributions considered in Models 1 and 2 in subsection 3.1 in \cite{zhang2020simple}. The expression of the shifts are also the same, however, the values of $\delta$ considered and the values of the dimension $p$ and the sample sizes $n_1$ and $n_2$ are different. We choose these models because of their simplicity and convenience in comparison.
%of the performances of the ZGZC2020 test and our proposed KCDG2025 test.
%, along with several other tests from the literature. 
We further consider Model 3 to investigate the performances of the ZGZC2020 test and the other tests from the literature under significant departure from Gaussianity. In fact, the moments required for the theoretical validity of these tests fail to hold, while the proposed sKCDG2025 tests do not require the existence of the population moments.
\end{comment}

The distributions of Models 1 and 2 are either identical or very similar to those in Models 1 and 2 in subsection 3.1 of \cite{zhang2020simple}. The shifts are also expressed similarly, but the values of $\delta$, the dimension $p$ and the sample sizes $n_1$ and $n_2$ differ here. We choose these models because of their simplicity and convenience for comparisons.
%of the performances of the ZGZC2020 test and our proposed KCDG2025 test.
%, along with several other tests from the literature. 
We further consider Model 3 to investigate the performances of the ZGZC2020 and the other tests under significant departure from Gaussianity. In fact, the moments assumptions required for the theoretical validity of these tests fail to hold, while the proposed sKCDG2025 tests do not require the existence of such population moments.

\vspace*{-0.1in}
\subsection{Empirical analysis}
\label{sec:emp_analysis}

Table~\ref{tab:size_p=5} reports the estimated sizes for low-dimensional settings (with $p=5$), while Table~\ref{tab:size_p=100} presents results for the high-dimensional case (with $p=100$). Additional numerical results for $p=25$ and $p=50$ are given in Section \ref{sec:supplement_numerical} of the Supplementary.
%, each based on 1000 replications at the 5\% nominal level.
\begin{table}[b!]
\centering
\caption{Estimated sizes at nominal level $\alpha = 5\%$ for the different tests for multivariate data based on $10,000$ independent replications for $p=5$ with $n_1 = 40$ and $n_2 = 50$.}
\label{tab:size_p=5}

\vspace{0.1in}
\begin{tabular}{cccccccccc}
\hline
%Model & 1.i. & 1.ii. & 1.iii. & 2.i. & 2.ii. & 2.iii. & 3.i. & 3.ii. & 3.iii. \\ 
Model & 1.i. & 2.i. & 3.i. & 1.ii. & 2.ii. & 3.ii. & 1.iii. & 2.iii. & 3.iii. \\ \hline
KCDG2025$^1$ & 0.0482 & 0.0418 & 0.0088 & 0.0440 & 0.0432 & 0.0028 & 0.0542 & 0.0494 & 0.0128 \\
KCDG2025$^2$ & 0.0524 & 0.0460 & 0.0102 & 0.0446 & 0.0438 & 0.0032 & 0.0552 & 0.0518 & 0.0142 \\
sKCDG2025$^1$ & 0.0490 & 0.0474 & 0.0510 & 0.0448 & 0.0536 & 0.0464 & 0.0548 & 0.0482 & 0.0542 \\
sKCDG2025$^2$ & 0.0506 & 0.0484 & 0.0524 & 0.0474 & 0.0538 & 0.0478 & 0.0548 & 0.0478 & 0.0546 \\
ZGZC2020 & 0.0522 & 0.0464 & 0.0102 & 0.0518 & 0.0508 & 0.0050 & 0.0566 & 0.0522 & 0.0158 \\
HT2 & 0.0516 & 0.0464 & 0.0178 & 0.0488 & 0.0490 & 0.0144 & 0.0518 & 0.0436 & 0.0164 \\
CM1997 & 0.0658 & 0.0640 & 0.0512 & 0.0584 & 0.0608 & 0.0558 & 0.0658 & 0.0598 & 0.0544 \\ \hline
\end{tabular}
\end{table}

In Table \ref{tab:size_p=5}, the estimated sizes of the ZGZC2020, HT2, KCDG2025, sKCDG2025 and CM1997 tests are presented for all three models. It should be noted that the HT2 test is an exact non-asymptotic test, while the other tests are asymptotic tests. It can be observed from this table that the estimated sizes for all the tests are close to the nominal level in Models 1 and 2. However, in Model 3, the estimated sizes of the ZGZC2020 test and the HT2 test are considerably different from the nominal level, while the estimated sizes of the sKCDG2025 test and the CM1997 test are close to the nominal level. This indicates that significant departure from Gaussianity affects the ZGZC2020 test and the HT2 test in a similar way, but
the sKCDG2025 test and the CM1997 test remain comparably unaffected.
\begin{table}[b!]
\centering
\caption{Estimated sizes at nominal level $\alpha = 5\%$ for the different tests for high-dimensional data based on $10,000$ independent replications for $p=100$ with $n_1 = 40$ and $n_2 = 50$.}
\label{tab:size_p=100}

\vspace{0.1in}
\begin{tabular}{cccccccccc}
\hline
%Model & 1.i. & 1.ii. & 1.iii. & 2.i. & 2.ii. & 2.iii. & 3.i. & 3.ii. & 3.iii. \\ \hline
Model & 1.i. & 2.i. & 3.i. & 1.ii. & 2.ii. & 3.ii. & 1.iii. & 2.iii. & 3.iii. \\ \hline
KCDG2025$^1$ & 0.0487 & 0.0497 & 0.0083 & 0.0090 & 0.0003 & 0.0000 & 0.0323 & 0.0180 & 0.0000 \\
KCDG2025$^2$ & 0.1870 & 0.1837 & 0.0900 & 0.0383 & 0.0260 & 0.0003 & 0.0587 & 0.0540 & 0.0053 \\
sKCDG2025$^1$ & 0.0499 & 0.0543 & 0.0516 & 0.0065 & 0.0060 & 0.0033 & 0.0297 & 0.0303 & 0.0324 \\
sKCDG2025$^2$ & 0.1627 & 0.1661 & 0.1610 & 0.0415 & 0.0409 & 0.0369 & 0.0503 & 0.0518 & 0.0579 \\
ZGZC2020    & 0.0543 & 0.0560 & 0.0090 & 0.0517 & 0.0043 & 0.0000 & 0.0613 & 0.0343 & 0.0003 \\
BS1996      & 0.0737 & 0.0703 & 0.0150 & 0.0577 & 0.0060 & 0.0000 & 0.0697 & 0.0437 & 0.0003 \\
CLX2014     & 0.0337 & 0.0313 & 0.0020 & 0.0620 & 0.0487 & 0.0020 & 0.0513 & 0.0403 & 0.0027 \\
CQ2010      & 0.0737 & 0.0707 & 0.0160 & 0.0590 & 0.0057 & 0.0000 & 0.0703 & 0.0437 & 0.0003 \\
CLZ2014     & 0.2037 & 0.2027 & 0.1410 & 0.0720 & 0.0580 & 0.0107 & 0.1590 & 0.1527 & 0.0493 \\
SD2008      & 0.0243 & 0.0250 & 0.0027 & 0.0503 & 0.0050 & 0.0000 & 0.0470 & 0.0233 & 0.0000 \\
\hline
\end{tabular}
\end{table}

In Table \ref{tab:size_p=100}, we can observe that the estimated sizes of the sKCDG2025$^1$ test (for case~i, namely, equi-correlation) and the sKCDG2025$^2$ test (the other choices) are close to the nominal level of 5\% across all models for $p=100$. 
The estimated sizes of the ZGZC2020 test are close to the nominal level in Models 1 and 2, but significantly smaller than the nominal level for the non-Gaussian distribution in Model 3. The estimated sizes of the BS1996, SD2008 and CQ2010 tests are usually slightly larger than the nominal level of 5\%
in the majority of cases in Models 1 and 2, but quite small compared to the nominal level in Model 3. The estimated sizes of the CLX2014 test are usually not far from the nominal level except in Model 3, where it is almost $0$. 
%The estimated sizes of the SD2008 test are somewhat smaller than the nominal level throughout Models 1 and 2, and become close to zero in Model 3. 
Irrespective of the models, %and values of $p$, 
the estimated sizes of the CLZ2014 test are always quite higher than the nominal~level. Similar observations hold for $p=25$ and $p=50$ (see Section \ref{sec:supplement_numerical} of the Supplementary for these results).
%, its lowest estimated size being almost 12\%.

These observations indicate that the sKCDG2025 test has satisfactory sizes irrespective of Gaussianity or non-Gaussianity, or different values of $p$. The ZGZC2020 test has satisfactory sizes for the Gaussian as well as non-Gaussian distributions (when the moment assumptions for its validity are satisfied), but it fails to yield a proper size in case the sample arises from a non-Gaussian distribution (which does not satisfy the moment conditions for its theoretical validity). The estimated sizes of the other tests are generally not close to the nominal level compared to those of the ZGZC2020 test and the sKCDG2025 test.
\begin{comment}
The results show the following key patterns:
\begin{itemize}
    \item The {KCDG2025 test} consistently achieves empirical sizes close to the nominal level across all models and dimensions.
    \item The {ZGZC2020} test performs comparably under Gaussian and mild heavy-tailed data, but becomes conservative under the Cauchy model.
    \item The {BS1996} and {CQ2010} tests tend to be liberal in Models~1 and~2 but highly conservative in Model~3.
    \item The {CLX2014} and {SD2008} tests exhibit acceptable behavior under Gaussianity but poor calibration for heavy-tailed distributions.
    \item The {CLZ2014} test is severely over-sized across all settings and is therefore omitted from power analysis.
\end{itemize}
\end{comment}
Overall, the sKCDG2025 test demonstrates remarkable stability for the empirical size and remains unaffected by increasing the data dimension, or departure from Gaussianity.

%\subsection{Empirical Power Analysis}
%\newpage
%Estimated powers of the different tests are summarized in Figures~\ref{fig:fig_p=5} and~\ref{fig:fig_p=100}. Each figure presents nine subplots, corresponding to combinations of the three models and three covariance sturctures for a fixed value of $p$.

Next, we study the estimated  powers of these tests. We drop the CLZ2014 test from our power comparison since it is improperly sized. As mentioned before, by varying the values of $\delta$, 
%within a range from 0 to a positive number, 
we generate estimated power curves of the tests. 
For the nine total combinations of the three models (i.e., 1, 2 and 3) and three choices of $\bSigma$ (i.e., i, ii and iii), we get a total of nine plots, which we present together in \autoref{fig:fig_p=5}. 
In the heading of each plot, the column corresponds to one of the models, while the choices of $\bSigma$ are mentioned across the rows. 
Within each plot, the estimated power curves are plotted against varying values of $\delta$.

\begin{figure}
	\centering
	\includegraphics[width=0.95\textwidth]{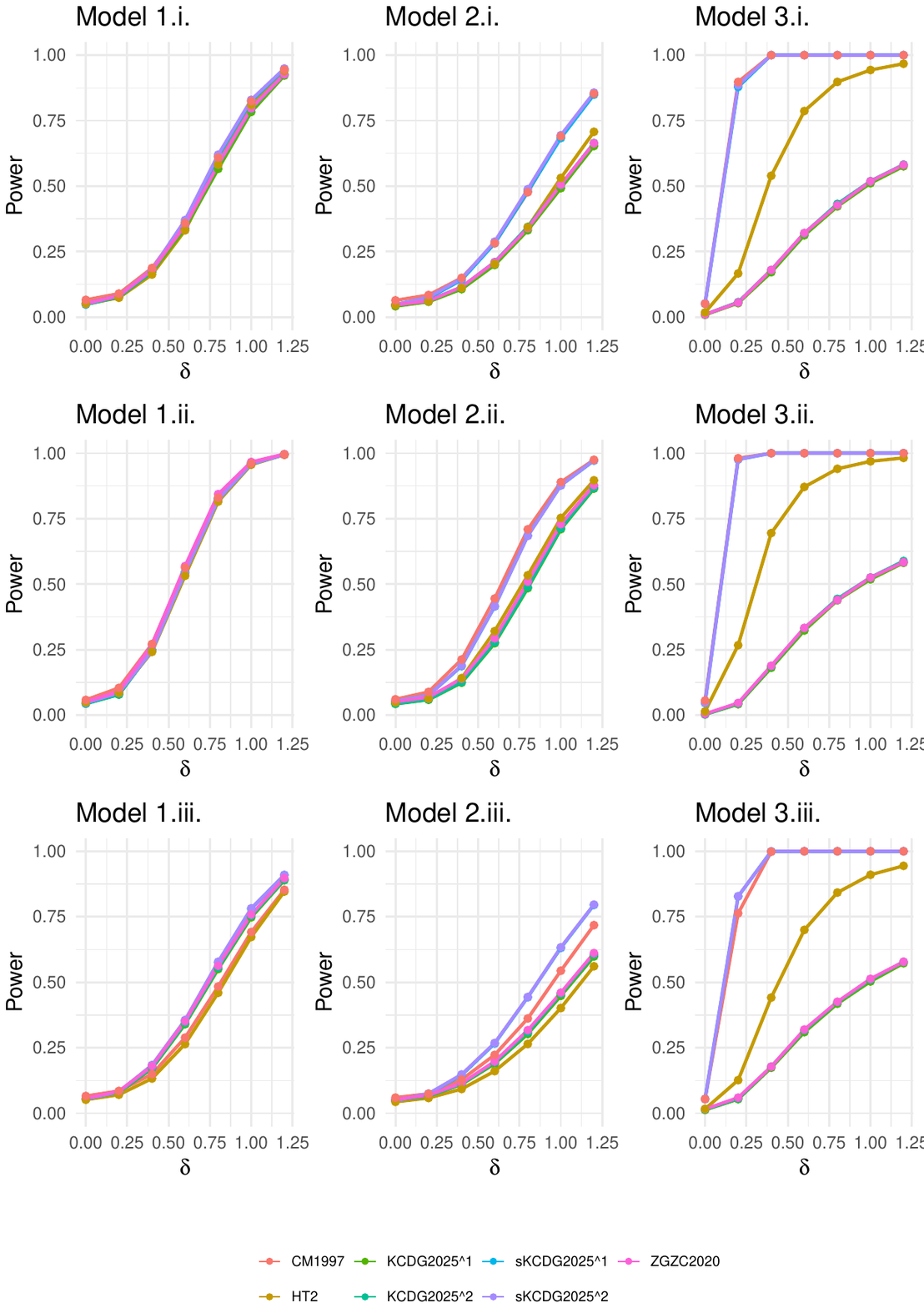}
	\caption{Estimated powers at nominal level $\alpha = 5\%$ based on $10,000$ independent replications at $p=5$ for $n_1 = 40$ and $n_2 = 50$.}
	\label{fig:fig_p=5}
\end{figure}

In \autoref{fig:fig_p=5} (with $p=5$), it can be observed that the un-normalized tests ZGZC2020 and sKCDG2025 have higher estimated power than their normalized counterparts, namely, the HT2 test and the CM1997 test, over the entire range of $\delta$. In Model 1, although the underlying distribution is Gaussian, both the ZGZC2020 test and the sKCDG2025 test have higher estimated power than the HT2 test. This is in spite of the fact that the HT2 test is an exact non-asymptotic test for Gaussian distributions and it is the most powerful invariant test. 
%In Case 2, however, all the normalized tests have uniformly higher estimated power compared to the un-normalized counterparts. In both Case 1 and Case 2, 
Generally, the estimated power curve of the sKCDG2025 test is either higher than or coincides with that of the ZGZC2020 test. However, the ZGZC2020 test has a slightly higher power than the sKCDG2025 test in Model 1.

%\item In low-dimensional Gaussian settings, both KCDG2025 and ZGZC2020 outperform the normalized tests (namely, $T^2$ and CM1997) across nearly the entire range of $\delta$, despite $T^2$ being theoretically optimal under normality.
%\item When the covariance structure is complex or mildly non-Gaussian, normalized tests regain slight advantages; yet, the KCDG2025 test remains competitive.

\begin{figure}
	\centering
	\includegraphics[width=0.95\textwidth]{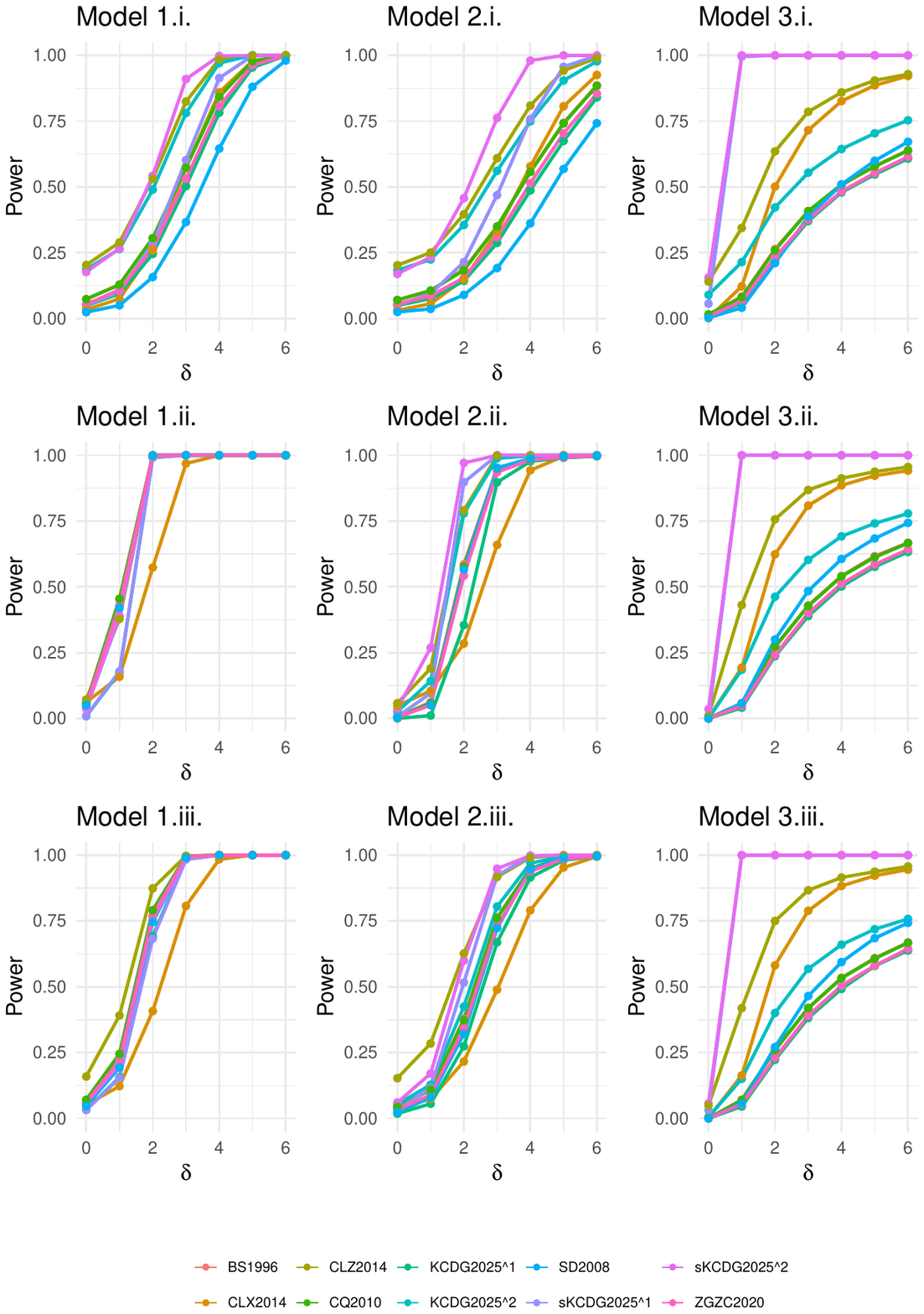}
	\caption{Estimated powers at nominal level $\alpha = 5\%$ based on $10,000$ independent replications at $p=100$ for $n_1 = 40$ and $n_2 = 50$.}
	\label{fig:fig_p=100}
\end{figure}

From the plots in \autoref{fig:fig_p=100} (with $p=100$), we observe that the estimated power curve of the KCDG2025 test is almost uniformly higher than all the other tests, except the BS1996 test and the CQ2010 test for small to moderate values of $\delta$. This holds irrespective of Gaussian, or non-Gaussian distributions. As expected, the difference between the estimated powers of the sKCDG2025 test and the best performing tests among the rest is quite narrow for the Gaussian distribution. For small to moderate values of $\delta$, the estimated powers of the KCDG2025 test is also close to those of the BS1996 test and the CQ2010 test.

For $\delta = 0$ (corresponds to the null hypothesis), the estimated powers of the BS1996 test and the CQ2010 test are %slightly but 
noticeably higher than the nominal level of 5\%.
% This is similar to what we had observed in \autoref{tab:size_p=100}.
This perhaps contributes to the head start of their power curves against the KCDG2025 test for small $\delta$, and vanishes gradually as $\delta$ becomes larger (also see \autoref{tab:size_p=100}). The difference between the power curves of the KCDG2025 test and the other tests increases in the non-Gaussian distribution for Model 2. However, in Model 3, the difference between the powers of the KCDG2025 test and all other tests is quite stark. The ZGZC2020 test performs well in Models 1 and 2, but not in Model 3. Additional power curves for $p = 25$ and $p=50$ are provided in Section \ref{sec:supplement_numerical} of the Supplementary. From \autoref{fig:fig_p=100} and also the plots given in the Supplementary, we observe that the powers of all the tests decrease as $p$ increases, and the difference between the power of the KCDG2025 test and the other tests increases with~$p$.

%The overall finding from 
%\autoref{fig:fig_p=100} 
These observations clearly validate the utility of the sKCDG2025 test irrespective of 
%the Gaussianity or non-Gaussianity of 
the underlying distribution. 
Whether the magnitude of $p$ is similar to the sample size or considerably larger, the KCDG2025 test is observed to perform well and almost uniformly outperforms all other tests if one accounts for the proper size of the tests at the nominal level of $5\%$. Moreover, for a non-Gaussian distribution (when the moment conditions for the validity of other tests do not hold), the difference between the performances of the KCDG2025 tests and all other tests becomes clearly significant.
In fact, the differences in power between the sKCDG2025 test and other tests widen when we increase the value of $p$ from $5$ to $100$, which further reinforces its superior performance.
These results indicate that the sKCDG2025 test successfully balances robustness, validity and power across both classical multivariate as well as high-dimensional situations. In fact, we observe that sKCDG2025$^2$ has marginally better power than sKCDG2025$^1$ in some~cases.

\subsection{Analysis of human colon tissue data} \label{subsec:realdata}

We analyze the colon data available from \href{http://genomics-pubs.princeton.edu/oncology/affydata/index.html}{this link}. The colon data was first analyzed in \cite{alon1999broad}, and later in other papers including \cite{zhang2020simple}. It contains the expression values of 2000 genes from 62 samples of human colon tissue, 22 are from normal tissue and 40 are from tumor tissue. We are interested in testing whether the gene expression values differ between normal and tumor tissue. So, we have $p = 2000$ and $n_1 = 40, n_2 = 22$. Let us consider the tests of the previous section that belong to the high-dimensional regime. The p-values obtained are presented in the first row of \autoref{tab:pvalues}. Except for the SD2008 test, all p-values are notably~small.

\begin{table}[t!]
\centering
\caption{First row: p-values from the colon data with $p = 2000$. Second row: Average p-values from the colon data, averaged over 50 blocks with $p = 40$.}
\label{tab:pvalues}

\vspace{0.1in}
\resizebox{1\textwidth}{!}{
\begin{tabular}{cccccccccc}
\hline
ZGZC2020   &BS1996      &CLX2014     &CQ2010   &SD2008 &KCDG2025$^1$   &KCDG2025$^2$   &sKCDG2025$^1$   &sKCDG2025$^2$  \\\hline
0.000625 &0.000000 &0.000000 &0.000000 &0.270284 &0.001300 &0.000000 &0.000800 &0.000000  \\\hline
0.039675 &0.039190 &0.080385 &0.031240  &0.158169 &0.030928 &0.020946 &0.016822 &0.010854  \\\hline
\end{tabular}
}
\end{table}

Next, we are interested in observing the p-values of the tests when the dimension $p$ is closer in magnitude to either of the sample sizes $n_1$ and $n_2$. To implement this, we partition the data matrix by considering blocks of $40$ consecutive columns. This way, we have $50$ data matrices each with $p = 40$, $n_1 = 40$ and $n_2 = 22$. So, here the magnitude of $p$ is comparable to the sample size.

%\vspace*{0.15in}
\begin{figure}[b!]
	\centering
	\includegraphics[width = 14cm]{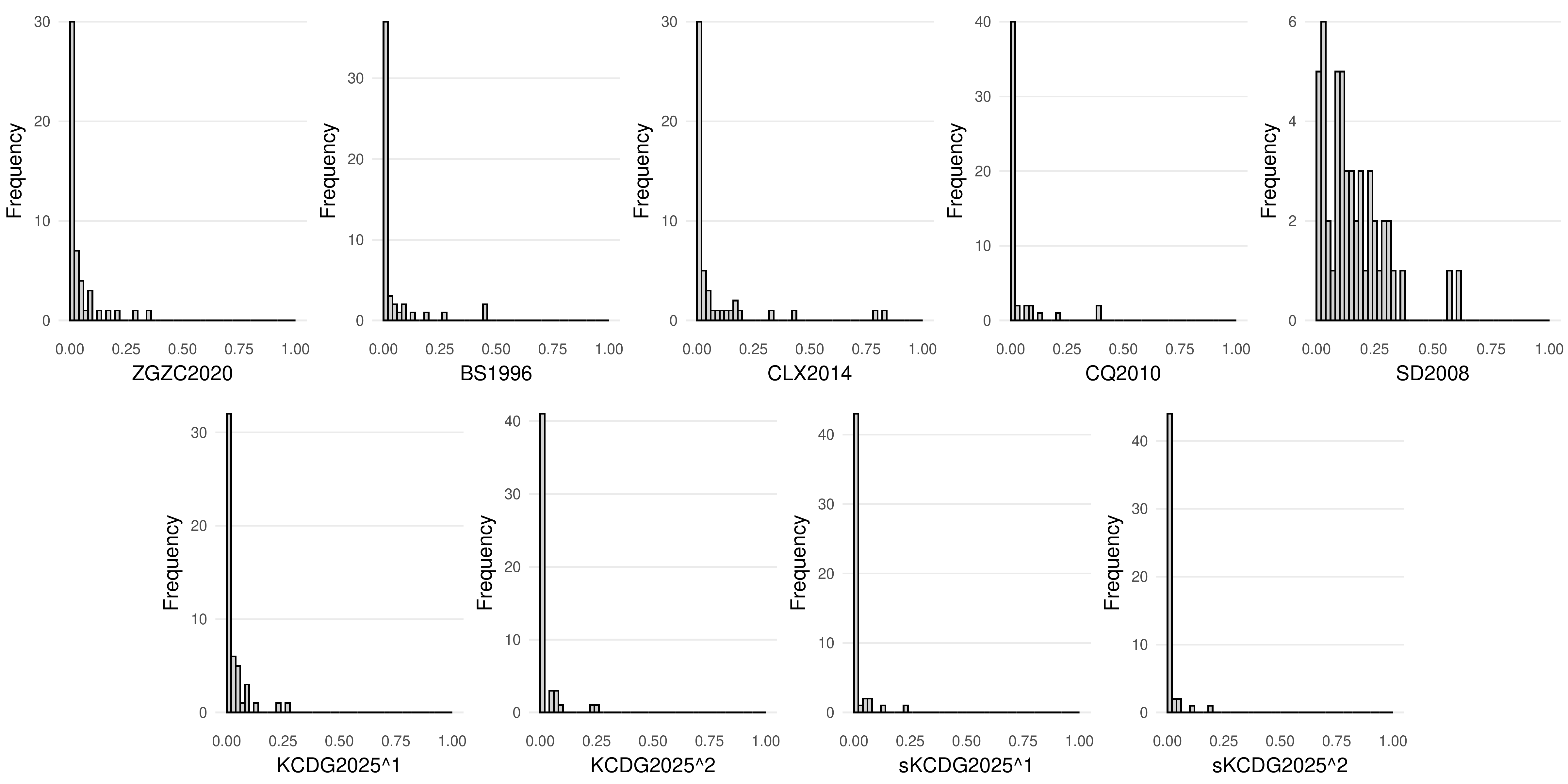}
	\caption{Histograms of p-values of the tests in the human colon tissue data.}
	\label{fig:fig_histogram}
\end{figure}
%\vspace*{-0.1in}

We present the histogram of the 50 p-values for each of the tests in \autoref{fig:fig_histogram}. The averages of the p-values are presented in the second row of \autoref{tab:pvalues}. 
Although the p-values of the sKCDG2025 tests were quite small in the first row of this table, most of the other p-values were also nearly zero and it seems difficult to conclude that one of the tests is better than another test.
%has higher statistical power compared to the other. 
However, as the value of $p$ becomes lower and comparable to the magnitude of $n_1$ and $n_2$ (see the second row of this table), the p-values of all the tests increase. In this case, it can be seen that the sKCDG2025 tests have the lowest average p-value. In \autoref{fig:fig_histogram} too, it can be observed that the p-values of the sKCDG2025 tests are concentrated closest to zero which gives us a strong evidence towards rejecting the~null. %hypothesis. 
%This indicates that the sKCDG2025 tests indeed have higher statistical power compared to the other tests carried out on this~data set.

%\begin{comment}
%\vspace*{-0.25in}

\section{Conclusions} \label{sec:conclusion}

In this work, we have proposed a notion of convergence for fixed dimensional functions of potentially high-dimensional data, that holds uniformly over the dimension. We used this notion to obtain a general class of kernel-based two-sample location tests applicable to data of any dimension. For two specific kernel choices, we obtained the KCDG2025 test and the spatial sign–based sKCDG2025 test. An extensive simulation study and a real data analysis show that the KCDG2025 test outperforms several popular tests for high-dimensional data under both Gaussian and non-Gaussian models, while the sKCDG2025 test surpasses the classical Hotelling's $T^2$ test in certain low-dimensional Gaussian as well as non-Gaussian~settings.

We also demonstrated the applicability of uniform-over-dimension convergence to several standard results on convergence in distribution. Many other results, such as U-statistic asymptotics and other uniform-over-dimension hypothesis tests remain to be explored. We plan to investigate  this in a future work.

\vspace*{-0.2in}
\section*{Appendix: Proofs of mathematical results}%% if no title is needed, leave empty \section*{}.

The following lemma will be used in what follows. The proofs of all the lemmas introduced in this section are given in the Supplementary.

\begin{lemma}\label{fourth_moment}
Suppose $\bz = (z_1, \ldots, z_m)^\top$ satisfies the conditions imposed on $\bz_{ij}$ outlined in assumption \ref{cond3}.  
For any $p \times m$ matrix $\bA$, we have
$\Var(\|\bA\bz\|^2) \leq (5 + \Delta) \tr([\bA \bA^\top]^2 )$, where $\Delta$ is as defined in condition \ref{cond3}.
\end{lemma}

\noindent
\begin{proof}[Proof of \autoref{thm:1}] \textbf{(a)} Under assumption \ref{cond4}, the claim follows from the classical\\ Lindeberg-Feller central limit theorem (CLT). Now, $\E[(\xi_p - \xi)^2] = \sum_{r=1}^p (\rho_{p,r} - \rho_r)^2 + \sum_{r=p+1}^\infty \rho_r^2$ which can be bounded above by $\sum_{r=1}^q (\rho_{p,r} - \rho_r)^2 + \sup_p\sum_{r=q+1}^p \rho_{p,r}^2 + \sum_{r=q+1}^\infty \rho_r^2$ for any $ q< p$. Under assumption \ref{cond5}, taking $\limsup$ as ${p \to \infty}$ and then letting $q \to \infty$ yields the claim.

\textbf{(b)} For real sequences $\{x_{n,p}\}_{n, p\in \mathbb{N}}$ and $\{x_{p}\}_{p\in\mathbb{N}}$, we say $x_{n,p} \to x_p$ uniformly-over-$p$ if $\sup_p |x_{n,p} - x_p| \to 0$ as $n \to \infty$. Since $\tr(\bA) \geq 0$ and $\tr(\bA \bB) \geq 0$ for any two symmetric non-negative definite matrices $\bA$ and $\bB$, we have
\begin{align*}
\left|\frac{\tr(\bSigma_{n,p}^2)}{\tr(\bSigma_{p}^2)} - 1\right| &= \left|\frac{\left(\frac{n_2^2}{n^2} - (1-\tau)^2\right)\tr(\bSigma_{p,1}^2) +  \left(\frac{n_1^2}{n^2} - \tau^2\right)\tr(\bSigma_{p,2}^2) +  \left(\frac{2n_1 n_2}{n^2} - 2\tau(1-\tau)\right)\tr(\bSigma_{p,1} \bSigma_{p,2})}{(1-\tau)^2\tr(\bSigma_{p,1}^2) + \tau^2 \tr(\bSigma_{p,2}^2) + 2\tau(1-\tau)\tr(\bSigma_{p,1}\bSigma_{p,2})} \right|\\
& \leq \left|\frac{n_2^2}{n^2} - (1-\tau)^2\right| \frac{1}{(1-\tau)^2} + \left|\frac{n_1^2}{n^2} - \tau^2\right| \frac{1}{\tau^2} + \left|\frac{2n_1 n_2}{n^2} - 2\tau(1-\tau)\right| \frac{1}{2\tau(1-\tau)}
\end{align*}
which goes to $0$ uniformly-over-$p$ as $n \to \infty$. By \autoref{Slutsky}, the given assertion is equivalent to showing
$T_{n,p} / \sqrt{2\tr(\bSigma_{n,p}^2)} \stackrel{D}{\longrightarrow} \zeta$ as $n\to \infty, p \to \infty$.

Define  $\tilde{\bh}_p\left( \bX_{p,i}, \bY_{p,j} \right)
= \bh_p\left( \bX_{p,i}, \bY_{p,j} \right) - \bdelta_p $, $\bbf(\bX_{p,i}) = \E[\tilde{\bh}_p(\bX_{p,i}, \bY_{p,1} \vert \bX_{p,i})]$ and $\bg(\bY_{p,j}) = \E[\tilde{\bh}_p(\bX_{p,1}, \bY_{p,j} \vert \bY_{p,j})]$. Define $\bar\bh_p\left( \bX_{p,i}, \bY_{p,j} \right)$ as in (\ref{hbar}). 
Observe that
\begin{align*}
& nn_1n_2T_{n,p,0} = \|\sum_{i=1}^{n_1} \sum_{j=1}^{n_2}(\bbf(\bX_{p,i}) + \bg(\bY_{p,j}))\|^2 - n_2^2\sum_{i=1}^{n_1} \|\bbf(\bX_{p,i})\|^2 - n_1^2\sum_{j=1}^{n_2} \|\bg(\bY_{p,j})\|^2 \\
& +   \sum_{i_1 \neq i_2}\sum_{j_1 \neq j_2} \bar{\bh}_p(\bX_{p,i_1}, \bY_{p, j_1})^\top \bar{\bh}_p(\bX_{p,i_2}, \bY_{p, j_2}) +  2  \sum_{i_1 \neq i_2}\sum_{j_1 \neq j_2} \bar{\bh}_p(\bX_{p,i_1}, \bY_{p, j_1})^\top (\bbf(\bX_{p,i_2}) + \bg(\bY_{p,j_2}))\\
&- 2(n_1 + n_2 - 1) \sum_{i=1}^{n_1}\sum_{j=1}^{n_2} \bbf(\bX_{p,i})^\top \bg(\bY_{p,j}) - n_2 \sum_{i_1 \neq i_2} \bbf(\bX_{p,i_1})^\top \bbf(\bX_{p,i_2}) - n_1 \sum_{j_1\neq j_2} \bg(\bY_{p,j_1})^\top \bg(\bY_{p,j_2}).
\end{align*}
Now, $\E[n_2^2\sum_{i=1}^{n_1} \|\bbf(\bX_{p,i})\|^2 + n_1^2\sum_{j=1}^{n_2} \|\bg(\bY_{p,j})\|^2] = nn_1n_2 \tr(\bSigma_{n,p})$. Using \autoref{fourth_moment} and Markov's inequality, it can be shown that 
$$\left[\frac{1}{nn_1n_2}\left(n_2^2\sum_{i=1}^{n_1} \|\bbf(\bX_{p,i})\|^2 + n_1^2\sum_{j=1}^{n_2} \|\bg(\bY_{p,j})\|^2\right) - \tr(\bSigma_{n,p})\right] / \sqrt{2\tr(\bSigma_{n,p}^2)} \stackrel{\P}{\longrightarrow} 0$$ as $n \to \infty, p\to \infty$. We shall show that 
$$ \left[\frac{1}{nn_1n_2}\|\sum_{i=1}^{n_1} \sum_{j=1}^{n_2}(\bbf(\bX_{p,i}) + \bg(\bY_{p,j}))\|^2 - \tr(\bSigma_{n,p})\right] / \sqrt{2\tr(\bSigma_{n,p}^2)} \stackrel{D}{\longrightarrow} \zeta $$
as $n \to \infty, p \to \infty$, which by usual Slutsky's theorem will imply that 
$$ \left[\|\sum_{i=1}^{n_1} \sum_{j=1}^{n_2}(\bbf(\bX_{p,i}) + \bg(\bY_{p,j}))\|^2 - n_2^2\sum_{i=1}^{n_1} \|\bbf(\bX_{p,i})\|^2 - n_1^2\sum_{j=1}^{n_2} \|\bg(\bY_{p,j})\|^2\right]/ \left[nn_1n_2 \sqrt{2\tr(\bSigma_{n,p}^2)}\right] \stackrel{D}{\longrightarrow} \zeta $$
as $n \to \infty, p \to \infty$. Now, we show that all the terms other than $\|\sum_{i=1}^{n_1} \sum_{j=1}^{n_2}(\bbf(\bX_{p,i}) + \bg(\bY_{p,j}))\|^2$, $n_2^2\sum_{i=1}^{n_1} \|\bbf(\bX_{p,i})\|^2$ and $n_1^2\sum_{j=1}^{n_2} \|\bg(\bY_{p,j})\|^2$ converge to $0$ uniformly-over-$p$ in probability after dividing by $\sqrt{\tr(\bSigma_{n,p}^2)}$. The expectation of each of these terms is $0$. It is enough to verify that the expectation of the squares of each of these terms converge to $0$ uniformly-over-$p$. Consider the term 
$\sum_{i_1 \neq i_2}\sum_{j_1 \neq j_2} \bar{\bh}_p(\bX_{p,i_1}, \bY_{p, j_1})^\top \bar{\bh}_p(\bX_{p,i_2}, \bY_{p, j_2})$, whose~square~is
%of the form:
$$\sum_{i_1 \neq i_2}\sum_{j_1 \neq j_2}\sum_{i_3 \neq i_4}\sum_{j_3 \neq j_4} \bar{\bh}_p(\bX_{p,i_1}, \bY_{p, j_1})^\top \bar{\bh}_p(\bX_{p,i_2}, \bY_{p, j_2}) \bar{\bh}_p(\bX_{p,i_3}, \bY_{p, j_3})^\top \bar{\bh}_p(\bX_{p,i_4}, \bY_{p, j_4}).$$
After taking expectation, only the terms for which $|\{i_1, i_2, i_3, i_4\}| \leq 2$ and $|\{j_1, j_2, j_3, j_4\}| \leq 2$ survive and there are $O(n_1^2 n_2^2)$ such terms. Using the Cauchy-Schwartz (CS) inequality, each such term can be bounded above by
$$\left(\E[\|\bar{\bh}_p(\bX_{p,i_1}, \bY_{p, j_1})\|^4] \E[\|\bar{\bh}_p(\bX_{p,i_2}, \bY_{p, j_2})\|^4] \E[\|\bar{\bh}_p(\bX_{p,i_3}, \bY_{p, j_3})\|^4] \E[\|\bar{\bh}_p(\bX_{p,i_4}, \bY_{p, j_4})\|^4]\right)^{1/4}.$$
By condition \ref{cond2}, this becomes bounded uniformly-over-$p$ after dividing by $\tr(\bSigma_{p}^2)$. Since $\tr(\bSigma_{n,p}^2) / \tr(\bSigma_{p}^2)$ is bounded uniformly-over-$p$ for large $n$, $$\sup_p\E\left[\left(\frac{1}{nn_1n_2}\sum_{i_1 \neq i_2}\sum_{j_1 \neq j_2} \bar{\bh}_p(\bX_{p,i_1}, \bY_{p, j_1})^\top \bar{\bh}_p(\bX_{p,i_2}, \bY_{p, j_2})/ \sqrt{\tr(\bSigma_{n,p}^2)}\right)^2\right] = O\left(1 / n^2\right),$$ which shows that the concerned term goes to $0$ in probability uniformly-over-$p$. The other terms can be handled similarly using \autoref{fourth_moment} and \ref{cond1}-\ref{cond3}. Using Slutsky's theorem, we can then get $T_{n,p,0} / \sqrt{2\tr(\bSigma_{n,p}^2)} \stackrel{D}{\longrightarrow} \zeta$ as $n\to \infty, p \to \infty$. Therefore, all we need to show is that
$$ \left[\|\bw_{n,p}\|^2 - \tr(\bSigma_{n,p})\right]/\sqrt{2\tr(\bSigma_{n,p}^2)} \stackrel{D}{\longrightarrow} \zeta \quad \text{as $n\to \infty, p \to \infty$},$$
where 
$$\bw_{n,p} = \frac{1}{\sqrt{nn_1n_2}}\left(\sum_{i=1}^{n_1} \sum_{j=1}^{n_2}(\bbf(\bX_{p,i}) + \bg(\bY_{p,j}))\right) = \sqrt{\frac{n_1n_2}{n}} \left(\frac{1}{n_1}\sum_{i=1}^{n_1} \bbf(\bX_{p,i}) - \frac{1}{n_2}\sum_{j=1}^{n_2} (-\bg(\bY_{p,j}))\right).$$
The remainder of the proof establishes this claim. Theorem 1 of \cite{zhang2021hetero} handles a similar expression, with $\bbf(\bX_{p,i})$ and $(-\bg(\bY_{p,j}))$ as observations from the two groups. Under condition \ref{cond4}, $\zeta \sim N(0,1)$ and the claim follows from the proof of Theorem 1, part (b) of \cite{zhang2021hetero}, which builds on the work of \cite{chen2010two}. 

%{\color{red}
We now establish the claim under condition \ref{cond5}. Let $\bZ_i = \sqrt{\frac{n_2}{nn_1}}\bbf(\bX_{p,i})$ for $1 \leq i \leq n_1$ and $\bZ_i = \sqrt{\frac{n_1}{nn_2}}\bg(\bY_{p,i})$ for $n_1+1\leq i \leq n$ (the dependence on $p$ is implicit here). The $\bZ_i$'s are independent, $\E[\bZ_i] = {\bf 0}_p$, $\bw_{n,p} = \sum_{i=1}^n\bZ_i$ and $\E[\|\bw_{n,p}\|^2] = \tr(\bSigma_{n,p})$. Let $\bu_{p,1}, \ldots, \bu_{p,p}$ be the mutually orthogonal eigenvectors of $\bSigma_{p}$ associated with the eigenvalues $\lambda_{p,1}\geq \cdots \geq \lambda_{p,p}$. Define $\xi_{n,p,r} = \bu_{p,r}^\top \bw_{n,p}$ for $1 \leq r \leq p$. We have $\|\bw_{n,p}\|^2 = \sum_{r=1}^p \xi_{n,p,r}^2$ and $\E[\xi^2_{n,p,r}] = \bu_{p,r}^\top \bSigma_{n,p} \bu_{p,r} = \lambda_{n,p,r}$ (say). Fix $t \in \mathbb{R}$ and $\epsilon> 0$. Let $\varphi_{n,p}(t)$, $\varphi_{n,p}^{(q)}(t)$, $\varphi_{\tilde{\zeta}_p^{(q)}}(t)$, $\varphi_{\zeta^{(q)}}(t)$ and $\varphi_{\zeta}(t)$ denote the characteristic functions of 
$\sum_{r=1}^p (\xi_{n,p,r}^2 - \lambda_{n,p,r}) / \sqrt{\tr(\bSigma_{n,p}^2)}$, $\sum_{r=1}^q (\xi_{n,p,r}^2 - \lambda_{n,p,r}) /\sqrt{\tr(\bSigma_{n,p}^2)}$, $\sum_{r=1}^q\rho_{p,r}(W_r - 1) /\sqrt{2}$, $\sum_{r=1}^q\rho_{r}(W_r - 1) / \sqrt{2}$ and $\sum_{r=1}^\infty\rho_{r}(W_r - 1) / \sqrt{2}$, respectively, for any $q \leq p$. Our target is to show $|\varphi_{n,p}(t) - \varphi_{\zeta}(t)| \to 0$ as $n\to \infty, p \to \infty$. Consider the bound
\begin{align}
|\varphi_{n,p}(t) - \varphi_{\zeta}(t)| \leq & ~|\varphi_{n,p}(t) - \varphi_{n,p}^{(q)}(t)| + |\varphi_{n,p}^{(q)}(t) - \varphi_{\tilde{\zeta}_p^{(q)}}(t)|\nonumber\\ & \qquad+ |\varphi_{\tilde{\zeta}_p^{(q)}}(t) - \varphi_{\zeta^{(q)}}(t)| + |\varphi_{\zeta^{(q)}}(t) - \varphi_{\zeta}(t)|. \label{thm1:eq1}
\end{align}  
Define $i=\sqrt{-1}$. We shall use the inequality $[\E|e^{iX} - e^{iY}|]^2 \leq E[(X-Y)^2]$ for rvs $X$ and $Y$. 
\begin{lemma}\label{thm:1_lemma}
Under assumptions \ref{cond1}-\ref{cond3} and \ref{cond5}, for any $t\in\mathbb{R}$ and $\epsilon >0$, there exist $Q, N \in \mathbb{N}$ such that  $|\varphi_{n,p}(t) - \varphi_{n,p}^{(q)}(t)| < \epsilon$ for all $q \geq Q$, $n \geq N$ and $ p \geq  q$.
\end{lemma}

Pick $Q, N$ using this lemma. Since $\lim_{q\to\infty}\sum_{r=q+1}^\infty \rho_r^2 = 0$, as there exists $Q'$ such that for $q \geq Q'$, $|\varphi_{\zeta^{(q)}}(t) - \varphi_{\zeta}(t)| < \epsilon$. Fix any $q \geq \max\{Q, Q'\}$. Since $\lim_{p\to\infty}\rho_{p,r} = \rho_r$ for any $r$, there exists $P\geq q$ depending on $q, t$ and $\epsilon$ such that for $p \geq P$, $|\varphi_{\tilde{\zeta}_p^{(q)}}(t) - \varphi_{\zeta^{(q)}}(t)| < \epsilon$. 

Define ${\bf \Lambda}_q = [\bu_{p,1}, \ldots, \bu_{p,q}]$. We shall apply \autoref{lyapunov} on $\sum_{r=1}^n {\bf \Lambda}_{q}^\top \bZ_r / \left(2\tr(\bSigma_{n,p}^2)\right)^{1/4}$. The $(i,j)$-th element of its covariance matrix is $\bu_{p,i}^\top \bSigma_{n,p} \bu_{p,j} / \sqrt{2\tr(\bSigma_{n,p}^2)}$ which converges to $\rho_{p,i} \mathbb{I}\{i=j\}/ \sqrt{2}$ as $n \to \infty$ uniformly-over-$p$, for any $1 \leq i,j \leq q$. Also $\rho_{p,i} \to \rho_i \in [0,1]$ as $p \to \infty$. Hence, conditions (a) and (b) of \autoref{lindeberg} hold. The uniform-over-$p$ Lyapunov's condition can be verified (with $\delta = 2$) by using \autoref{fourth_moment}. Thus, $\sum_{r=1}^n {\bf \Lambda}_{q}^\top \bZ_r / \left(2\tr(\bSigma_{n,p}^2)\right)^{1/4} \Longrightarrow N_q(\mathbf{0}_{q},  \mathbf{D}_{p,q})$ uniformly-over-$p$ with $\mathbf{D}_{p,q}$ being a diagonal matrix with the elements $\{\rho_{p,r} / \sqrt{2}\}_{r=1,\ldots,q}$. Note that
$\{\zeta_p\}_{p\in\mathbb{N}}$ satisfies assumptions \ref{assumption1} and \ref{assumption2}. Using Theorems \autoref{mappingthm} and \autoref{portmanteau}, we conclude that there exists $N'$ depending only on $q, t$ and $\epsilon$ such that for $n \geq N'$, we have 
$\sup_{p\geq q}|\varphi_{n,p}^{(q)}(t) - \varphi_{\tilde{\zeta}_p^{(q)}}(t)| < \epsilon$.

Thus, equation \eqref{thm1:eq1} gives $|\varphi_{n,p}(t) - \varphi_{\zeta}(t)| < 4\epsilon$ for $n \geq \max\{N, N'\}$ and $p \geq P$.
Hence, $\varphi_{n,p}(t) \to \varphi_{\zeta}(t)$ as $n, p \to \infty$ (with no relation between growth rates). This proves part (b).
%}

\textbf{(c)} Let $F_{n,p}, F_p$ and $F$ denote the dfs of $T_{n,p} /\sqrt{2\mathrm{tr}(\bSigma_{p}^2)}, \zeta_p$ and $\zeta$, respectively. Fix $x \in \mathbb{R}$ and $\epsilon > 0$. Using (a), pick $P$ such that for $p \geq P$, $|F_p(x) - F(x)| < \epsilon$. Choose $P', N$ so that for $n \geq N$ and $p \geq P'$, $|F_{n,p}(x) - F(x)| < \epsilon$. Using the usual CLT, we choose $N'$ large enough so that for $n \geq N'$ and $p \leq \max\{P, P'\}$, $|F_{n,p}(x) - F_p(x)| < \epsilon$. For $p \geq \max\{P,P'\}$ and any $n \geq N$, we have $|F_{n,p}(x) - F_p(x)| \leq |F_{n,p}(x) - F(x)| + |F_p(x) - F(x)| < 2\epsilon$. Thus, for $n \geq \max(N, N')$, we have $\sup_{p}|F_{n,p}(x) - F_p(x)| < 2\epsilon$. Using \autoref{portmanteau}, we conclude that $T_{n,p,0} / \sqrt{2\mathrm{tr}(\bSigma_{p}^2)} \Longrightarrow \zeta_p$ uniformly-over-$p$ as $n \to \infty$.
\end{proof}

\begin{proof}[Proof of \autoref{lemma_gaussian_estimation}]
Let $\hat\lambda_{n,p,1} \geq  \cdots \geq \hat\lambda_{n,p,p}$ be the eigenvalues of $\hat{\bSigma}_p$. Fix $t \in \mathbb{R}$. Let $\varphi_{n, p}(t)$ and $\varphi_{p}(t)$ denote the characteristic functions of $V_{n, p  ,1} / \sqrt{2\tr(\bSigma_p^2)} \stackrel{d}{=} (\sum_{i = 1}^{p } \hat\lambda_{n,p, i}(W_i - 1)) / \sqrt{2 \sum_{i = 1}^{p } \lambda_{p, i}^2}$ and $\zeta_p \stackrel{d}{=} (\sum_{i = 1}^{p } \lambda_{n,p, i}(W_i - 1)) / \sqrt{2 \sum_{i = 1}^{p } \lambda_{p, i}^2}$, where $W_i$'s are i.i.d. $\chi^2_1$ variates (independent of $\hat\lambda_{n,p, i}$'s). Therefore, we have
\begin{align}
&|\varphi_{n, p}(t) - \varphi_p(t)|^2\leq \left[|t| \E\left|\frac{\sum_{i = 1}^{p } (\hat\lambda_{n,p, i} - \lambda_{p, i})(W_i -1)}{\sqrt{2 \sum_{i = 1}^{p } \lambda_{p, i}^2}} \right|\right]^2 \leq  \frac{t^2}{2\tr(\bSigma_p^2)} \E\left[\sum_{i = 1}^{p } (\hat\lambda_{n,p, i} - \lambda_{p, i})(W_i-1)\right]^2 \nonumber\\
&= \frac{t^2}{2\tr(\bSigma_p^2)} \left(\sum_{i = 1}^{p } \E(\hat\lambda_{p, i} - \lambda_{p, i})^2 \E(W_i - 1)^2 + 2\sum_{i< j}\E\left[(\hat\lambda_{p,i} - \lambda_{p,i})(\hat\lambda_{p,j} - \lambda_{p,j})\right] \E\left[(W_i -1 )(W_j - 1)\right]\right) \nonumber\\
&= \frac{t^2}{\tr(\bSigma_p^2)} \E\left[\sum_{i = 1}^{p } (\hat\lambda_{p, i} - \lambda_{p, i})^2\right]  \leq \frac{t^2}{\tr(\bSigma_p^2)} \E\left[\tr[(\hat\bSigma_p - \bSigma_p)^2]\right].\label{lemma_gaussian_estimation_eq1}
\end{align}
We have used Corollary 6.3.8 of \cite{horn2012} in the last line. 
%Now, we use the following lemma. 
%The proof is given in the Supplementary.

\begin{lemma}\label{estimation_auxiliary}
Under the assumptions of \autoref{lemma_gaussian_estimation}, we have 
$$ \sup_{p\in \mathbb{N}} \left[\tr\left(\E\left[(\hat\bSigma_p - \bSigma_p)^2\right]\right) / \tr(\bSigma_p^2)\right]  \to 0 \text{ as } n \to \infty.$$
\end{lemma}
\noindent
The result is now immediate.
\end{proof}

\begin{proof}[Proof of \autoref{thm:2}] 
\textbf{(a)} Fix $\alpha \in (0, 1)$. Let $\tilde{c}_{n,p}, c_{n,p}$ and $c_p$ denote the $(1-\alpha)$-th quantile of the sample $V_{n,p,1} /\sqrt{\tr(\bSigma_{n,p}^2)}, \ldots, V_{n,p,M} /\sqrt{\tr(\bSigma_{n,p}^2)}$ and the distributions $V_{n,p,1} /\sqrt{2\tr(\bSigma_{n,p}^2)} \sim F_{n,p}$ and $\zeta_p \sim F_p$, respectively. Fix any $\epsilon, \epsilon' > 0$. Using Theorem 5.9 of \cite{junshao2003}, we get
\begin{equation}
\P(|\tilde{c}_{n,p} - c_{n,p}| > \epsilon) \leq 2Ce^{-2M[\min\left\{F_{n,p}\left(c_{n,p} + \epsilon\right) - (1-\alpha), (1-\alpha) - F_{n,p}\left(c_{n,p} - \epsilon\right) \right\}]^2}, \label{theorem_level_alpha_eq1}
\end{equation}
where $C$ is a universal constant. By \autoref{thm:1}, $F_p(x) \to F(x)$ uniformly over $x \in \mathbb{R}$ with $\zeta \sim F$. By \autoref{lemma_quantile},
$\sup_{\, p \in \mathbb{N}} \sup_{x \in \mathbb{R}}|F_{n,p}(x) - F(x)| \to 0$ and $c_{n,p} \to c_p$ uniformly-over-$p$ as $n \to \infty$. From the proof of \autoref{lemma_quantile}, we have $\lim_{p\to\infty}c_p = c$, the $(1-\alpha)$-th quantile of $F$. Pick $P$ and $N_1$ such that for $n \geq N_1$ and $p \geq P$, $|c_{n,p} - c| < \epsilon / 2 $. Define $\delta = \min\{F(c + \epsilon / 2) - (1-\alpha), (1-\alpha) - F(c - \epsilon / 2)\} > 0$. Pick $N_2$ such that for $n \geq N_2$ and any $p$, $\sup_{x \in \mathbb{R}}|F_{n,p}(x) - F(x)| < \delta / 2$.  Then, for $n \geq \max\{N_1, N_2\}$ and $p \geq P$, $F_{n,p}(c_{n,p} + \epsilon) \geq F(c_{n,p} + \epsilon) - \delta / 2 \geq F(c + \epsilon/2) - \delta / 2$ and $F_{n,p}(c_{n,p} - \epsilon) \leq F(c_{n,p} - \epsilon) + \delta / 2 \leq F(c - \epsilon/2) + \delta / 2$ and
$$\min\left\{F_{n,p}\left(c_{n,p} + \epsilon\right) - (1-\alpha), (1-\alpha) - F_{n,p}\left(c_{n,p} - \epsilon\right) \right\} \geq \delta - \frac{\delta}{2} = \frac{\delta}{2} > 0.$$
Note that $\delta$ does not depend on $n$. Similarly, it can be shown that there exists $N_3$ such~that  
$$\min\left\{F_{n,p}\left(c_{n,p} + \epsilon\right) - (1-\alpha), (1-\alpha) - F_{n,p}\left(c_{n,p} - \epsilon\right) \right\} > 0
\vspace*{-0.1in}
$$
for $n \geq N_3$ and $p < P$ and the lower bound does not depend on $n$. Thus, for any $n \geq N = \max\{N_1, N_2, N_3\}$, we get
\vspace*{-0.15in}
$$\inf_{p} \min\left\{F_{n,p}\left(c_{n,p} + \epsilon\right) - (1-\alpha), (1-\alpha) - F_{n,p}\left(c_{n,p} - \epsilon\right) \right\} > 0
\vspace*{-0.1in}
$$
and the lower bound does not depend on $n$. Using (\ref{theorem_level_alpha_eq1}), there exists a $M_0 \in \mathbb{R}$ such that
$\P(|\tilde{c}_{n,p} - c_{n,p}| > \epsilon) < \epsilon'$
for any $M \geq M_0$ and  $n \geq N$. Combining this with the fact that $c_{n,p} \to c_p$ uniformly-over-$p$, we have
$\tilde{c}_{n,p} \stackrel{\P}{\longrightarrow} c_p$
uniformly-over-$p$ as $n \to \infty$ and $M \to \infty$. 

Recall that $\text{H}_0$ is rejected if $T_{n,p} > \hat{c}_{n,p}(\alpha) = \tilde{c}_{n,p} \sqrt{\tr(\bSigma_{n,p}^2)}$. By \autoref{thm:1} and \autoref{Slutsky}, we~have
$$\sup_p\left|\P\left(T_{n,p} > \hat{c}_{n,p}\right) - (1-\alpha)\right| = \sup_p\left|\P\left(T_{n,p} / \sqrt{\tr(\bSigma_p^2)} > \tilde{c}_{n,p}\right) - \P(\zeta_p > c_p)\right| \to 0$$
as $n \to \infty$ and $M \to \infty$. This completes the proof of part (a).

\textbf{(b)} We have $\E[T_{n,p}] = (n_1 -1)(n_2 - 1) \|\bdelta_p\|^2 / n$ and $\Var(T_{n,p}) $ is equal to 
\begin{align*}
\frac{1}{n^2 n_1^2 n_2^2} \sum_{i_1 \neq i_2} \sum_{j_1 \neq j_2} \sum_{i_3 \neq i_4} \sum_{j_3 \neq j_4} \Cov (\bh_p(\bX_{p,i_1}, \bY_{p,j_1})^\top \bh_p(\bX_{p,i_2}, \bY_{p,j_2}), \bh_p(\bX_{p,i_3}, \bY_{p,j_3})^\top \bh_p(\bX_{p,i_4}, \bY_{p,j_4})).
\end{align*}
If $|\{i_1, i_2, i_3, i_4\}| = |\{j_1, j_2, j_3, j_4\}| = 4$, then the covariance term is $0$. The number of non-zero covariance terms is at most $O(n^7)$. Using the Cauchy-Schwartz (CS) inequality, any such term can be bounded above by $\Var(\bh_p(\bX_{p,i_1}, \bY_{p,j_1})^\top \bh_p(\bX_{p,i_2}, \bY_{p,j_2}))$, which is equal to
\vspace*{-0.1in}
\begin{align*}
& \E[\tr\left(\bh_p(\bX_{p,i_1}, \bY_{p,j_1})\bh_p(\bX_{p,i_1}, \bY_{p,j_1})^\top \bh_p(\bX_{p,i_2}, \bY_{p,j_2})\bh_p(\bX_{p,i_2}, \bY_{p,j_2})^\top \right)] - \|\bdelta_p\|^4 \\
&= \tr((\bSigma_{p,0} + \bdelta_p \bdelta_p^\top)^2) - \|\bdelta_p\|^4 \leq  \tr(\bSigma_{p,0}^2) + 2\|\bdelta_p\|^2 \sqrt{\tr(\bSigma_{p,0}^2)}.
\end{align*}
\vspace*{-0.1in}
Here, $\bSigma_{p,0} = \Var(\bh_p(\bX_{p,i}, \bY_{p,j}))$. Under (i), $\tr(\bSigma_{p,0})$ and $\tr(\hat\bSigma_{p,0})$ are bounded uniformly-over-$p$. Using a similar line of arguments as given in \autoref{estimation_auxiliary}, it follows that under (ii), $\tr(\bSigma_{p, 0}^2) / \tr(\bSigma_{p}^2)$ and $\E[\tr(\hat\bSigma_p^2)] / \tr(\bSigma_{p}^2)$ are both bounded uniformly-over-$p$. Thus, using (i) or (ii), both $\tr(\bSigma_{p, 0}^2) / \|\bdelta_p\|^4$ and $\E[\tr(\hat\bSigma_p^2)] / \|\bdelta_p\|^4$ are bounded uniformly-over-$p$. Fix $0 < \epsilon < 1$. By Markov's inequality, we have
$\P\left(\left|T_{n,p} / \E[T_{n,p}] - 1\right| > \epsilon\right) \leq \Var(T_{n,p})/(\epsilon^2\E[T_{n,p}]^2) \leq C / (n\epsilon^2)$
for some constant $C$ (independent of $n$ and $p$). Thus,
\begin{align*}
\P(T_{n,p} > \hat c_{n,p}(\alpha)) &\geq \P\left((1-\epsilon)\E[T_{n,p}] > \hat{c}_{n,p}, (1-\epsilon)\E[T_{n,p}] < T_{n,p}\right)\\
&\geq \P\left((1-\epsilon)\E[T_{n,p}] > \hat{c}_{n,p}(\alpha)) - \P((1-\epsilon)\E[T_{n,p}] \geq T_{n,p}\right)  \\
&\geq \P\left((1-\epsilon)\E[T_{n,p}] > |\hat{c}_{n,p}(\alpha)|\right) -\frac{C}{n\epsilon^2} \\
& = 1- \E[\P\left(|\hat{c}_{n,p}(\alpha)| \geq (1-\epsilon)\E[T_{n,p}] \right \vert \bX_i, \bY_j)] - \frac{C}{n\epsilon^2} \\
& \geq 1 - \frac{2n^2\E[\tr(\hat\bSigma_p^2)]}{(n_1-1)^2 (n_2-1)^2(1-\epsilon)^2 \|\bdelta_p\|^4} - \frac{C}{n \epsilon^2}.
\end{align*}
In the last line, we have used Markov's inequality and the fact that $\E[\hat{c}_{n,p}(\alpha)^2 \vert \bX_{p,i}, \bY_{p,j}] = 2\tr(\hat\bSigma_p^2)$. Thus, $\P(T_{n,p} > \hat c_{n,p}(\alpha)) \to 1$ uniformly-over-$p$ as $ n \to \infty$.
\end{proof}

% \end{appendix}

%\newpage
%\medskip
%\vspace*{0.1in}
%\begin{center}
%{\large\bf SUPPLEMENTARY MATERIAL}
%\end{center}
%\vspace*{-0.3in}

%\begin{description}

%\item \textbf{Supplementary: Uniform-over-dimension location tests for high-dimensional~data:} The Supplementary contains the proofs of the additional theoretical results used in this paper and some numerical results for $p=25$ and $p=50$.
%\item[R-package for  MYNEW routine:] R-package ÒMYNEWÓ containing code to perform the diagnostic methods described in the article. The package also contains all datasets used as examples in the article. (GNU zipped tar file)
%\vspace*{-0.3in}

%\end{description}

\newpage

% increment the 'supplement' counter at the begining of the supplement
\stepcounter{supplement}
% section, equation, figure, table numbering renamed to S1, S2 etc.
\renewcommand{\thesection}{S\arabic{section}}
\renewcommand{\thefigure}{S\arabic{figure}}
\renewcommand{\thetable}{S\arabic{table}}
\renewcommand{\theequation}{S\arabic{equation}}

\section*{Supplementary Material}

\section{Proofs of the mathematical results of Section~2} \label{sec:supplement_theory}
	%% if no title is needed, leave empty \section*{}.
	
This section provides detailed statements and proofs of the results stated in Section \ref{sec:theory} of the paper concerning uniform-over-$p$ convergence.
	%We begin by stating and proving the uniform-over-$p$ version of the continuous mapping theorem.
	
	\begin{stheorem}[Uniform-over-$p$ continuous mapping theorem]\label{mappingthm}
		Let $\bX_{n, p}$ and $\bX_p$ be $d$-dimensional random vectors with $\bX_{n, p} \Longrightarrow \bX_p$ uniformly-over-$p$ and $g : \mathbb{R}^d \to \mathbb{R}$ be a continuous function. Then, $g\left( \bX_{n, p} \right) \Longrightarrow g\left( \bX_p \right)$ uniformly-over-$p$.
	\end{stheorem}
	\begin{proof}[Proof of \autoref{mappingthm}]
		For any bounded and continuous $f : \mathbb{R} \to \mathbb{R}$, the composite function $h : \mathbb{R}^d \to \mathbb{R}$ defined by $h( \bx ) = f( g( \bx ) )$ is bounded and continuous. Since $\bX_{n, p} \Longrightarrow \bX_p$ uniformly-over-$p$, from \autoref{definition1}, we have
		\begin{align*}
			\sup_p \left| \E[ f( g( \bX_{n, p} ) ) ] - \E[ f( g( \bX_p ) ) ] \right|
			= \sup_p \left| \E[ h( \bX_{n, p} ) ] - \E[ h( \bX_p ) ] \right| \to 0
		\end{align*}
		as $n \to \infty$, which implies that $g\left( \bX_{n, p} \right) \Longrightarrow g\left( \bX_p \right)$ uniformly-over-$p$.
	\end{proof}
	
	Next, we state a basic lemma concerning approximation by indicator functions. To do so, we need some notation. For any $A \subseteq \mathbb{R}^d$, let $\partial A$ and $\bar{A}$ respectively denote the boundary and closure of the set $A$ with respect to the standard metric topology. For $A \subseteq \mathbb{R}^d$ and $\delta > 0$, define $(\partial{A})^\delta = \{\bx : dist(\bx, A) < \delta\}$, where the distance of the point $\bx \in \mathbb{R}^d$ from $A$ is defined as $dist(\bx, A) = \inf_{\bz \in A}\|\bx-\bz\|$. It is easy to check that $dist(\bx, A) = 0$ iff $\bx \in \bar{A}$ and $dist(\bx, A) \leq \|\bx -\by\| + dist(\by, A)$ (using the triangle inequality) for any $\by \in \mathbb{R}^d$. The last inequality also shows that for any $\bx, \by \in \mathbb{R}^d$, we have
	\begin{equation}
		|dist(\bx, A) - dist(\by, A)| \leq \|\bx - \by\|.
		\label{indapprox_eq1}
	\end{equation} 
	%Lemma for approximation of indicators
	\begin{slemma}\label{indapprox}
		Let $A \subseteq \mathbb{R}^d.$ For all $\epsilon>0$, there exists functions $g_\epsilon, h_\epsilon : \mathbb{R}^d \to [0,1] $ such that:
		\begin{enumerate}[label = (\roman*), ref = (\roman*)]
			\item $g_\epsilon \leq \mathbb{I}_A \leq h_\epsilon$, where $\mathbb{I}_A$ is the indicator function of the set $A$. \label{indapprox_pt1}
			\item $h_\epsilon(\bx) - g_\epsilon(\bx) \leq 1$ for all $\bx \in \mathbb{R}^d$ and if $dist(x, \partial A) \geq \epsilon$, then $h_\epsilon(\bx) - g_\epsilon(\bx) = 0$. \label{indapprox_pt2}
			\item $g_\epsilon$ and $h_\epsilon$ are bounded, Lipschitz continuous functions that are zero outside a compact~set. \label{indapprox_pt3}
		\end{enumerate}
	\end{slemma}
	
	%proof of the lemma 
	\begin{proof}[Proof of \autoref{indapprox}]
		For any $\epsilon>0,$ define
		$$ h_\epsilon(\bx) = \left\{\begin{array}{ll}
			1 -\frac{dist(\bx, A)}{\epsilon} & \text{if}\,\, dist(\bx, A) \leq \epsilon, \\
			0 & \text{otherwise},
		\end{array}\right.$$
		and
		$$ g_\epsilon(\bx) = \left\{\begin{array}{ll}
			\frac{dist(\bx, A^c)}{\epsilon} & \text{if}\,\, dist(\bx, A^c) \leq \epsilon, \\
			1 & \text{otherwise}.
		\end{array}\right.$$
		In other words, $h_\epsilon$ takes the value $0$ on $\{\bx : dist(\bx, A) \geq \epsilon\}$, lies strictly between $0$ and $1$ on $\{\bx : 0 < dist(\bx, A) < \epsilon\}$ and takes the value $1$ on $\bar{A}$. Similarly, $g_\epsilon$ takes the value $0$ on $\overline{A^c}$, lies strictly between $0$ and $1$ on $\{\bx \in A : 0 < dist(\bx, A^c) < \epsilon\}$ and takes the value $1$ on $\{\bx \in A : dist(\bx, A^c) \geq \epsilon\}$. Hence, $0 \leq g_\epsilon \leq \mathbb{I}_A  \leq h_\epsilon \leq 1,$ $h_\epsilon(\bx) - g_\epsilon(\bx) \leq 1$, and $g_\epsilon, h_\epsilon$ are zero outside of the compact set $\{\bx: dist(\bx , A) \leq \epsilon\}.$ If  $dist(\bx, \partial A) \geq \epsilon$, then either $\bx \in A$ and $dist(\bx, A^c) \geq \epsilon$ or $\bx \in A^c$ and $dist(\bx, A) \geq \epsilon$ and therefore, $h_\epsilon(\bx) = g_\epsilon(\bx)$. Using \eqref{indapprox_eq1}, it can now be verified that both $g_\epsilon$ and $h_\epsilon$ are Lipschitz continuous with Lipschitz constant~${1}/{\epsilon}$.
	\end{proof}

	The subsequent results will use assumptions \ref{assumption1} and \ref{assumption2}. Now, we state the uniform-over-$p$ analogue of the Portmanteau theorem (see Theorem 2.1 in \cite{billingsley2013convergence})~below. 
	
	\begin{stheorem}[Uniform-over-$p$ Portmanteau theorem]\label{portmanteau}
		Let $\mu_{n, p}$ and $\mu_p$ be probability measures on $(\mathbb{R}^d, \mathcal{R}^d)$ with corresponding distribution functions $F_{n, p}$ and $F_p$, where $\mathcal{R}^d$ is the Borel sigma field on $\mathbb{R}^d$, and the indices $n, p \in \mathbb{N}$. Suppose that the collection $\{F_p\}_{p \in \mathbb{N}}$ satisfies assumptions \ref{assumption1} and \ref{assumption2}. Then, the following are equivalent:
		\begin{enumerate}[label = (\roman*), ref = (\roman*)]
			\item \label{p1} $ \mu_{n, p} \Longrightarrow \mu_p $ uniformly-over-$p$.
			
			\item \label{p2} For every bounded and uniformly continuous function $ f : \mathbb{R}^d \to \mathbb{R} $,
			\begin{align*}
				\lim\limits_{n \to \infty} \sup_p \left| \int f \mathrm{d} \mu_{n, p} - \int f \mathrm{d} \mu_p \right| = 0 .
			\end{align*}
			
			\item \label{p3} For every continuous function $ f : \mathbb{R}^d \to \mathbb{R} $, which is zero outside of a compact set,
			\begin{align*}
				\lim\limits_{n \to \infty} \sup_p \left| \int f \mathrm{d} \mu_{n, p} - \int f \mathrm{d} \mu_p \right| = 0 .
			\end{align*}
			
			\item \label{p4} For every bounded and Lipschitz continuous function $ f : \mathbb{R}^d \to \mathbb{R} $,
			\begin{align*}
				\lim\limits_{n \to \infty} \sup_p \left| \int f \mathrm{d} \mu_{n, p} - \int f \mathrm{d} \mu_p \right| = 0 .
			\end{align*}

			\item \label{p5} For every Borel set $ A $ satisfying $\lim_{\delta\to 0}\sup_p\mu_p((\partial{A})^\delta) = 0$,
			\begin{align*}
				\lim\limits_{n \to \infty} \sup_p | \mu_{n, p}( A ) - \mu_p( A ) | = 0 .
			\end{align*}
			
			\item \label{p6} For $ \bx $ being a equicontinuity point of $\{ F_p\}_{p \in \mathbb{N}}$,
			\begin{align*}
				\lim\limits_{n \to \infty} \sup_p | F_{n, p}( \bx ) - F_p( \bx ) | = 0 .
			\end{align*}

		\end{enumerate}
	\end{stheorem}
	
	\begin{proof}[Proof of \autoref{portmanteau}]
		\ref{p1} $\implies$ \ref{p2}, \ref{p1} $\implies$ \ref{p3}, and \ref{p1} $\implies$ \ref{p4} are obvious from \autoref{definition1}. We shall prove that each of \ref{p2}, \ref{p3}, and \ref{p4} implies \ref{p5} and that \ref{p5} $\implies$ \ref{p6} $\implies$ \ref{p1}.
		
		\ref{p2} $\implies$ \ref{p5}:
		Fix any $A \in \mathbb{R}^d$ such that $\lim_{\delta \to 0}\sup_p\mu_p((\partial{A})^\delta) = 0$. Fix $\epsilon > 0$. Choose a $\delta > 0$ such that $\sup_p\mu_p((\partial{A})^\delta) < \epsilon$. Consider the functions $g_\delta$ and $h_\delta$ as defined in \autoref{indapprox}.  Both functions are Lipschitz continuous and hence, uniformly continuous. Therefore, there exists $N \in \mathbb{N}$ such that for any $n \geq N$ and for all $p \in \mathbb{N}$,
		$$\int g_\delta d\mu_{p} - \epsilon < \int g_\delta d\mu_{n,p} < \int g_\delta d\mu_{p} + \epsilon \quad \text{and}\quad  \int h_\delta d\mu_{p} - \epsilon < \int h_\delta d\mu_{n,p} < \int h_\delta d\mu_{p} + \epsilon.$$
		Now, $g_\delta \leq \mathbb{I}_A \leq h_\delta$. For all $n \geq N$ and $p \in \mathbb{N}$, we get
		$$\int g_\delta d\mu_p \leq \mu_{p}(A) \leq \int h_\delta d\mu_{p},$$
		$$\int g_\delta d\mu_{p} - \epsilon < \int g_\delta d\mu_{n,p} \leq \mu_{n,p}(A) \leq \int h_\delta d\mu_{n,p} < \int h_\delta d\mu_{p} + \epsilon.$$
		This yields
		$$|\mu_{n,p}(A) - \mu_p(A)| \leq \int (h_\delta - g_\delta) d\mu_p + 2\epsilon.$$
		Now, using part \ref{indapprox_pt2} of \autoref{indapprox} we have
		$$\int (h_\delta - g_\delta) d\mu_p \leq \int_{\{\bx: dist(x, \partial A) < \delta\}} d\mu_p \leq \sup_p \mu_p((\partial{A})^\delta) < \epsilon$$
		for all $p \in \mathbb{N}$. Thus, we have
		$$\sup_p |\mu_{n,p}(A) - \mu_p(A)| \leq 3\epsilon$$
		for all $n \geq N$. Hence, $$\lim_{n \to \infty}\sup_p |\mu_{n,p}(A) - \mu_p(A)| = 0.$$
		
		\ref{p3} $\implies$ \ref{p5}: By \autoref{indapprox}, the functions $g_\delta$ and $h_\delta$ are zero outside a compact set. Hence, the proof of \ref{p2} $\implies$ \ref{p5} applies without any modification. 
		
		\ref{p3} $\implies$ \ref{p5}: The functions $g_\delta$ and $h_\delta$ are Lipschitz continuous as well. Hence, the earlier proof carries over verbatim. 
		
		\ref{p5} $\implies$ \ref{p6}:
		Let $\bx = (x_1, \ldots, x_d)^\top$ be an equicontinuity point of $F_p$. Consider $A = (-\infty, x_1] \times \cdots \times (-\infty, x_d]$. Fix any $\epsilon > 0$. There exists $\delta_0 > 0$ such that for any $\by$ with $\|\bx -\by \| < \delta_0$, and we get $|F_p(\bx) - F_p(\by)| < \epsilon$ for any $p \in \mathbb{N}$. Let $0 <\delta < \frac{\delta_0}{2\sqrt{d}}$. Since there are only countably many non-equicontinuity points for the collection $\{F_p\}_{p\in\mathbb{N}}$, we can find two equicontinuity points $\bv = (v_1,\ldots, v_d)^\top$ and $\bu = (u_1,\ldots, u_d)^\top$ such that $u_i = x_i -r_1$ and $v_i = x_i + r_2$ for all $1 \leq i \leq d$ and $\frac{\delta_0}{2\sqrt{d}} < r_1, r_2 < \frac{\delta_0}{\sqrt{d}}$. Thus, we have $\delta_0 / 2 < \|\bx - \bu\| < \delta_0$ and $\delta_0 / 2 < \|\bx - \bv\| < \delta_0$. Let $\bz = (z_1, \ldots,z_d)^\top \in (\partial{A})^\delta$. Observe that if $\bz \in A^c$, then $dist(\bz, \partial A) = \sqrt{\sum_i [(z_i -x_i)_{+}]^2}$ where $(z_i -x_i)_{+} = (z_i -x_i)$ if $(z_i -x_i) \geq 0$ and $(z_i -x_i)_{+} = 0$ otherwise. If $\bz \in A$, then $dist(\bz, \partial A) = \min_i (x_i - z_i)$ . Using these observations we conclude that $dist(\bz, \partial A) < \delta$ implies $z_i < x_i + \delta$ for all $1 \leq i \leq d$ and $z_i > x_i -\delta$ for some $i$, which further implies $\bz \in (-\infty, x_1+ \delta] \times \cdots \times (-\infty, x_d+\delta]$ and $\bz \notin (-\infty, x_1-\delta] \times \cdots \times(-\infty, x_d-\delta]$. So, 
		\begin{align*}
			\mu_p((\partial{A})^\delta) &\leq  \mu_p(\{\bz: \bz \in (-\infty, x_1+ \delta] \times \cdots \times(-\infty, x_d+\delta], \bz \notin (-\infty, x_1-\delta] \times \cdots \times(-\infty, x_d-\delta]\})\\
			&\leq  \mu_p(\{\bz: \bz \in (-\infty, v_1] \times \cdots \times(-\infty, v_d], \bz \notin (-\infty, u_1] \times \cdots \times(-\infty, u_d]\})\\
			&= F_p(\bv) - F_p(\bu) \\
			&= (F_p(\bv) - F_p(\bx)) + (F_p(\bx) - F_p(\bu))& \\
			&< 2\epsilon \text{ for all } p \in \mathbb{N} \text{ and } 0 < \delta < \delta_0.
		\end{align*}
		Therefore, $\lim_{\delta \to 0} \sup_p\mu_p((\partial{A})^\delta) = 0$.
		By \ref{p5},
		$\lim_{n \to \infty} \sup_p |\mu_{n,p}(A) - \mu_p(A)| = 0$ which implies $\lim_{n \to \infty} \sup_p |F_{n,p}(x) - F_p(x)| = 0.$ This proves our claim.
		
		Now, we show that \ref{p6} $\implies$ \ref{p3}. After this, we shall prove that \ref{p6} and \ref{p3} together imply \ref{p1}. This will prove \ref{p6} $\implies$ \ref{p1}. 
		
		\ref{p6} $\implies$ \ref{p3}:
		Under assumption \ref{assumption2}, $A_i = \{u \in \mathbb{R}: (x_1, \ldots,x_d)^\top \in \mathbb{R}^d \setminus C, x_i =u\}$ is at most countable for $i =1, \ldots,d$. Hence, the set
		\begin{equation}
			B = \cup_{i=1}^d A_i = \cup_{i=1}^d\{u \in \mathbb{R}: (x_1, \ldots,x_d)^\top \in \mathbb{R}^d \setminus C, x_i =u\}
			\label{portmanteau_eq16}
		\end{equation}
		is at most countable. Clearly, any $\bx = ( x_1, \ldots, x_d )^\top$ such that $x_i \notin B$ for every $i$ is an equicontinuity point of $\{F_p\}_{p\in \mathbb{N}}$. 
		
		Next, let $\ba = ( a_1, \ldots, a_d )^\top$ and $\bb = ( b_1, \ldots, b_d )^\top$ be such that $a_i \notin B$, $b_i \notin B$ and $a_i < b_i$ for every $i = 1,\ldots,d$. Let $A = \{ \bx = ( x_1, \ldots, x_d )^\top \in \mathbb{R}^d : a_i < x_i \le b_i \text{ for } i = 1, \ldots, d \}$. We proceed to show that under \ref{p6}, $\sup_p | \mu_{n, p}( A ) - \mu_p( A ) | \to 0$ as $n \to \infty$. Let $\bc = ( c_1, \ldots, c_d )^\top$ be such that each $c_i$ is either $a_i$ or $b_i$, and let $n( \bc, \ba )$ denote the number of $a_i$'s in the components of the vector $\bc$. So, there are $2^d$ possible values of the vector $\bc$ depending on the values of $\ba$ and~$\bb$. Let $I_{r, i} : \mathbb{R}^d \to \mathbb{R}$ be defined as $I_{r, i}( \bx ) = \mathbb{I}\{\bx = ( x_1, \ldots, x_d )^\top : x_i \le r \}$ for $i = 1, \ldots, d$ and $r \in \mathbb{R}$. Clearly, for any probability measure $\nu$ on $\mathbb{R}^d$, $\nu( A ) = \int \prod_{i = 1}^d \{ I_{b_i, i}( \bx ) - I_{a_i, i}( \bx ) \} \mathrm{d} \nu( \bx )$. Also, $\prod_{i = 1}^d \{ I_{b_i, i}( \bx ) - I_{a_i, i}( \bx ) \} = \sum_{\bc} (-1)^{ n( \bc, \ba ) } \prod_{i = 1}^d I_{c_i, i}( \bx )$, where the sum is over all $2^d$ values of~$\bc$. Hence, $\nu( A ) = \sum_{\bc} (-1)^{ n( \bc, \ba ) } \int \prod_{i = 1}^d I_{c_i, i}( \bx ) \mathrm{d} \nu( \bx ) = \sum_{\bc} (-1)^{ n( \bc, \ba ) } \nu( B_\bc )$, where $B_\bc = \{ \bx = ( x_1, \ldots, x_d )^\top \in \mathbb{R}^d : x_i \le c_i \text{ for } i = 1, \ldots, d \}$. Since $\mu_{n, p}( B_\bc ) = F_{n, p}( \bc )$ and $\mu_p( B_\bc ) = F_p( \bc )$, we get that $\mu_{n, p}( A ) = \sum_{\bc} (-1)^{ n( \bc, \ba ) } F_{n, p}( \bc )$ and $\mu_p( A ) = \sum_{\bc} (-1)^{ n( \bc, \ba ) } F_p( \bc )$, which from \ref{p6}~yields
		\begin{equation}
			\sup_p | \mu_{n, p}( A ) - \mu_p( A ) | \le \sum_{\bc} (-1)^{ n( \bc, \ba ) } \sup_p \left| F_{n, p}( \bc ) - F_p( \bc ) \right| \to 0 \text{ as } n \to \infty.
			\label{portmanteau_eq17}
		\end{equation}
		Let $f : \mathbb{R}^d \to \mathbb{R}$ be a continuous function which is zero outside the compact set $K$. The restriction of $f$ onto $K$ is a continuous function defined on a compact set and hence is bounded and uniformly continuous. From this it is easy to show that $f$ is uniformly continuous on $\mathbb{R}^d$. Fix $\epsilon > 0$. Then there is a $\delta > 0$ such that $\| \bx - \by \| < \delta$ implies $| f( \bx ) - f( \by ) | <~\epsilon$. We now construct a finite collection of disjoint hyper-rectangles covering $K$ such that for every vertex $\bx = ( x_1, \ldots, x_d )^\top$ of any of the hyper-rectangles, we have $x_i \notin B$ for $i = 1, \ldots, d$, and hence, every such vertex is an equicontinuity point of $\{F_p\}_{p \in \mathbb{N}}$, and also, for points $\bx$ and $\by$ lying in a hyper-rectangle~$\| \bx - \by \| < \delta$. 
		
		The construction of hyper-rectangles is as follows. Since $B$ given in \eqref{portmanteau_eq16} is at most countable, $B^c$ is dense in $\mathbb{R}$. Because $K$ is compact and $B^c$ is dense in $\mathbb{R}$, there is $s > 0$ such that $s \in B^c$ and $K$ is a subset of the hyper-rectangle $( -s, s ]^d$. Take $s_1 = -s$, and given $s_k < s$ with $s - s_k \ge \delta / \sqrt{d}$, choose $s_{k + 1} > s_k$ such that $s_{k + 1} \in B^c$, $s_{k + 1} \le s$ and $\delta / 2 \sqrt{d} < s_{k + 1} - s_k < \delta / \sqrt{d}$. If $s - s_k < \delta / \sqrt{d}$, take $s_{k + 1} = s$. Such a sequence of $s_k$'s can be chosen because $B^c$ is dense. Also, by the method of construction, the set $\{ s_k : k \in \mathbb{N}\}$ is finite and can have at most $ ( 4 s \sqrt{d} / \delta ) + 1$ many elements. Let $q$ be the number of elements in $\{ s_k : k \in \mathbb{N}\}$. Then, the collection of intervals $\{ ( s_k, s_{k + 1} ] : k = 1, \ldots, q - 1 \}$ partitions $( -s, s ]$ and $s_i \notin B$ for $i = 1, \ldots, q$. Consider the corresponding partition of $( -s, s ]^d$ formed by hyper-rectangles with vertices $\bb = ( b_1, \ldots, b_d )^\top$ such that $b_i = s_j$ for some $j \in \{ 1, \ldots, q \}$ and $i = 1, \ldots, d$. Clearly, this partition has at most $(q - 1)^d$ many hyper-rectangles, which are formed by taking products of the intervals $\{ ( s_k, s_{k + 1} ] \}$. Now, from this partition of $( -s, s ]^d$, we remove those hyper-rectangles $A$ such that $A \cap K = \emptyset$. Let $m$ hyper-rectangles remain in the collection, and we denote them as $A_1, \ldots, A_m$. The collection $\{ A_1, \ldots, A_m \}$ of disjoint hyper-rectangles cover $K$ and for every vertex $\bb = ( b_1, \ldots, b_d )^\top$ of any of the hyper-rectangles, we have $b_i \notin B$ for $i = 1, \ldots, d$. If $\bx = ( x_1, \ldots, x_d )^\top$ and $\by = ( y_1, \ldots, y_d )^\top$ lie in a hyper-rectangle, then $| x_i - y_i | < \delta / \sqrt{d}$ for every $1 \leq i \leq d$ and hence, $\| \bx - \by \| < \delta$. Each $A_k$ is of the form $A_k = \prod_{i = 1}^d ( s_{k_i}, s_{k_i + 1} ]$, and define $\bz_k = ( z_{k, 1}, \ldots, z_{k, d} )^\top$ such that $z_{k, i} = ( s_{k_i} + s_{k_i + 1} ) / 2$ for $i = 1, \ldots, d$, i.e., $\bz_k$ is the mid-point of the hyper-rectangle $A_k$. In addition, from \eqref{portmanteau_eq17} we~get
		\begin{align}
			\sup_p | \mu_{n, p}( A_k ) - \mu_p( A_k ) | \to 0 \text{ as } n \to \infty \text{ for } k = 1, \ldots, m .
			\label{portmanteau_eq18}
		\end{align}
		
		Now, corresponding to $f$, we define another function $f_1 : \mathbb{R}^d \to \mathbb{R}$ in the following way. For any $\bx$, if $\bx \notin \cup_{i = 1}^m A_i$, then define $f_1( \bx ) = 0$. If $\bx \in \cup_{i = 1}^m A_i$, then there is exactly one hyper-rectangle $A_k$ which contains $\bx$, because $A_i$'s are disjoint, and define $f_1( \bx )$ for such an $\bx$ as $f_1( \bx ) = f( \bz_k )$, where $\bz_k$ is the mid-point of $A_k$. Note that for $\bx \notin \cup_{i = 1}^m A_i$, $f( \bx ) = f_1( \bx ) = 0$. For $\bx \in A_i$ for any $A_i$, we get $| f( \bx ) - f_1( \bx ) | = | f( \bx ) - f( \bz_i ) | < \epsilon$, since $\| \bx - \bz \| < \delta$. Therefore,
		\begin{align}
			\sup_\bx | f( \bx ) - f_1( \bx ) |
			& = \max\left\{ \sup_{\bx \in \cup_{i = 1}^m A_i} | f( \bx ) - f_1( \bx ) |, \sup_{\bx \notin \cup_{i = 1}^m A_i} | f( \bx ) - f_1( \bx ) | \right\} \nonumber\\
			& = \sup\left\{ | f( \bx ) - f_1( \bx ) | : \bx \in \cup_{i = 1}^m A_i \right\} \nonumber\\
			& = \max\left[ \sup\left\{ | f( \bx ) - f_1( \bx ) | : \bx \in A_i,\; i = 1, \ldots, m \right\} \right]
			< \epsilon .
			\label{portmanteau_eq19}
		\end{align}
		On the other hand, we have
		\begin{align*}
			\int f_1 \mathrm{d} \mu_{n, p} = \sum_{k = 1}^m f( \bz_k ) \mu_{n, p}( A_k ) \;\text{and}\;
			\int f_1 \mathrm{d} \mu_p = \sum_{k = 1}^m f( \bz_k ) \mu_p( A_k ) ,
		\end{align*}
		which from equation \eqref{portmanteau_eq18} implies that
		\begin{align}
			\sup_p \left| \int f_1 \mathrm{d} \mu_{n, p} - \int f_1 \mathrm{d} \mu_p \right|
			\le \sum_{k = 1}^m \left| f( \bz_k ) \right| \sup_p \left| \mu_{n, p}( A_k ) - \mu_p( A_k ) \right| \to 0 \mbox{ as } n \to \infty.
			\label{portmanteau_eq20}
		\end{align}
		Using \eqref{portmanteau_eq19} and \eqref{portmanteau_eq20}, we get
		\begin{align*}
			\sup_p \left| \int f \mathrm{d} \mu_{n, p} - \int f \mathrm{d} \mu_p \right|
			& \le \sup_p \left| \int \{ f( \bx ) - f_1( \bx ) \} \mathrm{d} \mu_{n, p}( \bx ) \right| + \sup_p \left| \int \{ f( \bx ) - f_1( \bx ) \} \mathrm{d} \mu_p( \bx ) \right| \\
			& \quad + \sup_p \left| \int f_1 \mathrm{d} \mu_{n, p} - \int f_1 \mathrm{d} \mu_p \right| \\
			& \le 2 \sup_\bx\left| f( \bx ) - f_1( \bx ) \right| 
			+ \sup_p \left| \int f_1 \mathrm{d} \mu_{n, p} - \int f_1 \mathrm{d} \mu_p \right| \\
			& < 3 \epsilon
		\end{align*}
		for all sufficiently large $n$. This establishes \ref{p6} $\implies$ \ref{p3} since $\epsilon > 0$ is arbitrary. 
		
		\ref{p6} $\implies$ \ref{p1}: We shall show \ref{p6} and \ref{p3} together implies \ref{p1}. Fix $\epsilon > 0$. Since $\{\mu_p\}_{p\in \mathbb{N}}$ is assumed to be tight, there exists a compact set $K$ such that $\mu_p(K^c) < \epsilon$ for all $p$. Since the number of non-equicontinuity points for the collection $\{F_p\}$ is countable, by an argument similar to the one that leads to (\ref{portmanteau_eq17}), we can find a rectangular set $R$ such that $ R \supseteq K$ and $\lim_{n\to\infty} \sup_p |\mu_{n,p}(R) - \mu_p(R)| = 0$ which implies $\lim_{n\to\infty} \sup_p |\mu_{n,p}(R^c) - \mu_p(R^c)| = 0.$ There exists $N$ such that for $n \geq N$ and for any $p \in \mathbb{N}$, we have
		\begin{equation}
			\mu_{n,p}(R^c) < \mu_{p}(R^c) + \epsilon < \mu_p(K^c) + \epsilon < 2\epsilon.
			\label{portmanteau_eq21}
		\end{equation}
		Let $f$ be any bounded continuous function, and $M > 0$ be an upper bound for $f$. Using \autoref{indapprox}, we can obtain a continuous function $0 \leq h \leq 1$ such that $h$ is $1$ on $R$ and is zero outside of a compact set. Hence, $f h$ is a continuous function that is zero outside of a compact set. Further, there exists $N_1$ such that for $n \geq N_1$ and any $p \in \mathbb{N}$, we get
		$$|\int fh d\mu_{n,p} - \int fh d\mu_p| < \epsilon.$$
		Thus, for any $p \in \mathbb{N}$ and $n > \max(N, N_1)$, we have
		\begin{align*}
			|\int f d\mu_{n,p} - \int f d\mu_p| &\leq  |\int fh d\mu_{n,p} - \int fh d\mu_p| + |\int f(1-h) d\mu_{n,p}| + |\int f(1-h) d\mu_p|\\
			&< \epsilon + |\int_{R^c} f(1-h) d\mu_{n,p}| + |\int_{R^c} f(1-h) d\mu_p| \\
			&\leq \epsilon + M\mu_{n,p}(R^c)+ M \mu_{p}(R^c)\\
			&\leq (3M + 1)\epsilon.
		\end{align*}
		This shows that
		$\lim_{n\to \infty}\sup_p|\int f d\mu_{n,p} - \int f d\mu_p| = 0$, which completes the proof of this~theorem.
	\end{proof}
	
	Next, we state and prove a uniform-over-$p$ version of Slutsky's theorem. First, we need the following lemma.
	
	\begin{slemma}\label{lemmaSlutsky}
		Let $\{ X_{n, p} \}_{p \in \mathbb{N}}$, $\{ Y_{n, p} \}_{p \in \mathbb{N}}$ and $\{ X_p \}_{p \in \mathbb{N}}$ be real-valued random variables and $\{ c_p \}_{p \in \mathbb{N}}$ be a bounded sequence of real numbers with countable closure. Then, the collection of measures assigning point mass $1$ to $\{ c_p \}_{p \in \mathbb{N}}$ satisfies the assumptions \ref{assumption1} and \ref{assumption2}. Moreover, if $\{X_p\}_{p \in \mathbb{N}}$ satisfies these assumptions, so does the collection of vectors $\{\left( X_p, c_p \right)^\top\}_{p \in \mathbb{N}}$. Finally, if $X_{n, p} \Longrightarrow X_p$ uniformly-over-$p$ and $Y_{n, p} \stackrel{\P}{\longrightarrow} c_p$ uniformly-over-$p$, then $\left( X_{n, p}, Y_{n, p} \right)^\top \Longrightarrow \left( X_p, c_p \right)^\top$ uniformly-over-$p$.
	\end{slemma}
	
	\begin{proof}[Proof of \autoref{lemmaSlutsky}]
		Let $\delta_x$ denote the probability measure assigning point mass $1$ at $x$. Since $\{c_p\}$ is bounded, the collection $\{\delta_{c_p}\}$ is tight. Let $x \in \mathbb{R}$ be any point outside the closure of $\{c_p\}$. We can find a open interval containing $x$ which contains none of the points in the sequence $\{c_p\}$ and hence, the distribution function of $\delta_{c_p}$ is constant over this interval for any $p$. Thus $x$ is a point of equicontinuity for the collection of functions $\{F_{\delta_{c_p}}\}_{p\in\mathbb{N}}$. Here, $F_Y$ denote the cumulative distribution function of the random variable $Y$. Since the closure of $\{c_p\}$ is at most countable, $\{\delta_{c_p}\}$ satisfies assumptions \ref{assumption1} and \ref{assumption2}.
		
		Now, suppose that $\{F_{X_p}\}_{p\in\mathbb{N}}$ is equicontinuous at $x$ and $\{F_{\delta_{c_p}}\}_{p\in\mathbb{N}}$ is equicontinuous at $y$. For any $\epsilon >0$, we can get $\delta_1, \delta_2 > 0$ such that if $|x - z_1| < \delta_1$ then $\sup_p|F_{X_p}(x) - F_{X_p}(z_1)| < \epsilon$ and if $|x - z_2| < \delta_2$ then $\sup_p|F_{\delta_{c_p}}(x) - F_{\delta_{c_p}}(z_2)| < \epsilon$. Take $\delta = \min(\delta_1, \delta_2)$.  Let $(z_1, z_2) \in \mathbb{R}^2$ be such that $\|(z_1,z_2) - (x,y)\| < \delta$. Then, $|z_1 -x| < \delta, |z_2 -y| < \delta$. It is easy to check that
		\begin{align*}
			&\sup_p|\P(X_p \leq x, \delta_{c_p} \leq y) -\P(X_p \leq z_1, \delta_{c_p}\leq z_2)| \\
			&\leq \sup_p \left[\P(z_1 < X \leq x) + \P(z_2 < \delta_{c_p} \leq y) + \P(x < X \leq z_1) + \P(y < \delta_{c_p} \leq z_2) \right] \\
			&\leq \sup_p|F_{X_p}(x) - F_{X_p}(z_1)| + \sup_p|F_{\delta_{c_p}}(x) - F_{\delta_{c_p}}(z_2)| \\
			&\leq 2\epsilon.
		\end{align*}
		Since $\epsilon > 0$ is arbitrary, $\{(X_p, \delta_{c_p})\}$ satisfies assumption \ref{assumption2}. Now, given any $\epsilon > 0$ we can find two closed bounded intervals $K_1, K_2 \subseteq \mathbb{R}$ such that $\sup_p\P(X_p \notin K_1) < \epsilon$ and $\sup_p\P(\delta_{c_p} \notin K_2) < \epsilon$. Taking $K = K_1 \times K_2$, we have $\sup_p\P((X_p, \delta_{c_p}) \notin K) \leq \sup_p\P(X_p \notin K_1) + \sup_p\P(\delta_{c_p} \notin K_2) \leq 2\epsilon$. Thus, $\{(X_p, \delta_{c_p})\}$ satisfies assumption \ref{assumption1}.
		
		Now, let $f : \mathbb{R}^2 \to \mathbb{R}$ be a bounded and uniformly continuous function. Given any $\epsilon > 0$, there exists a $\delta > 0$ such that $\| ( x_1, x_2 ) - ( y_1, y_2 ) \| \le \delta$ implies that $| f( x_1, x_2 ) - f( y_1, y_2 ) | < \epsilon$.
		Therefore, $| c_1 - c_2 | \le \delta$ implies that for any probability measure $\nu$ on $\mathbb{R}^d$, we obtain
		\begin{align}
			& \left| \int f( x, c_1 ) \mathrm{d} \nu( x ) - \int f( x, c_2 ) \mathrm{d} \nu( x ) \right|
			\le \int \left| f( x, c_1 ) - f( x, c_2 ) \right| \mathrm{d} \nu( x )
			< \epsilon .
			\label{lemmaSlutsky_eq1}
		\end{align}
		Since $\{ c_p \}$ is a bounded sequence, we can construct a finite collection $\{ a_1, \ldots, a_m \} \subset \mathbb{R}$ such that for every $c_p$, we have $| c_p - a_k | < \delta$ for at least one $a_k$. Let $b_p$ denote the number $a_k \in \{ a_1, \ldots, a_m \}$ closest to $c_p$. Using \eqref{lemmaSlutsky_eq1}, we get
		\begin{align*}
			\left| \E[ f( X_{n, p}, c_p ) ] - \E[ f( X_{n, p}, b_p ) ] \right|
			+ \left| \E[ f( X_p, c_p ) ] - \E[ f( X_p, b_p ) ] \right|
			< 2 \epsilon.
		\end{align*}
		This implies that for every $p$ and any $n$, we have
		\begin{align}
			\min\{ \left| \E[ f( X_{n, p}, c_p ) ] - \E[ f( X_{n, p}, a_k ) ] \right| + \left| \E[ f( X_p, c_p ) ] - \E[ f( X_p, a_k ) ] \right| : k = 1, \ldots, m \} < 2 \epsilon.
			\label{lemmaSlutsky_eq2}
		\end{align}
		Next, define $g_c : \mathbb{R} \to \mathbb{R}$ by $g_c( x ) = f( x, c )$ for $c \in \mathbb{R}$. Then, $g_c$ is a bounded and uniformly continuous function for every $c$, and from \ref{p2} in \autoref{portmanteau}, we have
		\begin{align*}
			\sup_p \left| \E[ f( X_{n, p}, c ) ] - \E[ f( X_p, c ) ] \right| =
			\sup_p \left| \E[ g_c( X_{n, p} ) ] - \E[ g_c( X_p ) ] \right| \to 0
		\end{align*}
		as $n \to \infty$. Therefore, there is $n_1$ such that for all $n \ge n_1$,
		\begin{align}
			\max\left\{ \sup_p \left| \E[ f( X_{n, p}, a_k ) ] - \E[ f( X_p, a_k ) ] \right| : k = 1, \ldots, m \right\}
			< \epsilon .
			\label{lemmaSlutsky_eq3}
		\end{align}
		On the other hand, for every $p$, we have
		\begin{align*}
			& \left| \E[ f( X_{n, p}, c_p ) ] - \E[ f( X_p, c_p ) ] \right| \\
			& \le \min\{ \left| \E[ f( X_{n, p}, c_p ) ] - \E[ f( X_{n, p}, a_k ) ] \right| 
			+ \left| \E[ f( X_p, c_p ) ] - \E[ f( X_p, a_k ) ] \right| \\
			& \qquad\quad\;
			+ \left| \E[ f( X_{n, p}, a_k ) ] - \E[ f( X_p, a_k ) ] \right|
			: k = 1, \ldots, m \} \\
			& \le \min\{ \left| \E[ f( X_{n, p}, c_p ) ] - \E[ f( X_{n, p}, a_k ) ] \right| 
			+ \left| \E[ f( X_p, c_p ) ] - \E[ f( X_p, a_k ) ] \right| : k = 1, \ldots, m \} \\
			& \quad
			+ \max\left\{ \sup_p \left| \E[ f( X_{n, p}, a_k ) ] - \E[ f( X_p, a_k ) ] \right| : k = 1, \ldots, m \right\} ,
		\end{align*}
		which by an application of \eqref{lemmaSlutsky_eq2} and \eqref{lemmaSlutsky_eq3} together, yields the following
		\begin{align*}
			& \left| \E[ f( X_{n, p}, c_p ) ] - \E[ f( X_p, c_p ) ] \right| < 3 \epsilon
		\end{align*}
		for every $n \ge n_1$ and every $p$. Hence, for all $n \ge n_1$,
		\begin{align}
			\sup_p \left| \E[ f( X_{n, p}, c_p ) ] - \E[ f( X_p, c_p ) ] \right| < 3 \epsilon .
			\label{lemmaSlutsky_eq4}
		\end{align}
		Let $l = \sup\{ | f( x, y ) | : ( x, y ) \in \mathbb{R}^2 \}$.
		Now, $\| ( X_{n, p}, Y_{n, p} ) - ( X_{n, p}, c_p ) \| = | Y_{n, p} - c_p |$. So, $| Y_{n, p} - c_p | \le \delta$ implies $| f( X_{n, p}, Y_{n, p} ) - f( X_{n, p}, c_p ) | < \epsilon$ and using this fact, we have
		\begin{align*}
			& | \E[ f( X_{n, p}, Y_{n, p} ) ] - \E[ f( X_{n, p}, c_p ) ] | \\
			& \le \left| \E[ \{ f( X_{n, p}, Y_{n, p} ) - f( X_{n, p}, c_p ) \} I( | Y_{n, p} - c_p | \le \delta ) ] \right| \\
			& \quad + \left| \E[ \{ f( X_{n, p}, Y_{n, p} ) - f( X_{n, p}, c_p ) \} I( | Y_{n, p} - c_p | > \delta ) ] \right| \\
			& \le \E\left[ | f( X_{n, p}, Y_{n, p} ) - f( X_{n, p}, c_p ) | I( | Y_{n, p} - c_p | \le \delta ) \right] \\
			& \quad + 2 \sup\{ | f( x, y ) | : ( x, y ) \in \mathbb{R}^2 \} \P[ | Y_{n, p} - c_p | > \delta ] \\
			& \le \epsilon \P\left[ | Y_{n, p} - c_p | \le \delta \right]
			+ 2 l \P[ | Y_{n, p} - c_p | > \delta ] ,
		\end{align*}
		which implies
		\begin{align}
			& \sup_p \left| \E[ f( X_{n, p}, Y_{n, p} ) ] - \E[ f( X_{n, p}, c_p ) ] \right| \le \epsilon \sup_p \P\left[ | Y_{n, p} - c_p | \le \delta \right]
			+ 2 l \sup_p \P[ | Y_{n, p} - c_p | > \delta ] .
			\label{lemmaSlutsky_eq5}
		\end{align}
		Since $Y_{n, p} \stackrel{\P}{\longrightarrow} c_p$ uniformly-over-$p$, from \autoref{definition2} we get that there is $n_2$ such that for all $n \ge n_2$, $\sup_p \P[ | Y_{n, p} - c_p | > \delta ] < \epsilon$. Hence, from \eqref{lemmaSlutsky_eq5}, we have
		\begin{align}
			\sup_p \left| \E[ f( X_{n, p}, Y_{n, p} ) ] - \E[ f( X_{n, p}, c_p ) ] \right|
			\le \epsilon ( 1 + 2 l )
			\label{lemmaSlutsky_eq6}
		\end{align}
		for all $n \ge n_2$. Using \eqref{lemmaSlutsky_eq4} and \eqref{lemmaSlutsky_eq6}, we get that for all $n \ge \max\{ n_1, n_2 \}$,
		\begin{align}
			& \sup_p \left| \E[ f( X_{n, p}, Y_{n, p} ) ] - \E[ f( X_p, c_p ) ] \right| \nonumber\\
			& \le \sup_p \left| \E[ f( X_{n, p}, Y_{n, p} ) ] - \E[ f( X_{n, p}, c_p ) ] \right| 
			+ \sup_p \left| \E[ f( X_{n, p}, c_p ) ] - \E[ f( X_p, c_p ) ] \right| \nonumber\\
			& < \epsilon ( 1 + 2 l ) + 3 \epsilon = 2 \epsilon ( 2 + l ) .
			\label{lemmaSlutsky_eq7}
		\end{align}
		Since $\epsilon > 0$ is arbitrary, \eqref{lemmaSlutsky_eq7} along with \ref{p2} of \autoref{portmanteau} imply that $\left( X_{n, p}, Y_{n, p} \right) \Longrightarrow \left( X_p, c_p \right)$ uniformly-over-$p$.
	\end{proof}
	
	We now state and prove the uniform-over-$p$ version of Slutsky's theorem.
	
	\begin{stheorem}[Uniform-over-$p$ Slutsky's theorem]\label{Slutsky}
		Let $\{ X_{n, p} \}_{p \in \mathbb{N}}$, $\{ Y_{n, p} \}_{p \in \mathbb{N}}$ and $\{ X_p \}_{p \in \mathbb{N}}$ be real-valued random variables and $\{ c_p \}_{p \in \mathbb{N}}$ be a bounded sequence of real numbers with countable closure. Suppose that $\{X_p\}_{p \in \mathbb{N}}$ satisfies assumptions \ref{assumption1} and \ref{assumption2}, $X_{n, p} \Longrightarrow X_p$ uniformly-over-$p$ and $Y_{n, p} \stackrel{\P}{\longrightarrow} c_p$ uniformly-over-$p$. Then,
		\begin{enumerate}[label = (\alph*), ref = (\alph*)]
			\item $X_{n, p} + Y_{n, p} \Longrightarrow X_p + c_p$ uniformly-over-$p$,
			\label{s1}
			
			\item $ X_{n, p} Y_{n, p} \Longrightarrow c_p X_p$ uniformly-over-$p$,
			\label{s2}
			
			\item $ X_{n, p} / Y_{n, p} \Longrightarrow X_p / c_p$ uniformly-over-$p$ if $\{ c_p \}_{p \in \mathbb{N}}$ is bounded away from 0, i.e., $\inf_p | c_p |~>~0$.
			\label{s3}
		\end{enumerate}
	\end{stheorem}
	
	\begin{proof}[Proof of \autoref{Slutsky}]
		From \autoref{lemmaSlutsky}, we get $\left( X_{n, p}, Y_{n, p} \right) \Longrightarrow \left( X_p, c_p \right)$ uniformly-over-$p$.
		Now, to prove \ref{s1}, define $g : \mathbb{R}^2 \to \mathbb{R}$ to be $g( x, y ) = x + y$. Then, from the continuity of $g$ and \autoref{mappingthm}, we get $X_{n, p} + Y_{n, p} \Longrightarrow X_p + c_p$ uniformly-over-$p$.
		
		To prove \ref{s2}, define $g : \mathbb{R}^2 \to \mathbb{R}$ by $g( x, y ) = x y$ and similarly, we get $ X_{n, p} Y_{n, p} \Longrightarrow c_p X_p$ uniformly-over-$p$.
		
		Next, let $l = \inf_p | c_p |$ and define $g : \mathbb{R}^2 \to \mathbb{R}$ by $g( x, y ) = x / y$ if $| y | \ge l / 2$ and $g( x, y ) = 2 x / l$ if $| y | < l / 2$. Clearly, $g$ is a continuous function on $\mathbb{R}^2$. Also, $g( x, y ) \neq x / y$ if and only if $| y | < l / 2$, and in particular $g( x, c_p ) = x / c_p$ for any $x$ and every $p$, since $| c_p | \ge l$ for every $p$. Hence, for any bounded continuous function $f : \mathbb{R} \to \mathbb{R}$, we obtain
		\begin{align*}
			& | \E[ f( X_{n, p} / Y_{n, p} ) ] - \E[ f( X_p / c_p ) ] | \\
			& \le \E[ | f( X_{n, p} / Y_{n, p} ) - f( g( X_{n, p}, Y_{n, p} ) ) | ] \\
			& \quad + \E[ | f( g( X_p, c_p ) ) - f( X_p / c_p ) | ] + | \E[ f( g( X_{n, p}, Y_{n, p} ) ) ] - \E[ f( g( X_p, c_p ) ) ] | \\
			& = \E[ | f( X_{n, p} / Y_{n, p} ) - f( g( X_{n, p}, Y_{n, p} ) ) | I( | Y_{n, p} | < l / 2 ) ] + | \E[ f( g( X_{n, p}, Y_{n, p} ) ) ] - \E[ f( g( X_p, c_p ) ) ] | \\
			& \le 2 \sup\{ | f( x, y ) | : ( x, y ) \in \mathbb{R}^2 \} \P[ | Y_{n, p} | < l / 2 ] + | \E[ f( g( X_{n, p}, Y_{n, p} ) ) ] - \E[ f( g( X_p, c_p ) ) ] | ,
		\end{align*}
		which implies that
		\begin{align}
			& \sup_p | \E[ f( X_{n, p} / Y_{n, p} ) ] - \E[ f( X_p / c_p ) ] | \nonumber\\
			& \le 2 \sup\{ | f( x, y ) | : ( x, y ) \in \mathbb{R}^2 \} \sup_p \P[ | Y_{n, p} | < l / 2 ] + \sup_p | \E[ f( g( X_{n, p}, Y_{n, p} ) ) ] - \E[ f( g( X_p, c_p ) ) ] | .
			\label{Slutsky_eq1}
		\end{align}
		Now, $| Y_{n, p} | < l / 2$ implies that $| Y_{n, p} - c_p | \ge | c_p | - | Y_{n, p} | > l - ( l / 2 ) = l / 2$. Hence, $\P[ | Y_{n, p} | < l / 2 ] \le \P[ | Y_{n, p} - c_p | > l / 2 ]$ for every $n$ and $p$. So, from \eqref{Slutsky_eq1}, we have
		\begin{align}
			& \sup_p | \E[ f( X_{n, p} / Y_{n, p} ) ] - \E[ f( X_p / c_p ) ] | \nonumber\\
			& \le 2 \sup\{ | f( x, y ) | : ( x, y ) \in \mathbb{R}^2 \} \sup_p \P[ | Y_{n, p} - c_p | > l / 2 ] \nonumber\\
			& \quad + \sup_p | \E[ f( g( X_{n, p}, Y_{n, p} ) ) ] - \E[ f( g( X_p, c_p ) ) ] | .
			\label{Slutsky_eq2}
		\end{align}
		Consider any $\epsilon > 0$.
		From \autoref{lemmaSlutsky}, we have $( X_{n, p}, Y_{n, p} ) \Longrightarrow ( X_p, c_p )$ uniformly-over-$p$. Now, the composite function $h : \mathbb{R}^2 \to \mathbb{R}$ defined by $h( x, y ) = f( g( x, y ) )$ is bounded and continuous. Using \autoref{definition1}, we get that there is an integer $n_1$ such that such for all $n \ge n_1$, we have $\sup_p | \E[ f( g( X_{n, p}, Y_{n, p} ) ) ] - \E[ f( g( X_p, c_p ) ) ] | < \epsilon$. Also, from \autoref{definition2}, we get that there is another integer $n_2$ such that for all $n \ge n_2$, we have $\sup_p \P[ | Y_{n, p} - c_p | > l / 2 ] < \epsilon$. Therefore, from \eqref{Slutsky_eq2}, we get that for all $n \ge \max\{ n_1, n_2 \}$,
		\begin{align*}
			& \sup_p | \E[ f( X_{n, p} / Y_{n, p} ) ] - \E[ f( X_p / c_p ) ] | < \epsilon ( 1 + 2 \sup\{ | f( x, y ) | : ( x, y ) \in \mathbb{R}^2 \} ) .
		\end{align*}
		Since $\epsilon > 0$ is arbitrary and $f$ is any bounded continuous function, the above inequality establishes part \ref{s3}.
	\end{proof}
	
	The remainder of this section is devoted to the proofs of the uniform-over-$p$ L\'{e}vy's continuity result given by \autoref{levy}, \autoref{lemma_quantile} and the uniform-over-$p$ central limit theorem (CLT) \autoref{lindeberg} (see Section \ref{sec:theory} of the paper).
	
	\begin{proof}[Proof of \autoref{levy}]
		%The proof follows arguments analogous to those in the proof of Proposition 8.8.1 in \cite{lebanon2012probability}.
		
		Let $ \bX_{n, p} $ and $ \bX_p $ be random vectors with distributions $ \mu_{n, p} $ and $ \mu_p $, respectively.
		Suppose $ \mu_{n, p} \Longrightarrow \mu_p $ uniformly over $ p $. For any fixed $ \bt \in \mathbb{R}^d $, $ \cos( \bt^\top \bx ) $ and $ \sin( \bt^\top \bx ) $ are bounded continuous functions of $ \bx $. Hence, from \autoref{definition1}, we have
		\begin{align*}
			\lim\limits_{n \to \infty}\sup_p \left| \E\left[ \cos\left( \bt^\top \bX_{n, p} \right) \right] - \E\left[ \cos\left( \bt^\top \bX_p \right) \right] \right| = 0  \text{ and } \\
			\lim\limits_{n \to \infty}\sup_p \left| \E\left[ \sin\left( \bt^\top \bX_{n, p} \right) \right] - \E\left[ \sin\left( \bt^\top \bX_p \right) \right] \right| = 0 .
		\end{align*}
		Therefore,
		\begin{align*}
			& \sup_p | \varphi_{n, p}( \bt ) - \varphi_p( \bt ) | \\
			& = \sup_p \left( \left| \E\left[ \cos\left( \bt^\top \bX_{n, p} \right) \right] - \E\left[ \cos\left( \bt^\top \bX_p \right) \right] \right|^2 
			+ \left| \E\left[ \sin\left( \bt^\top \bX_{n, p} \right) \right] - \E\left[ \sin\left( \bt^\top \bX_p \right) \right] \right|^2 \right)^{\frac{1}{2}} \\
			& \le \sup_p \left\{ \left| \E\left[ \cos\left( \bt^\top \bX_{n, p} \right) \right] - \E\left[ \cos\left( \bt^\top \bX_p \right) \right] \right| 
			+ \left| \E\left[ \sin\left( \bt^\top \bX_{n, p} \right) \right] - \E\left[ \sin\left( \bt^\top \bX_p \right) \right] \right| \right\} \\
			& \le \sup_p \left| \E\left[ \cos\left( \bt^\top \bX_{n, p} \right) \right] - \E\left[ \cos\left( \bt^\top \bX_p \right) \right] \right| 
			+ \sup_p \left| \E\left[ \sin\left( \bt^\top \bX_{n, p} \right) \right] - \E\left[ \sin\left( \bt \bX_p \right) \right] \right| \\
			& \to 0 \text{ as } n \to \infty .
		\end{align*}
		
		Next, assume that $ \sup_p | \varphi_{n, p}( \bt ) - \varphi( \bt ) | \to 0 $ as $ n \to \infty $. We shall show that for any continuous function $ g : \mathbb{R}^d \to \mathbb{R} $ (that is zero outside a compact set), we have $ \sup_p \left| \E\left[ g\left( \bX_{n, p} \right) \right] - \right.$ $\left. \E\left[ g\left( \bX_p \right) \right] \right| \to 0 $ as $ n \to \infty $. Define $ g_+( \bx ) = g( \bx ) \mathbb{I}\{g( \bx ) \ge 0\} $ and $ g_-( \bx ) = | g( \bx ) | \mathbb{I}\{ g( \bx ) < 0 \}$.
		Consider the following:
		\begin{align*}
			\left| \E\left[ g\left( \bX_{n, p} \right) \right] - \E\left[ g\left( \bX_p \right) \right] \right|
			& \le \left| \E\left[ g_+\left( \bX_{n, p} \right) \right] - \E\left[ g_+\left( \bX_p \right) \right] \right|
			+ \left| \E\left[ g_-\left( \bX_{n, p} \right) \right] - \E\left[ g_-\left( \bX_p \right) \right] \right|
		\end{align*}
		which implies
		\begin{align*}
			&\sup_p \left| \E\left[ g\left( \bX_{n, p} \right) \right] - \E\left[ g\left( \bX_p \right) \right] \right|\\
			& \le \sup_p \left| \E\left[ g_+\left( \bX_{n, p} \right) \right] - \E\left[ g_+\left( \bX_p \right) \right] \right| 
			+ \sup_p \left| \E\left[ g_-\left( \bX_{n, p} \right) \right] - \E\left[ g_-\left( \bX_p \right) \right] \right| .
		\end{align*}
		Since $ g_+( \bx ) $ and $ g_-( \bx ) $ are non-negative functions of $ \bx $, it is enough to consider $ g $ as a non-negative continuous function which is zero outside a compact set. If $ g $ is identically 0, the assertion $ \sup_p \left| \E\left[ g\left( \bX_{n, p} \right) \right] - \E\left[ g\left( \bX_p \right) \right] \right| \to 0 $ as $ n \to \infty $ is vacuously satisfied. So, without loss of generality, we assume that $ g $ is a non-negative and non-zero continuous function which is zero outside a compact set.
		
		Fix $\epsilon > 0$. As $g$ is continuous and zero outside a compact set, it is uniformly continuous and we can find $ \delta > 0 $ such that $\| \bx - \by \| \le \delta $ implies $ | g( \bx ) - g( \by ) | < \epsilon $.
		Let $ \bZ $ be a $d$-dimensional random vector independent of $ \bX_{n, p} $ and $ \bX_p $ for all $ n $ and $ p $, which follows the $N_d( \mathbf{0}, \sigma^2 \mathbf{I}_d ) $ distribution. Here, $ \sigma^2 $ is small enough so that $ P[ \| \bZ \| > \delta ] < \epsilon $ holds. We~have
		\begin{align}
			\left| \E\left[ g\left( \bX_{n, p} \right) \right] - \E\left[ g\left( \bX_p \right) \right] \right|
			\le & \left| \E\left[ g\left( \bX_{n, p} \right) \right] - \E\left[ g\left( \bX_{n, p} + \bZ \right) \right] \right| \nonumber\\
			& + \left| \E\left[ g\left( \bX_{n, p} + \bZ \right) \right] - \E\left[ g\left( \bX_p + \bZ \right) \right] \right| \nonumber + \left| \E\left[ g\left( \bX_p + \bZ \right) \right] - \E\left[ g\left( \bX_p \right) \right] \right| . \label{eq:levy1}
		\end{align}
		Now,
		\begin{align}
			&\sup_p \left| \E\left[ g\left( \bX_{n, p} \right) \right] - \E\left[ g\left( \bX_{n, p} + \bZ \right) \right] \right| \\
			& \le \sup_p \left| \E\left[ \left\{ g\left( \bX_{n, p} \right) - g\left( \bX_{n, p} + \bZ \right) \right\} \left\{ I( \| \bZ \| \le \delta ) + I( \| \bZ \| > \delta ) \right\} \right] \right| \nonumber\\
			& \le \sup_p \E\left[ \left| g\left( \bX_{n, p} \right) - g\left( \bX_{n, p} + \bZ \right) \right| I( \| \bZ \| \le \delta ) \right] + \sup_p \E\left[ \left| g\left( \bX_{n, p} \right) - g\left( \bX_{n, p} + \bZ \right) \right| I( \| \bZ \| > \delta ) \right] \nonumber\\
			& < \epsilon + 2 \sup_\bx | g( \bx ) | P[ \| \bZ \| > \delta ]
			= \epsilon \left( 1 + 2 \sup_\bx | g( \bx ) | \right) . \label{eq:levy2}
		\end{align}
		Similarly,
		\begin{align}
			\sup_p \left| \E\left[ g\left( \bX_p \right) \right] - \E\left[ g\left( \bX_p + \bZ \right) \right] \right|
			< \epsilon \left( 1 + 2 \sup_\bx | g( \bx ) | \right) . \label{eq:levy3}
		\end{align}
		Next, using Fubini's theorem, we have
		\begin{align}
			& \E\left[ g\left( \bX_{n, p} + \bZ \right) \right] \nonumber\\
			& = \frac{1}{\left( \sqrt{2 \pi} \sigma \right)^d} \int \int g( \bx + \bz ) \exp\left( - \frac{\bz^\top \bz}{2 \sigma^2} \right) \mathrm{d} \bz\, \mathrm{d} \mu_{n, p}( \bx ) \nonumber\\
			& = \frac{1}{\left( \sqrt{2 \pi} \sigma \right)^d} \int \int g( \bu ) \exp\left( - \frac{( \bu - \bx )^\top ( \bu - \bx )}{2 \sigma^2} \right) \mathrm{d} \bu\, \mathrm{d} \mu_{n, p}( \bx ) \nonumber\\
			& = \frac{1}{\left( \sqrt{2 \pi} \sigma \right)^d} \int \int g( \bu ) \prod_{i = 1}^d \exp\left( - \frac{( u_i - x_i )^2}{2 \sigma^2} \right) \mathrm{d} \bu\, \mathrm{d} \mu_{n, p}( \bx ) \nonumber\\
			& = \frac{1}{\left( \sqrt{2 \pi} \sigma \right)^d} \int \int g( \bu ) \prod_{i = 1}^d \frac{\sigma}{\sqrt{2 \pi}} \int \exp\left( \mathrm{i} t_i ( u_i - x_i ) - \frac{\sigma^2 t_i^2}{2} \right) \mathrm{d} t_i\, \mathrm{d} \bu\, \mathrm{d} \mu_{n, p}( \bx ) \nonumber\\
			& = \frac{1}{( 2 \pi )^d} \int \int \int g( \bu ) \exp\left( \mathrm{i} \bt^\top \bu - \frac{\sigma^2 \bt^\top \bt}{2} \right) e^{- \mathrm{i} \bt^\top \bx} \mathrm{d} \mu_{n, p}( \bx ) \, \mathrm{d} \bt \, \mathrm{d} \bu \nonumber\\
			& = \frac{1}{( 2 \pi )^d} \int \int g( \bu ) \exp\left( \mathrm{i} \bt^\top \bu - \frac{\sigma^2 \bt^\top \bt}{2} \right) \varphi_{n, p}( - \bt ) \mathrm{d} \bt \, \mathrm{d} \bu \nonumber\\
			& = \frac{\int g( \bu ) \mathrm{d} \bu}{( \sqrt{2 \pi} \sigma )^d} \int \int \exp\left( \mathrm{i} \bt^\top \bu \right) \varphi_{n, p}( - \bt ) \frac{g( \bu )}{\int g( \bu ) \mathrm{d} \bu} \left( \frac{\sigma}{\sqrt{2 \pi}} \right)^d \exp\left( - \frac{\sigma^2 \bt^\top \bt}{2} \right) \mathrm{d} \bt \, \mathrm{d} \bu \nonumber\\
			& = \frac{\int g( \bu ) \mathrm{d} \bu}{( \sqrt{2 \pi} \sigma )^d} \E\left[ \exp\left( \mathrm{i} \bT^\top \bU \right) \varphi_{n, p}( - \bT) \right] , \label{eq:levy4}
		\end{align}
		where $ \bT $ and $ \bU $ are independent random vectors, with $ \bT \sim N_d( {\bf 0}, \sigma^{-2} \mathbf{I}_d ) $ and $ \bU $ has the density function $ \left( \int g( \bu ) \mathrm{d} \bu \right)^{-1} g( \bu ) $. Note that the fourth equality is obtained using the expression of the characteristic function for the normal distribution:
		$$\int e^{ibx} \frac{1}{a\sqrt{2\pi}} e^{-\frac{x^2}{2a^2}}dx = e^{-\frac{a^2b^2}{2}}$$
		for any $a > 0, b\in \mathbb{R}$.
		Similarly,
		\begin{align}
			\E\left[ g\left( \bX_p + \bZ \right) \right]
			& =  \frac{\int g( \bu ) \mathrm{d} \bu}{( \sqrt{2 \pi} \sigma )^d} \E\left[ \exp\left( \mathrm{i} \bT^\top \bU \right) \varphi_p( - \bT) \right]. \label{eq:levy5}
		\end{align}
		From \eqref{eq:levy4} and \eqref{eq:levy5}, we have
		\begin{align}
			& \left| \E\left[ g\left( \bX_{n, p} + \bZ \right) \right] - \E\left[ g\left( \bX_p + \bZ \right) \right] \right| \nonumber\\
			& = \frac{\int g( \bu ) \mathrm{d} \bu}{( \sqrt{2 \pi} \sigma )^d} \left| \E\left[ \exp\left( \mathrm{i} \bT^\top \bU \right) \varphi_{n, p}( - \bT ) \right] - \E\left[ \exp\left( \mathrm{i} \bT^\top \bU \right) \varphi_p( - \bT ) \right] \right| \nonumber\\
			& \le \frac{\int g( \bu ) \mathrm{d} \bu}{( \sqrt{2 \pi} \sigma )^d} \E\left[ \left| \exp\left( \mathrm{i} \bT^\top \bU \right) \varphi_{n, p}( - \bT ) - \exp\left( \mathrm{i} \bT^\top \bU \right) \varphi_p( - \bT ) \right| \right] \nonumber\\
			& = \frac{\int g( \bu ) \mathrm{d} \bu}{( \sqrt{2 \pi} \sigma )^d} \E\left[ \left| \exp\left( \mathrm{i} \bT^\top \bU \right) \right| \left| \varphi_{n, p}( - \bT ) - \varphi_p( - \bT ) \right| \right] \nonumber\\
			& = \frac{\int g( \bu ) \mathrm{d} \bu}{( \sqrt{2 \pi} \sigma )^d} \E\left[ \left| \varphi_{n, p}( - \bT ) - \varphi_p( - \bT ) \right| \right] . \label{eq:levy6}
		\end{align}
		Since $ \sup_p | \varphi_{n, p}( \bt ) - \varphi_p( \bt ) | \to 0 $ as $ n \to \infty $ and $ \sup_p | \varphi_{n, p}( \bt ) - \varphi_p( \bt ) | $ is a measurable function of $ \bt $, from \eqref{eq:levy6} we have the following:
		\begin{align*}
			\sup_p \left| \E\left[ g\left( \bX_{n, p} + \bZ \right) \right] - \E\left[ g\left( \bX_p + \bZ \right) \right] \right|
			& \le \frac{\int g( \bu ) \mathrm{d} \bu}{( \sqrt{2 \pi} \sigma )^d} \sup_p \E\left[ \left| \varphi_{n, p}( - \bT ) - \varphi_p( - \bT ) \right| \right] \\
			& \le \frac{\int g( \bu ) \mathrm{d} u}{( \sqrt{2 \pi} \sigma )^d} E\left[ \sup_p \left| \varphi_{n, p}( - \bT ) - \varphi_p( - \bT ) \right| \right] .
		\end{align*}
		By the bounded convergence theorem, $ \E\left[ \sup_p \left| \varphi_{n, p}( - \bT ) - \varphi( - \bT ) \right| \right] \to 0 $ as $ n \to \infty $ and hence, for all sufficiently large $n$, we get
		\begin{align}
			\sup_p \left| \E\left[ g\left( \bX_{n, p} + \bZ \right) \right] - \E\left[ g\left( \bX_p + \bZ \right) \right] \right| < \epsilon .\label{eq:levy7}
		\end{align} 
		From \eqref{eq:levy1}, \eqref{eq:levy2}, \eqref{eq:levy3} and \eqref{eq:levy7}, we have
		\begin{align*}
			\sup_p \left| \E\left[ g\left( \bX_{n, p} \right) \right] - \E\left[ g\left( \bX_p \right) \right] \right|
			\le \epsilon \left( 3 + 4 \sup_\bx | g( \bx ) | \right) .
		\end{align*}
		Since $ \epsilon > 0 $ is arbitrary, we have
		\begin{align*}
			\lim\limits_{n \to \infty} \sup_p \left| \E\left[ g\left( \bX_{n, p} \right) \right] - \E\left[ g\left( \bX_p \right) \right] \right| = 0.
		\end{align*}
		The proof now follows from \autoref{portmanteau}.
	\end{proof}

	\begin{proof}[Proof of \autoref{lemma_quantile}]
		We first show that $\{F_p\}_{p\in\mathbb{N}}$ satisfies assumptions \ref{assumption1} and \ref{assumption2}. The sequence of distributions $F_p$ converges weakly, and hence, it is tight. Since $F$ is continuous, by Polya's theorem, $F_p$ converges to $F$ uniformly. Also, $F$ is uniformly continuous on $\mathbb{R}$. Fix $\epsilon > 0$. There exists $\delta_1 > 0$ such that if $|x -y| < \delta_1$, then $|F(x) - F(y)|\leq \epsilon$. Next, we choose $P$ large enough so that for any $p \geq P$, $\sup_{x}|F_p(x) - F(x)| \leq \epsilon$. Now, for any $p \geq P$, we have $|F_p(x) - F_p(y)| \leq |F_p(x) -F(x)| + |F(x) - F(y)| + |F(y) -F_p(y)| \leq 3\epsilon$ whenever $|x -y| < \delta_1$. Since each $F_p$ is a continuous distribution function, it is uniformly continuous and we can find a $\delta_2 > 0$ such that whenever $|x -y| \leq \delta_2$, we have $|F_p(x) - F_p(y)| < \epsilon$ for any $p \in \{ 1, 2, \ldots, P\}$. Define $\delta = \min(\delta_1, \delta_2)$. Then, $|F_p(x) -F_p(y)| \leq 3\epsilon$ for all $p \in \mathbb{N}$, whenever $|x -y| <\delta$. Thus, $\{F_p\}_{p\in\mathbb{N}}$ is uniformly equicontinuous which immediately implies assumption \ref{assumption2}. Thus, the collection $\{F_p\}_{p\in\mathbb{N}}$ satisfies assumptions \ref{assumption1} and \ref{assumption2}. 
		
		Next, we show that $\sup_{p\in \mathbb{N}}\sup_{x\in\mathbb{R}}|F_{n,p}(x) - F(x)| \to 0$
		as $n \to \infty$. Fix $\epsilon > 0$. By \autoref{portmanteau}, $\sup_p|F_{n,p}(x) - F_p(x)| \to 0$ for all $x \in \mathbb{R}$. But, $F_p(x) \to F(x)$ and hence, $|F_{n,p}(x) - F(x)| \leq |F_{n,p}(x) - F_p(x)| + |F_p(x) - F(x)| \to 0$ as $n \to \infty, p \to \infty$ for any $x \in\mathbb{R}$. Thus, as $n \to \infty, p \to \infty$ together, $F_{n,p}$ converges to $F$ weakly. Using Polya's theorem, we get $\sup_x|F_{n,p}(x) - F(x)| \to 0$ as $n \to \infty, p \to \infty$. Also, $\sup_x|F_p(x) - F(x)| \to 0$ as $p \to \infty$. Thus we can choose $P_0$ and $N_0$ large enough so that if $p \geq P_0$ and $n \geq N_0$, then $|F_{n,p}(x) - F_p(x)| \leq |F_{n,p}(x) - F(x)| + |F_p(x) - F(x)| \leq 2 \epsilon$ for all $x \in \mathbb{R}$. Since each $F_p$ is continuous, we can find $N_0'$ such that for any $p \in \{1,\ldots, P_0\}$, $\sup_x|F_{n,p}(x) - F_p(x)| \leq \epsilon$ whenever $n \geq N_0'$. Thus, for $n \geq \max\{N_0, N_0'\}$, we have $\sup_x|F_{n,p}(x) - F_p(x)| \leq 2\epsilon$ for all $p \in \mathbb{N}$. This proves the assertion.

		Finally, we prove the convergence of quantiles. Each of these functions the functions $Q_p$ and $Q$ is non-decreasing and has only countably many points of discontinuities for any $p \in \mathbb{N}$. We claim that each of the functions is continuous at $\alpha$. We prove this for $Q$, the other cases being similar. Since quantile function is left continuous, we only need to show right continuity. If possible, assume that $\alpha_m$ decreases to $\alpha$ as $m \to \infty$ and $Q(\alpha_m) - Q(\alpha) > \epsilon_0$ for some $\epsilon_0 > 0$ and all $m \in \mathbb{N}$. Since $F$ is continuos everywhere, $\alpha_m = F(Q(\alpha_m)) \geq F(Q(\alpha) + \epsilon_0)$ and taking $m \to \infty$ we get $\alpha \geq F(Q(\alpha) + \epsilon_0) > F(Q(\alpha)) = \alpha$ which is a contradiction. This proves our claim. Now, fix $\epsilon > 0$. Choose $\delta > 0$ such that $0 < \alpha -\delta < \alpha+ \delta < 1$ and $Q(\alpha+\delta) - Q(\alpha - \delta) < \epsilon$ and $Q$ is continuous at $\alpha+\delta$ and $\alpha - \delta$. Using \cite[Lemma 21.2]{vandervaart2000book}, choose $P$ large enough so that for any $p > P$, we have $|Q_p(\alpha \pm \delta) - Q(\alpha \pm \delta)| < \epsilon$. Further, choose $\delta_1, \ldots, \delta_P$ such that $|Q_p(\alpha + \delta_p) - Q_p(\alpha - \delta_p)| < \epsilon$ for $p = 1, \ldots, P$. Set $\delta_0 = \min\{\delta, \delta_1, \ldots, \delta_P\} > 0$. We now have $|Q_p(\alpha + \delta_0) - Q_p(\alpha - \delta_0)| \leq |Q_p(\alpha + \delta_p) - Q_p(\alpha - \delta_p)| < \epsilon$ for $p = 1, \ldots, P$ and
		\begin{align*}
			&|Q_p(\alpha + \delta_0) - Q_p(\alpha - \delta_0)| \\
			&\leq |Q_p(\alpha + \delta) - Q_p(\alpha - \delta)|\\
			&\leq |Q_p(\alpha + \delta) - Q(\alpha + \delta)| + |Q_p(\alpha - \delta) - Q(\alpha - \delta)| +|Q(\alpha + \delta) - Q(\alpha - \delta)|\\
			&< 3\epsilon   
		\end{align*} 
		any for $p > P$. Consider $N$ large enough so that for any $n \geq N$ and any $p$, we have $\sup_x|F_{n,p}(x) - F_p(x)| < \delta_0$. Then, for any $n \geq N$ and any $p$, we get
		$$\alpha  - \delta_0 \leq F_{n,p}(Q_{n,p}(\alpha)) - \delta_0\leq F_{p}(Q_{n,p}(\alpha)) $$
		which gives $Q_p(\alpha - \delta_0)  \leq Q_{n,p}(\alpha)$. Again,
		$$\alpha  \leq F_p(Q_p(\alpha + \delta_0)) -\delta_0 \leq F_{n,p}(Q_p(\alpha + \delta_0)) $$
		which gives $Q_{n,p}(\alpha)  \leq Q_p(\alpha + \delta_0)$.
		Thus, $|Q_{n,p}(\alpha) - Q_p(\alpha)| \leq |Q_p(\alpha + \delta_0) - Q_p(\alpha - \delta_0)| < 3\epsilon$ for any $n \geq N$ and any $p \in \mathbb{N}$. This completes the proof.
	\end{proof}

	\begin{proof}[Proof of \autoref{lindeberg}]
		Fix any $\bt \in \mathbb{R}^d$. For any $\delta > 0$, $p \in \mathbb{N}$ and $1\leq i \leq d$, 
		\begin{align*}
			&\max_{1\leq r \leq k_n}\E[(\be_i^T\bX_{n,p,r})^2] \\
			&=\max_{1\leq r \leq k_n}  \left[\E[(\be_i^T\bX_{n,p,r})^2 \mathbb{I}\{|\be_i^T\bX_{n,p,r}| \leq \sqrt{\delta}\}]+ \E[(\be_i^T\bX_{n,p,r})^2 \mathbb{I}\{|\be_i^T\bX_{n,p,r}|> \sqrt{\delta}\}]\right] \\
			&\leq \delta + \sum_{r=1}^{k_n}\E[(\be_i^T\bX_{n,p,r})^2 \mathbb{I}\{|\be_i^T\bX_{n,p,r}|> \sqrt{\delta}\}]
		\end{align*}
		For any given $\delta > 0$, by condition (c) and the above inequality, $ \sup_p\max_{1\leq r \leq k_n}\E[(\be_i^T\bX_{n,p,r})^2] < 2\delta$ for large enough $n$. Thus
		$\sup_p\max_{1\leq r \leq k_n}\E[(\be_i^T\bX_{n,p,r})^2] \to 0$
		as $n \to \infty$. Hence, for all~$1\leq i,j\leq d$, we have
		\begin{align*}
			&\sup_p\max_{1\leq r\leq k_n}|\Cov(\be_i^T\bX_{n,p,r}, \be_j^T\bX_{n,p,r})|\\
			&\leq \sqrt{\sup_p\max_{1\leq r\leq k_n}\E[(\be_i^T\bX_{n,p,r})^2]} \sqrt{\sup_p\max_{1\leq r\leq k_n}\E[(\be_i^T\bX_{n,p,r})^2]} \to 0   
		\end{align*} 
		as $n \to \infty$. This immediately gives that for any $\mathbf{y} \in \mathbb{R}^d$,
		\begin{equation}
			\sup_p\max_{1\leq r\leq k_n}\mathbf{y}^T \Cov(\bX_{n,p,r})\mathbf{y} \to 0 \label{lindeberg_eq1}
		\end{equation}
		as $n\to \infty$. Since $\be_i^T\Cov(\bS_{n,p})\be_j\to \be_i^T\bSigma_p\be_j = \sigma_{p,ij}$ uniformly-over-$p$ as $n \to \infty$ for $1\leq i,j\leq d$, we have
		\begin{equation}
			\bt^T\Cov(\bS_{n,p})\bt\to \bt^T\bSigma_p\bt
			\label{lindeberg_eq2}
		\end{equation}
		uniformly-over-$p$ as $n \to \infty$. So, $\bSigma_p$ is a symmetric non-negative definite matrix for each~$p$. Further, 
		$$\sup_p|\bt^T\Sigma_p\bt|\leq \sum_{i=1}^d \sum_{j=1}^d |t_it_j| \sup_p|\be_i^T \bSigma_p \be_j| \leq \sum_{i=1}^d \sum_{j=1}^d |t_it_j| \sqrt{\left(\sup_p\be_i^T \bSigma_p \be_i\right)\left(\sup_p \be_j^T \bSigma_p \be_j\right)} < \infty.$$ 
		Therefore, from \eqref{lindeberg_eq1} and \eqref{lindeberg_eq2},
		\begin{align}
			\sum_{r=1}^{k_n}\left(\bt^T \Cov(\bX_{n,p,r})\bt\right)^2  &\leq \left[\max_{1\leq r\leq k_n}\bt^T \Cov(\bX_{n,p,r})\bt\right] \sum_{r=1}^{k_n}\bt^T \Cov(\bX_{n,p,r})\bt \nonumber\\
			&= \left[\max_{1\leq r\leq k_n}\bt^T \Cov(\bX_{n,p,r})\bt\right] \bt^T\Cov(\bS_{n,p})\bt\to 0 \label{lindeberg_eq3}
		\end{align}
		uniformly-over-$p$ as $n\to\infty$. Since $x \leq -\log(1-x) = \sum_{r=1}^\infty \frac{x^r}{r}  \leq x + \frac{x^2}{1-x} \leq x + 2x^2 $ for $0<x \leq \frac{1}{2}$, there exists $N$ such that for $n \geq N$ and any $p \in \mathbb{N}$,
		$$\exp\left(-\frac{1}{2}\bt^T \Cov(\bX_{n,p,r})\bt - \frac{1}{2}\left(\bt^T \Cov(\bX_{n,p,r})\bt\right)^2 \right) \leq 1- \frac{1}{2}\bt^T \Cov(\bX_{n,p,r})\bt \leq \exp\left(-\frac{1}{2}\bt^T \Cov(\bX_{n,p,r})\bt\right).$$
		Using this inequality along with \eqref{lindeberg_eq2} and \eqref{lindeberg_eq3}, we conclude 
		\begin{equation}
			\prod_{r=1}^{k_n} \left(1- \frac{1}{2}\bt^T \Cov(\bX_{n,p,r})\bt \right) \to \exp\left(- \frac{1}{2}\bt^T \Sigma_p\bt\right) \label{lindeberg_eq4} 
		\end{equation}
		uniformly-over-$p$ as $n \to \infty$.
		Let $\varphi_{\bS_{n,p}}(\bt) = \E[\exp(i\bt^T \bS_{n,p})]$ denote the characteristic function of $\bS_{n,p}$. From \eqref{lindeberg_eq1}, we have $\max_{1\leq r \leq k_n} \frac{1}{2} \bt^T \Cov(\bX_{n,p,r})\bt < 1$ for all $p$ and large enough $n$. Using the inequality {\color{red}($26.4_0$)} in \cite[p.~343]{billingsley2008probability} along with the Cauchy-Schwartz (CS) inequality, for any $p$, $1\leq r \leq k_n$ and large enough $n$, we obtain the following:
		\begin{align*}
			&\left|\varphi_{\bX_{n,p,r}}(\bt) - \left(1-\frac{1}{2} \bt^T \Cov(X_{n,p,r})\bt\right)\right|\\
			&= \left|\E\left[\exp(i\bt^T\bX_{n,p,r}) - 1 - i\bt^T\bX_{n,p,r} + \frac{1}{2} (\bt^T\bX_{n,p,r})^2\right]\right| \\
			&\leq \E \left[\min \left( \frac{1}{6} |\bt^T\bX_{n,p,r}|^3, (\bt^T\bX_{n,p,r})^2 \right) \right] \\
			&\leq \frac{1}{6} \E \left[ (\bt^T\bX_{n,p,r})^3 \,\mathbb{I}\{\|\bX_{n,p,r}\| < \delta \sqrt{d} \}\right] + \E \left[(\bt^T\bX_{n,p,r})^2 \,\mathbb{I}\{\|\bX_{n,p,r}\| \geq \delta \sqrt{d}\} \right] \\
			&\leq \frac{1}{6} \E \left[(\bt^T\bX_{n,p,r})^2 \|\bt\| \|\bX_{n,p,r}\| \,\mathbb{I}\{\|\bX_{n,p,r}\| < \delta \sqrt{d} \}\right] + \|\bt\|^2 \E \left[ \left[\sum_{i=1}^d (\be_i^T\bX_{n,p,r})^2\right] \,\mathbb{I}\{\sum_{i=1}^d (\be_i^T\bX_{n,p,r})^2 \geq \delta^2 d\}\right]\\
			&\leq \delta \sqrt{d}\frac{\|\bt\|}{6} \bt^T \Cov(X_{n,p,r})\bt + \|\bt\|^2 \E\left[\sum_{i=1}^d\left[d(\be_i^T\bX_{n,p,r})^2 \,\mathbb{I}\{d (\be_i^T\bX_{n,p,r})^2 \geq \delta^2 d\}\right]\right] \\
			&\leq \delta \sqrt{d} \frac{\|\bt\|}{6} \bt^T \Cov(X_{n,p,r})\bt + d\|\bt\|^2 \sum_{i=1}^d\E\left[(\be_i^T\bX_{n,p,r})^2 \,\mathbb{I}\{ |\be_i^T\bX_{n,p,r}| \geq \delta \}\right].
		\end{align*}
		Using the inequality in ($27.5$) of \cite[p.~358]{billingsley2008probability}, we get
		\begin{align*}
			&\sup_p\left| \varphi_{\bS_{n,p}}(\bt) - \prod_{r=1}^{k_n}\left(1-\frac{1}{2} \bt^T \Cov(\bX_{n,p,r})\bt\right)\right| \\
			&= \sup_p\left| \prod_{r=1}^{k_n}\varphi_{\bX_{n,p,r}}(\bt)  - \prod_{r=1}^{k_n}\left(1-\frac{1}{2} \bt^T \Cov(\bX_{n,p,r})\bt\right)\right|\\
			&\leq \sup_p\sum_{r=1}^{k_n} \left|\varphi_{\bX_{n,p,r}}(\bt) - \left(1-\frac{1}{2} \bt^T \Cov(\bX_{n,p,r})\bt\right)\right|\\
			&\leq \delta \sqrt{d} \frac{\|\bt\|}{6} \sup_p\bt^T \Cov(\bS_{n,p})\bt + d\|\bt\|^2 \sum_{i=1}^d \sup_p \sum_{r=1}^{k_n} \E\left[(\be_i^T\bX_{n,p,r})^2 \,\mathbb{I}\{ |\be_i^T\bX_{n,p,r}| \geq \delta \}\right] \\
			&\leq \delta \sqrt{d} \frac{\|\bt\|}{6} \sum_{i=1}^d\sum_{j=1}^d |t_it_j| \sup_p \left|\be_i^T\Cov(\bS_{n,p})\be_j\right|+ d\|\bt\|^2 \sum_{i=1}^d \sup_p\sum_{r=1}^{k_n} \E\left[(\be_i^T\bX_{n,p,r})^2 \,\mathbb{I}\{ |\be_i^T\bX_{n,p,r}| \geq \delta \}\right]
		\end{align*}
		for large enough $n$. Therefore,
		$$\limsup_{n\to\infty} \sup_{p\in \mathbb{N}}\left| \varphi_{\bS_{n,p}}(\bt) - \prod_{r=1}^{k_n}\left(1-\frac{1}{2} \bt^T \Cov(\bX_{n,p,r})\bt\right)\right| \leq \delta \sqrt{d} \frac{\|\bt\|}{6} \sum_{i=1}^d\sum_{j=1}^d |t_it_j| \sup_p\left|\be_i^T\Sigma_p\be_j\right|.$$
		Since $\sup_p \left|\be_i^T\Sigma_p\be_j\right| < \infty$, by taking $\delta \to 0$ in the inequality given above and using \eqref{lindeberg_eq4}, we~obtain
		\begin{equation}
			\lim_{n\to \infty} \sup_{p\in \mathbb{N}} \left| \varphi_{\bS_{n,p}}(\bt) - \exp\left(-\frac{1}{2} \bt^T \Sigma_p\bt\right)\right| = 0. \label{lindeberg_eq5}
		\end{equation}
		Let $\bZ_p = (Z_{p,1}, \ldots, Z_{p,d})^\top \sim N_d(\mathbf{0}_d, \bSigma_p)$ for all $p \in \mathbb{N}$. We shall show that the collection of probability distributions of of the random variables $\{\bZ_p\}$ satisfies assumptions \ref{assumption1} and \ref{assumption2}. Let $M = \max_{1\leq i \leq d} \sup_{p\in \mathbb{N}} \sigma_{p,ii} < \infty$. By Markov's inequality,
		\begin{align*}
			\P(\|\bZ_p\|> K) \leq  \frac{1}{K^2} \tr(\bSigma_p) \leq \frac{Md }{K^2} \to 0  
		\end{align*}
		as $K \to \infty$. This proves assumption \ref{assumption1}. Let $F_{\bZ_p}$ be the distribution function of $\bZ_p$. Fix any $\bx = (x_1,\ldots,x_d) \in \mathbb{R}^d\setminus\{\mathbf{0}_d\}$. Without loss of generality let $x_1 \neq 0$. Suppose $0 < \delta < |x_1| $ and let  $\by = (y_1,\ldots,y_d)^\top$ be such that $\|\bx - \by\| < \delta$.  Then,
		$$|F_{\bZ_p}(\bx) - F_{\bZ_p}(\by)| \leq \P(\min\{x_1, y_1\} \leq Z_{p,1}\leq \max\{x_1, y_1\}) \leq \P(x_1 -\delta \leq Z_{p,1}\leq x_1 + \delta).$$
		The above quantity is equal to $0$ if $\sigma_{p,11} = 0$. Otherwise, suppose $x_1 > 0$. We have 
		$$\P(x_1 -\delta \leq Z_{p,1}\leq x_1 + \delta) = \int_{x_1-\delta}^{x_1+\delta}\frac{1}{\sigma_{p,11} \sqrt{2\pi}} \exp\left(-\frac{t^2}{2\sigma_{p,11}^2}\right) dt \leq \frac{1}{\sigma_{p,11}} \exp\left(-\frac{(x_1 - \delta)^2}{2\sigma_{p,11}^2}\right) \times 2\delta.$$ 
		It is easy to check that $\frac{1}{\sigma_{p,11}} \exp\left(-\frac{(x_1 - \delta)^2}{2\sigma_{p,11}^2}\right)$ is bounded by some constant which depends only on $x_1$ and hence as $\delta \to 0$, $\P(x_1 -\delta \leq Z_{p,1}\leq x_1 + \delta) \to 0$. The case when $x_1 < 0$ can be handled similarly. Thus, any $\bx \in \mathbb{R}^d\setminus\{\mathbf{0}_d\}$ is a point of equicontinuity for the collection of distributions. This proves assumption \ref{assumption2}. The desired result is now obtained from \eqref{lindeberg_eq5} and \autoref{levy}.
	\end{proof}
	
	\begin{proof}[Proof of \autoref{lyapunov}]
		Consider any random variable $Z$ with $\E[|Z|^{2+\delta}] < \infty$. Using Holder's and Markov's inequalities, we get
		\begin{align*}
			\E[Z^2 \mathbb{I}\{|X| > \epsilon\}] \leq \E[|Z|^{2+\delta}]^{\frac{2}{2+\delta}}\, \E[\mathbb{I}\{|Z| > \epsilon\}]^{\frac{\delta}{2+\delta}} \leq \E[|Z|^{2+\delta}]^{\frac{2}{2+\delta}}\, \frac{1}{\epsilon^{2+\delta}}\E[|Z|^{2+\delta}]^{\frac{\delta}{2+\delta}} \leq \frac{1}{\epsilon^{2+\delta}}\E[|Z|^{2+\delta}]. 
		\end{align*}
		The proof is now immediate.
	\end{proof}
	
	\section{Proofs of the lemmas used in the Appendix} \label{sec:supplement_appendix}
	
	\begin{proof}[Proof of \autoref{fourth_moment}]
		Define $\bB = \bA^\top \bA = (B_{ij})_{m \times m}$. Since $\bB$ is symmetric, we have
		\begin{align*}
			\E[\|\bA\bz\|^4] &= \E[(\bz^\top \bB \bz)^2] = \E[(\sum_{i=1}^m \sum_{j=1}^m z_i B_{ij}z_j)^2] = \sum_{i_1 = 1}^m\sum_{j_1 = 1}^m\sum_{i_2 = 1}^m\sum_{j_2 = 1}^m B_{i_1j_1} B_{i_2j_2} \E[z_{i_1} z_{j_1} z_{i_2} z_{j_2}] \\
			&\leq \sum_{i=1}^m (3+\Delta) B_{ii}^2 + \sum_{i \neq j} B_{ij}^2 + 2\sum_{i < j} B_{ii} B_{jj} + \sum_{i \neq j} B_{ij}B_{ji} \\
			& \leq (3 + \Delta) \sum_{i=1}^m \sum_{j=1}^m B_{ij}^2 + 2\sum_{i \neq j} B_{ij}^2 + 2\sum_{i < j} B_{ii} B_{jj}\\
			& \leq (5 + \Delta) \sum_{i=1}^m \sum_{j=1}^m B_{ij}^2 + \left(\sum_{i=1}^m B_{ii}\right)^2 \\
			&= (5 + \Delta) \tr(\bB^2) + [\tr(\bB)]^2.
		\end{align*}
		Here, $\Delta$ is as defined in \ref{cond3}.
		Further, $\E[\|\bA\bz\|^2] = \tr(\bA\E[\bz\bz^\top] \bA^\top) = \tr(\bA\bA^\top) = \tr(\bB)$ and
		\begin{align*}
			\Var(\|\bA\bz\|^2) \leq (5 + \Delta)\tr(\bB^2) = \tr([ \bA\bA^\top]^2 ). 
		\end{align*}
		This completes the proof.    
	\end{proof}

	\begin{proof}[Proof of \autoref{thm:1_lemma}]
		Fix $ t \in \mathbb{R}, \epsilon > 0$. Define $\bU_q = [\bu_{p,q+1}, \ldots, \bu_{p,p}]$, $\bZ_i^{(q)} = \bU_q^\top \bZ_i$ and $\bSigma_{p,q,i} = \bU_q^\top \bSigma_{p,i} \bU_q$ for $i=1, 2$. Also, let $\bSigma_{n,p,q} = \frac{n_1}{n} \bSigma_{p,q,2} + \frac{n_2}{n} \bSigma_{p,q,1}$ and $\bSigma_{p,q} = \tau \bSigma_{p,q,2} + (1-\tau) \bSigma_{p,q,1}$ for $i =1, 2$ and any $q<p$. Now, we have 
		\begin{align*}
			|\varphi_{n,p}(t) - \varphi_{n,p}^{(q)}(t)|^2 &\leq t^2\E\left[\left(\frac{\sum_{r=q + 1}^p (\xi_{n,p,r}^2 - \lambda_{n,p,r})}{\sqrt{\tr(\bSigma_{n,p}^2)}}\right)^2\right] = \frac{t^2}{\tr(\bSigma_{n,p}^2)}\E\left[\left(\left\|\sum_{i=1}^n \bZ_i^{(q)}\right\|^2 - \tr(\bSigma_{n,p,q})\right)^2\right]  \\
			& \leq \frac{t^2}{\tr(\bSigma_{n,p}^2)} \E\left[\left(\sum_{i=1}^n \left\|\bZ_i^{(q)}\right\|^2 - \tr(\bSigma_{n,p,q}) + 2\sum_{i < j}\bZ_i^{(q)\top} \bZ_j^{(q)}\right)^2 \right] \\
			& \leq \frac{2t^2}{\tr(\bSigma_{n,p}^2)} \left(\E\left[\left(\sum_{i=1}^n \left\|\bZ_i^{(q)}\right\|^2 - \tr(\bSigma_{n,p,q})\right)^2\right] + 4\E\left[\left(\sum_{i < j}\bZ_i^{(q)\top} \bZ_j^{(q)} \right)^2\right]\right) \\
			& =  \frac{2t^2}{\tr(\bSigma_{n,p}^2)} \left(\Var\left(\sum_{i=1}^n \left\|\bZ_i^{(q)}\right\|^2\right) + 4\E\left[\left(\sum_{i < j}\bZ_i^{(q)\top} \bZ_j^{(q)} \right)^2\right]\right) \\
			& = 2t^2 A_{n,p} + 8t^2 B_{n,p},
		\end{align*}
		where $A_{n,p} = \displaystyle \frac{\Var\left(\sum_{i=1}^n \left\|\bZ_i^{(q)}\right\|^2\right)}{\tr(\bSigma_{n,p}^2)}$ and $B_{n,p} = \displaystyle \frac{\E\left[\left(\sum_{i < j}\bZ_i^{(q)\top} \bZ_j^{(q)} \right)^2\right]}{\tr(\bSigma_{n,p}^2)}$. Using \ref{cond3} and \autoref{fourth_moment}, we have
		\begin{align}
			A_{n,p} &=  \frac{\sum_{i=1}^{n_1} \frac{n_2^2}{n^2n_1^2}\Var(\|\bU_q^\top\bGamma_{p,1} \bz_{p,i1}\|^2) + \sum_{i=n_1 +1}^{n} \frac{n_2^2}{n^2n_1^2}\Var(\|\bU_q^\top\bGamma_{p,2} \bz_{p,i2}\|^2)}{\tr(\bSigma_{n,p}^2)}\nonumber\\
			&\leq (5 + \Delta) \frac{\frac{n_2^2}{n^2n_1}\tr(\bSigma_{p,q,1}^2) + \frac{n_1^2}{n^2n_2}\tr(\bSigma_{p,q,2}^2)}{\frac{n_2^2}{n^2}\tr(\bSigma_{p,1}^2) + \frac{n_1^2}{n^2}\tr(\bSigma_{p,2}^2) + \frac{2n_1n_2}{n^2} \tr(\bSigma_{p,1}\bSigma_{p,2})}. \label{thm1:lem1:eq1}
		\end{align}
		Using the orthogonality of $\bu_{p,r}$, we have 
		\begin{align}
			\tr(\bSigma_{p,q,i} \bSigma_{p,q,j}) &= \tr(\bU_q^\top \bSigma_{p,i} \bU_q \bU_q^\top \bSigma_{p,j} \bU_q)= \tr(\bSigma_{p,i} \bU_q \bU_q^\top \bSigma_{p,j} \bU_q \bU_q^\top) \nonumber\\
			&\leq \sqrt{\tr(\bSigma_{p,i} \bU_q \bU_q^\top \bU_q \bU_q^\top \bSigma_{p,i}) \tr(\bSigma_{p,j} \bU_q \bU_q^\top \bU_q \bU_q^\top \bSigma_{p,j})} \nonumber\\
			&\leq \frac{1}{2} \tr(\bSigma_{p,i}^2 \bU_q \bU_q^\top) + \frac{1}{2} \tr(\bSigma_{p,j}^2 \bU_q \bU_q^\top) \nonumber\\
			&= \frac{1}{2}\left(\tr(\bSigma_{p,i}^2) - \sum_{r=1}^q \tr(\bSigma_{p,i}^2\bu_{p,r}\bu_{p,r}^\top)\right) + \frac{1}{2}\left(\tr(\bSigma_{p,j}^2) - \sum_{r=1}^q \tr(\bSigma_{p,j}^2\bu_{p,r}\bu_{p,r}^\top)\right) \nonumber\\
			&\leq \frac{1}{2} \tr(\bSigma_{p,i}^2) + \frac{1}{2} \tr(\bSigma_{p,j}^2) \label{thm1:lem1:eq2}
		\end{align}
		for $i,j \in \{1,2\}$ since $\tr(\bA) \geq 0$ and $\tr(\bA\bB) \geq 0$ for any two symmetric non-negative definite matrices $\bA$ and $\bB$. From (\ref{thm1:lem1:eq1}), we obtain
		\begin{align*}
			&A_{n,p} \leq (5 + \Delta) \frac{\frac{n_2^2}{n^2n_1}\tr(\bSigma_{p,1}^2) + \frac{n_1^2}{n^2n_2}\tr(\bSigma_{p,2}^2)}{\frac{n_2^2}{n^2}\tr(\bSigma_{p,1}^2) + \frac{n_1^2}{n^2}\tr(\bSigma_{p,2}^2) + \frac{2n_1n_2}{n^2} \tr(\bSigma_{p,1}\bSigma_{p,2})}\\
			& =(5+\Delta)\left[\frac{\frac{n_2^2}{n^2n_1}\tr(\bSigma_{p,1}^2)}{\frac{n_2^2}{n^2}\tr(\bSigma_{p,1}^2) + \frac{n_1^2}{n^2}\tr(\bSigma_{p,2}^2) + \frac{2n_1n_2}{n^2} \tr(\bSigma_{p,1}\bSigma_{p,2})}  \right.\\
			&\qquad\qquad\qquad\qquad\qquad\qquad+  \left. \frac{ \frac{n_1^2}{n^2n_2}\tr(\bSigma_{p,2}^2)}{\frac{n_2^2}{n^2}\tr(\bSigma_{p,1}^2) + \frac{n_1^2}{n^2}\tr(\bSigma_{p,2}^2) + \frac{2n_1n_2}{n^2} \tr(\bSigma_{p,1}\bSigma_{p,2})}\right] \\
			&\leq (5+\Delta)\left[\frac{\frac{n_2^2}{n^2n_1}\tr(\bSigma_{p,1}^2)}{\frac{n_2^2}{n^2}\tr(\bSigma_{p,1}^2)}   +   \frac{ \frac{n_1^2}{n^2n_2}\tr(\bSigma_{p,2}^2)}{ \frac{n_1^2}{n^2}\tr(\bSigma_{p,2}^2) }\right]\\
			&= (5 + \Delta) \left(\frac{1}{n_1} + \frac{1}{n_2}\right).
		\end{align*}
		Similarly, we have
		\begin{align*}
			B_{n,p} &= \frac{\sum_{i < j}\sum_{i' < j'}\E[\bZ_i^{(q)\top} \bZ_j^{(q)}\bZ_{i'}^{(q)\top} \bZ_{j'}^{(q)}]}{\tr(\bSigma_{n,p}^2)}\\
			&= \frac{\sum_{i < j}\E[\bZ_i^{(q)\top} \bZ_j^{(q)}\bZ_{i}^{(q)\top} \bZ_{j}^{(q)}]}{\tr(\bSigma_{n,p}^2)} = \frac{\sum_{i < j}\E[\tr(\bZ_i^{(q)} \bZ_i^{(q)\top}\bZ_{j}^{(q)} \bZ_{j}^{(q)\top)}]}{\tr(\bSigma_{n,p}^2)} \\
			& = \frac{\frac{n_1(n_1 - 1)}{2} \cdot \frac{n_2^2}{n^2n_1^2}\tr(\bSigma_{p,q,1}^2) + \frac{n_2(n_2 - 1)}{2} \cdot \frac{n_1^2}{n^2n_2^2}\tr(\bSigma_{p,q,2}^2) + n_1 n_2 \cdot \frac{1}{n^2}\tr(\bSigma_{p,q,1} \bSigma_{p,q,2})}{\frac{n_2^2}{n^2}\tr(\bSigma_{p,1}^2) + \frac{n_1^2}{n^2}\tr(\bSigma_{p,2}^2) + \frac{2n_1n_2}{n^2} \tr(\bSigma_{p,1}\bSigma_{p,2})} \\
			& \leq \frac{\frac{n_2^2}{n^2}\tr(\bSigma_{p,q,1}^2) +  \frac{n_1^2}{n^2}\tr(\bSigma_{p,q,2}^2) +  \frac{2n_1 n_2}{n^2}\tr(\bSigma_{p,q,1} \bSigma_{p,q,2})}{\frac{n_2^2}{n^2}\tr(\bSigma_{p,1}^2) + \frac{n_1^2}{n^2}\tr(\bSigma_{p,2}^2) + \frac{2n_1n_2}{n^2} \tr(\bSigma_{p,1}\bSigma_{p,2})} \\
			&= \tilde{B}_{n,p} + \frac{\left(\frac{n_2^2}{n^2} - (1-\tau)^2\right)\tr(\bSigma_{p,q,1}^2) +  \left(\frac{n_1^2}{n^2} - \tau^2\right)\tr(\bSigma_{p,q,2}^2) +  \left(\frac{2n_1 n_2}{n^2} - 2\tau(1-\tau)\right)\tr(\bSigma_{p,q,1} \bSigma_{p,q,2})}{\frac{n_2^2}{n^2}\tr(\bSigma_{p,1}^2) + \frac{n_1^2}{n^2}\tr(\bSigma_{p,2}^2) + \frac{2n_1n_2}{n^2} \tr(\bSigma_{p,1}\bSigma_{p,2})} \\
			& \leq \tilde{B}_{n,p} + \frac{\left|\frac{n_2^2}{n^2} - (1-\tau)^2\right|\tr(\bSigma_{p,1}^2) +  \left|\frac{n_1^2}{n^2} - \tau^2\right|\tr(\bSigma_{p,2}^2) +  \left|\frac{2n_1 n_2}{n^2} - 2\tau(1-\tau)\right|\left(\frac{1}{2}\tr(\bSigma_{p,1}^2) + \frac{1}{2}\tr(\bSigma_{p,2}^2)\right)}{\frac{n_2^2}{n^2}\tr(\bSigma_{p,1}^2) + \frac{n_1^2}{n^2}\tr(\bSigma_{p,2}^2) + \frac{2n_1n_2}{n^2} \tr(\bSigma_{p,1}\bSigma_{p,2})} \\
			& \leq \tilde{B}_{n,p} + \left|\frac{n_2^2}{n^2} - (1-\tau)^2\right| \frac{n^2}{n_2^2} + \left|\frac{n_1^2}{n^2} - \tau^2\right| \frac{n}{n_1^2} + \left|\frac{2n_1 n_2}{n^2} - 2\tau(1-\tau)\right| \frac{n^2}{2n_2^2} + \left|\frac{2n_1 n_2}{n^2} - 2\tau(1-\tau)\right| \frac{n^2}{2n_1^2},
		\end{align*}
		where 
		\begin{align*}
			\tilde{B}_{n,p} & = \frac{(1-\tau)^2\tr(\bSigma_{p,q,1}^2) +  \tau^2\tr(\bSigma_{p,q,2}^2) + 2\tau(1-\tau)\tr(\bSigma_{p,q,1} \bSigma_{p,q,2})}{\frac{n_2^2}{n^2}\tr(\bSigma_{p,1}^2) + \frac{n_1^2}{n^2}\tr(\bSigma_{p,2}^2) + \frac{2n_1n_2}{n^2} \tr(\bSigma_{p,1}\bSigma_{p,2})} \\
			&= \frac{\tr(\bSigma_{p,q}^2)}{\tr(\bSigma_{p}^2)} \times \frac{(1-\tau)^2\tr(\bSigma_{p,1}^2) + \tau^2\tr(\bSigma_{p,2}^2) + 2\tau(1-\tau) \tr(\bSigma_{p,1}\bSigma_{p,2})}{\frac{n_2^2}{n^2}\tr(\bSigma_{p,1}^2) + \frac{n_1^2}{n^2}\tr(\bSigma_{p,2}^2) + \frac{2n_1n_2}{n^2} \tr(\bSigma_{p,1}\bSigma_{p,2})} \\
			& \leq \frac{\tr(\bSigma_{p,q}^2)}{\tr(\bSigma_{p}^2)} \left(\tau^2 \frac{n^2}{n_1^2} + (1-\tau)^2 \frac{n^2}{n_2^2} + \tau(1-\tau) \frac{n^2}{n_1n_2}\right) \\ 
			&= \left(\tau^2 \frac{n^2}{n_1^2} + (1-\tau)^2 \frac{n^2}{n_2^2} + \tau(1-\tau) \frac{n^2}{n_1n_2}\right) \frac{\tr(\bU_q^\top \bSigma_p \bU_q \bU_q^\top \bSigma_p \bU_q)}{\tr(\bSigma_{p}^2)} \\
			& \leq \left(\tau^2 \frac{n^2}{n_1^2} + (1-\tau)^2 \frac{n^2}{n_2^2} + \tau(1-\tau) \frac{n^2}{n_1n_2}\right) \frac{\tr(\bSigma_p^2 \bU_q \bU_q^\top)}{\tr(\bSigma_{p}^2)} \qquad[\text{using (\ref{thm1:lem1:eq2})}]\\
			& = \left(\tau^2 \frac{n^2}{n_1^2} + (1-\tau)^2 \frac{n^2}{n_2^2} + \tau(1-\tau) \frac{n^2}{n_1n_2}\right) \frac{\sum_{r=q+1}^p \bu_{p,r}^\top \bSigma_{p}^2 \bu_{p,r}}{\tr(\bSigma_{p}^2)} \\
			&= \left(\tau^2 \frac{n^2}{n_1^2} + (1-\tau)^2 \frac{n^2}{n_2^2} + \tau(1-\tau) \frac{n^2}{n_1n_2}\right) \frac{\sum_{r=q+1}^p \lambda_{p,r}^2}{\sum_{r=1}^p \lambda_{p,r}^2} \\
			&\leq \left(\tau^2 \frac{n^2}{n_1^2} + (1-\tau)^2 \frac{n^2}{n_2^2} + \tau(1-\tau) \frac{n^2}{n_1n_2}\right) \sup_p\sum_{r=q+1}^p \rho_{p,r}^2.
		\end{align*}
		Using \ref{cond5} and the above upper bounds, we can obtain $Q$ and $N$ large enough so that for any $q \geq Q$ and $n \geq N$, $(2t^2  A_{n, p} + 8t^2  B_{n, p}) < \epsilon^2$ and hence
		$$ |\varphi_{n,p}(t) - \varphi_{n,p}^{(q)}(t)| \leq \sqrt{2t^2  A_{n, p} + 8t^2  B_{n, p}} < \epsilon.$$
		This completes the proof.
	\end{proof}

	\begin{proof}[Proof of \autoref{estimation_auxiliary}]
		We have
		$$\frac{\tr\left(\E\left[(\hat\bSigma_p^{(1)} - \bSigma_p)^2\right]\right)}{\tr(\bSigma_p^2)}  = \frac{\E[\tr((\hat\bSigma_p^{{(1)} ~ 2})]}{\tr(\bSigma_p^2)} - 2 \frac{\E[\tr(\hat\bSigma_p^{(1)} \bSigma_p)]}{\tr(\bSigma_p^2)} + 1.$$
		Since
		\begin{align*}
			\hat\bSigma_{p}^{(1)} &= \frac{1}{n n_1 n_2} \sum_{i=1}^{n_1}\left( \sum_{j=1}^{n_2} (\bh_p(\bX_{p,i}, \bY_{p,j}) - \hat{\bdelta}_p) \right) \left( \sum_{j=1}^{n_2} (\bh_p(\bX_{p,i}, \bY_{p,j}) - \hat{\bdelta}_p) \right)^\top \\
			&\;\;\;\;\;+ \frac{1}{n n_1 n_2} \sum_{j=1}^{n_2}\left( \sum_{i=1}^{n_1} (\bh_p(\bX_{p,i}, \bY_{p,j}) - \hat{\bdelta}_p) \right) \left( \sum_{i=1}^{n_1} (\bh_p(\bX_{p,i}, \bY_{p,j}) - \hat{\bdelta}_p) \right)^\top 
		\end{align*}
		without loss of generality, we can assume that $\bdelta_p = \mathbf{0}_p$. First, we consider the estimator $\hat\bSigma_p^{(1)}$ with assumption \ref{cond6}. Now,
		\begin{align}
			\hat\bSigma_p^{(1)} =& \frac{1}{nn_1n_2} \sum_{k=1}^{n_1} \sum_{i=1}^{n_2} \sum_{j=1}^{n_2} \bh_p(\bX_{p,k}, \bY_{p,i}) \bh_p(\bX_{p,k}, \bY_{p,j})^\top + \frac{1}{nn_1n_2} \sum_{k=1}^{n_2} \sum_{i=1}^{n_1} \sum_{j=1}^{n_1} \bh_p(\bX_{p,i}, \bY_{p,k}) \bh_p(\bX_{p,j}, \bY_{p,k})^\top \nonumber\\
			&- \frac{1}{n_1^2 n_2^2} \sum_{i_1 = 1}^{n_1}\sum_{j_1 = 1}^{n_2} \sum_{i_2 = 1}^{n_1}\sum_{j_2 = 1}^{n_2} \bh_p(\bX_{p,i_1}, \bY_{p,j_1}) \bh_p(\bX_{p,i_2}, \bY_{p,j_2})^\top. \label{lemma_gaussian_estimation_eq2}
		\end{align}
		Then,
		\begin{align*}
			\E[\hat\bSigma_p^{(1)}] &= \left(\frac{n_2 -1}{n} - \frac{n_2 -1}{n_1 n_2} \right) \bSigma_{p,1} + \left(\frac{n_1 -1}{n} - \frac{n_1 -1}{n_1 n_2} \right) \bSigma_{p,2} - \frac{1}{n_1^2 n_2^2} \sum_{i=1}^{n_1} \sum_{j=1}^{n_2} \E[\bh_p(\bX_{p,i}, \bY_{p,j}) \bh_p(\bX_{p,i}, \bY_{p,j})^\top] \\
			&= \left(\frac{n_2 -1}{n} - \frac{n_2 -2}{n_1 n_2} \right) \bSigma_{p,1} + \left(\frac{n_1 -1}{n} - \frac{n_1 -2}{n_1 n_2} \right) \bSigma_{p,2} - \frac{1}{n_1^2 n_2^2} \sum_{i=1}^{n_1} \sum_{j=1}^{n_2} \E[\bar\bh_p(\bX_{p,i}, \bY_{p,j}) \bar\bh_p(\bX_{p,i}, \bY_{p,j})^\top].
		\end{align*}
		Therefore,
		\begin{align*}
			\tr(\E[\hat\bSigma_p^{(1)} \bSigma_p]) &=  \left(\frac{n_2 -1}{n} - \frac{n_2 -2}{n_1 n_2} \right) \tr(\bSigma_{p,1} \bSigma_p) + \left(\frac{n_1 -1}{n} - \frac{n_1 -2}{n_1 n_2} \right) \tr(\bSigma_{p,2} \bSigma_p)\\
			&- \frac{1}{n_1^2 n_2^2} \sum_{i=1}^{n_1} \sum_{j=1}^{n_2} \E[\tr(\bar\bh_p(\bX_{p,i}, \bY_{p,j}) \bar\bh_p(\bX_{p,i}, \bY_{p,j})^\top \bSigma_p)].
		\end{align*}
		Using the CS inequality, we get 
		\begin{align*}
			\tr(\bar\bh_p(\bX_{p,i}, \bY_{p,j}) \bar\bh_p(\bX_{p,i}, \bY_{p,j})^\top\bSigma_p) &= \bar\bh_p(\bX_{p,i}, \bY_{p,j})^\top\bSigma_p\bar\bh_p(\bX_{p,i}, \bY_{p,j}) \\
			&\leq \|\bh_p(\bX_{p,i}, \bY_{p,j})\| \|\bSigma_p\bar\bh_p(\bX_{p,i}, \bY_{p,j})\|\\
			&\leq \|\bh_p(\bX_{p,i}, \bY_{p,j})\|^2 \sqrt{\tr(\bSigma_p^2)}.
		\end{align*}
		After dividing by $\tr(\bSigma_p^2)$, the upper bound becomes bounded uniformly-over-$p$ via condition \ref{cond2}. Now, using arguments similar to those used to bound $A_{n,p}$ and $B_{n,p}$ in \autoref{thm:1_lemma}, we see~that
		\begin{align*}
			&\left|\frac{\E[\tr(\hat\bSigma_p^{(1)} \bSigma_p)]}{\tr(\bSigma_p^2)} - 1\right| \to 0
		\end{align*}
		uniformly-over-$p$ as $n \to \infty$. Now, we shall show that $\tr(\E[\hat\bSigma_p^{(1)~2}]) / \tr(\bSigma_p^2) \to 1$ uniformly-over-$p$ as $n \to \infty$. Consider the first term given in the right hand side of (\ref{lemma_gaussian_estimation_eq2}). We look at the expected trace of the square of this term. It is of the form 
		\begin{equation}
			\frac{1}{n^2 n_1^2 n_2^2} \sum_{k_1=1}^{n_1} \sum_{i_1=1}^{n_2} \sum_{j_1=1}^{n_2} \sum_{k_2=1}^{n_1} \sum_{i_2=1}^{n_2} \sum_{j_2=1}^{n_2} \E[\tr(\Phi(i_1, j_1, k_1) \Phi(i_2, j_2, k_2))], \label{lemma_gaussian_estimation_eq3}
		\end{equation}
		where $\Phi(i,j,k) = \bh_p(\bX_{p,k}, \bY_{p,i}) \bh_p(\bX_{p,k}, \bY_{p,j})^\top$.
		Using the CS inequality, we can bound the term inside the expectation by 
		$$ \|\bh_p(\bX_{p,k_2}, \bY_{p,j_2})\| \|\bh_p(\bX_{p,k_1}, \bY_{p,i_1}) \|  \|\bh_p(\bX_{p,k_1}, \bY_{p,j_1})\|   \|\bh_p(\bX_{p,k_2}, \bY_{p,i_2})\|.$$
		This can be further expanded in terms of $\bar\bh(\cdot, \cdot)$ and $\bGamma_{p,j} \bz_{p,ij}$ and after taking expectation, each term in the upper bound of can be written as $\E[\prod_{i=1}^4\|\bv_i\|] \leq \prod_{i=1}^4[\E[\|\bv_i\|^4]]^{1/4}$,
		where $\bv_i$ is either $\bar\bh_p(\bX_{p,k}, \bY_{p,l})$ or $\bGamma_{p,j}\bz_{p,kj}$ for some $k,l$ and $j \in \{1,2\}$. By our assumptions and \autoref{fourth_moment}, $\E[\|\bh_p(\bX_{p,k}, \bY_{p,l})\|^4] / \tr(\bSigma_p^2)$  and $\E[\|\bGamma_{p,j}\bz_{p,ij}\|^4] / \tr(\bSigma_p^2)$ are bounded uniformly-over-$p$. Therefore, we conclude that the terms $\tr[\Phi(i_1,j_1,k_1) \Phi(i_2,j_2,k_2)] /\tr(\bSigma_p^2)$ are bounded uniformly-over-$p$.
		In the expansion of (\ref{lemma_gaussian_estimation_eq3}), the number of terms for which $i_1,j_1,k_1,i_2,j_2,k_2$ are not all distinct is at most $O(n^5)$. Thus, only the terms for which $i_1,j_1,k_1,i_2,j_2,k_2$ are all distinct will survive for large $n$. For any such term, we have
		$$\E[\tr[\Phi(i_1,j_1,k_1) \Phi(i_2,j_2,k_2)]] = \tr(\bSigma_{p,1}^2)$$ 
		and the number of such terms is $n_1^2n_2^4$. Hence, after dividing the expression in (\ref{lemma_gaussian_estimation_eq3}) by $\tr(\bSigma_p^2)$, it converges to $(1-\tau)^2 \tr(\bSigma_{p,1}^2) / \tr(\bSigma_p^2)$ uniformly-over-$p$ as $n \to \infty$. A similar argument can be used to show that the square of the second term in (\ref{lemma_gaussian_estimation_eq2}) and the cross terms between the first and second terms contribute to $(1-\tau)^2 \tr(\bSigma_{p,2}^2) / \tr(\bSigma_p^2)$ and  $2\tau(1-\tau) \tr(\bSigma_{p,1} \bSigma_{p,2}) / \tr(\bSigma_p^2)$, respectively. All other terms contribute to $0$ asymptotically. This proves our assertion.
		
		Now, consider the estimator $\hat\bSigma_{p}^{(2)}$ with assumption \ref{cond7}. Let $\hat\bSigma_{p}^{(1)} = ((\hat\sigma_{ij}))_{p \times p}$ and $\bSigma_{p} = ((\sigma_{ij}))_{p \times p}$. We need to show that $\E[\tr[(\hat\bSigma_{p}^{(2)} - \bSigma_p)^2]] / \tr(\bSigma_p^2)$ converge to $0$ uniformly-over-$p$ as $n \to \infty$. Now,
		\begin{align*}
			\frac{\E\left[\tr[(\hat\bSigma_{p}^{(2)} - \bSigma_p)^2]\right]}{\tr(\bSigma_p^2)} &= \frac{\E\left[\sum_{i=1}^p \sum_{j=1}^p [w_{ij}(\hat{\sigma}_{ij} - \sigma_{ij}) + \sigma_{ij}(w_{ij} -1)]^2\right]}{\tr(\bSigma_p^2)}\\
			&\leq \frac{\E\left[\sum_{i=1}^p \sum_{j=1}^p 2w_{ij}^2(\hat{\sigma}_{ij}  - \sigma_{ij})^2\right] + \sum_{i=1}^p \sum_{j=1}^p 2\sigma_{ij}^2(w_{ij} -1)^2}{\tr(\bSigma_p^2)} \\
			&\leq 2 \frac{\E\left[\sum_{|i-j| < k} (\hat{\sigma}_{ij}  - \sigma_{ij})^2\right]}{\tr(\bSigma_p^2)} +  2\frac{\sum_{|i-j| > k / 2} \sigma_{ij}^2 }{\tr(\bSigma_p^2)} \\
			&= 2 \frac{\E[\tr[(\hat\bSigma_{p}^\ast - \bSigma_p^\ast)^2]]}{\tr(\bSigma_p^2)} +  2\frac{\sum_{|i-j| > k / 2} \sigma_{ij}^2 }{\tr(\bSigma_p^2)},
		\end{align*}
		where the $(i,j)$-th entry of $\bSigma_p^\ast$ and $\hat\bSigma_{p}^\ast$ are same as that of $\bSigma_{p}$ and $\hat\bSigma_{p}^{(1)}$ if $|i-j| < k$ and both equal to $0$ if $|i-j| \geq k$. To show that the first term goes to $0$ uniformly-over-$p$, we use a similar argument used for $\E[\tr[(\hat\bSigma_{p}^{(1)} - \bSigma_p)^2]] / \tr(\bSigma_p^2)$ in the previous part of the proof. This line of argument immediately gives $\E[\tr(\hat\bSigma_p^\ast\bSigma_p^\ast)] / \tr(\bSigma_p^2) \to 1$ uniformly-over-$p$. Next, we follow the same argument to show $\tr(\E[(\hat\bSigma_p^\ast)^2]) / \tr(\bSigma_p^2) \to 1$ uniformly-over-$p$. Here, we need to bound the terms of the~form 
		$$\E[\tr(\tilde\Phi(i_1, j_1, k_1) \tilde\Phi(i_2, j_2, k_2))] = \sum_{1 \leq l_1, l_2 \leq p} [\tilde\Phi(i_1, j_1, k_1)]_{(l_1, l_2)} [\tilde\Phi(i_2, j_2, k_2))]_{(l_2, l_1)}$$
		uniformly-over-$p$, where $[\tilde\Phi(i, j, k)]_{(l_1, l_2)}$ is the $(l_1,l_2)$-th entry of $\tilde\Phi(i, j, k)$ and is the same as $[\Phi(i, j, k)]_{(l_1, l_2)}$ if $|l_1 - l_2| < k$ and equal to $0$ otherwise. The number of non-zero terms inside the sum is bounded by $2kp$ and $[\Phi(i, j, k)]_{(l_1, l_2)} = [\bh_p(\bX_{p,k}, \bY_{p,i})]_{l_1} [\bh_p(\bX_{p,k}, \bY_{p,j})]_{l_2}$. So, using the CS inequality, we get the following bound $$\E[\tr(\tilde\Phi(i_1, j_1, k_1) \tilde\Phi(i_2, j_2, k_2))] \leq C_1 \lambda_1^2 pk$$
		for some constant $C_1$ (which does not depend on $p,n,k$).
		Since $p\lambda_1^2 / \tr(\bSigma_p^2)$ is bounded uniformly-over-$p$, we have
		$$ \frac{\E[\tr[(\hat\bSigma_{p}^\ast - \bSigma_p^\ast)^2]]}{\tr(\bSigma_p^2)} \leq C_2 \frac{k}{n} + D_{n,p}$$
		with $C_2$ free of $p, n, k$ and $D_{n,p} \to 0$ uniformly-over-$p$ as $n \to \infty$. It can be shown that $\sum_{|i-j| > k / 2} \sigma_{ij}^2 \leq C_3\lambda_1^2 k^{-2\beta -1}$ for some $C_3$ (independent of $n,p,k$). Thus, we get
		$$\frac{\E\left[\tr[(\hat\bSigma_{p}^{(2)} - \bSigma_p)^2]\right]}{\tr(\bSigma_p^2)} \leq C_4 \left(\frac{k}{n} + k^{-2\beta -1}\right) + D_{n,p}$$
		for some $C_4$ (which does not depend on $n,p,k$). For $k = \min(n^{1/(2\beta + 2)}, p)$, the upper bound converges to $0$ uniformly-over-$p$. This completes the proof.
	\end{proof}
	
	\newpage
	\section{Additional numerical work}  \label{sec:supplement_numerical}
	
	In this section, we report additional numerical results for $p=25$ and $p=50$.
	Table~\ref{tab:size_p=25} reports estimated sizes for $p=25$, while Table~\ref{tab:size_p=50} presents the same for $p=50$.
	
	\begin{table}[h!]
		\centering
		\caption{Estimated sizes at nominal level 5\% for the different tests for multivariate data based on $10000$ independent replications for $p=25$ with $n_1 = 40$ and $n_2 = 50$.}
		\label{tab:size_p=25}
		
		\vspace{0.1in}
		\begin{tabular}{cccccccccc}
			\hline
			%Model & 1.i. & 1.ii. & 1.iii. & 2.i. & 2.ii. & 2.iii. & 3.i. & 3.ii. & 3.iii. \\ \hline
			Model & 1.i. & 2.i. & 3.i. & 1.ii. & 2.ii. & 3.ii. & 1.iii. & 2.iii. & 3.iii. \\ \hline
			KCDG2025$^1$    & 0.0492 & 0.0404 & 0.0072 & 0.0308 & 0.0136 & 0.0004 & 0.0484 & 0.0370 & 0.0020 \\
			KCDG2025$^2$  & 0.1014 & 0.0938 & 0.0316 & 0.0436 & 0.0266 & 0.0014 & 0.0560 & 0.0482 & 0.0080 \\
			sKCDG2025$^1$  & 0.0493 & 0.0502 & 0.0476 & 0.0348 & 0.0328 & 0.0296 & 0.0502 & 0.0473 & 0.0474 \\
			sKCDG2025$^2$  & 0.0911 & 0.0924 & 0.0909 & 0.0438 & 0.0447 & 0.0426 & 0.0500 & 0.0546 & 0.0511 \\
			ZGZC2020    & 0.0550 & 0.0472 & 0.0102 & 0.0542 & 0.0238 & 0.0006 & 0.0590 & 0.0454 & 0.0034 \\
			BS1996      & 0.0696 & 0.0620 & 0.0146 & 0.0656 & 0.0344 & 0.0010 & 0.0758 & 0.0590 & 0.0064 \\
			CLX2014     & 0.0380 & 0.0308 & 0.0028 & 0.0526 & 0.0384 & 0.0044 & 0.0438 & 0.0298 & 0.0034 \\
			CQ2010      & 0.0698 & 0.0620 & 0.0136 & 0.0664 & 0.0346 & 0.0008 & 0.0760 & 0.0590 & 0.0056 \\
			CLZ2014     & 0.1288 & 0.1286 & 0.0670 & 0.0588 & 0.0440 & 0.0078 & 0.1270 & 0.1162 & 0.0408 \\
			SD2008      & 0.0338 & 0.0320 & 0.0052 & 0.0562 & 0.0262 & 0.0004 & 0.0472 & 0.0346 & 0.0020 \\
			\hline
		\end{tabular}
	\end{table}
	
	\begin{table}[h!]
		\centering
		\caption{Estimated sizes at nominal level 5\% for the different tests for multivariate data based on $10000$ independent replications for $p=50$ with $n_1 = 40$ and $n_2 = 50$.}
		\label{tab:size_p=50}
		
		\vspace{0.1in}
		\begin{tabular}{cccccccccc}
			\hline
			%Model & 1.i. & 1.ii. & 1.iii. & 2.i. & 2.ii. & 2.iii. & 3.i. & 3.ii. & 3.iii. \\ \hline
			Model & 1.i. & 2.i. & 3.i. & 1.ii. & 2.ii. & 3.ii. & 1.iii. & 2.iii. & 3.iii. \\ \hline
			KCDG2025$^1$    & 0.0527 & 0.0403 & 0.0063 & 0.0147 & 0.0067 & 0.0000 & 0.0413 & 0.0257 & 0.0007 \\
			KCDG2025$^2$  & 0.1490 & 0.1303 & 0.0487 & 0.0373 & 0.0357 & 0.0010 & 0.0500 & 0.0430 & 0.0077 \\
			sKCDG2025$^1$  & 0.0512 & 0.0500 & 0.0484 & 0.0207 & 0.0233 & 0.0153 & 0.0430 & 0.0368 & 0.0397 \\
			sKCDG2025$^2$  & 0.1323 & 0.1251 & 0.1254 & 0.0432 & 0.0450 & 0.0406 & 0.0520 & 0.0530 & 0.0534 \\
			ZGZC2020    & 0.0620 & 0.0507 & 0.0090 & 0.0500 & 0.0190 & 0.0000 & 0.0503 & 0.0337 & 0.0010 \\
			BS1996      & 0.0790 & 0.0637 & 0.0130 & 0.0590 & 0.0243 & 0.0003 & 0.0620 & 0.0463 & 0.0030 \\
			CLX2014     & 0.0420 & 0.0257 & 0.0030 & 0.0640 & 0.0490 & 0.0017 & 0.0480 & 0.0293 & 0.0030 \\
			CQ2010      & 0.0793 & 0.0637 & 0.0123 & 0.0600 & 0.0250 & 0.0003 & 0.0623 & 0.0453 & 0.0020 \\
			CLZ2014     & 0.1680 & 0.1497 & 0.0957 & 0.0607 & 0.0553 & 0.0067 & 0.1497 & 0.1270 & 0.0420 \\
			SD2008      & 0.0337 & 0.0267 & 0.0033 & 0.0517 & 0.0203 & 0.0000 & 0.0450 & 0.0260 & 0.0003 \\
			
			\hline
		\end{tabular}
	\end{table}
	
	%\clearpage
	We observe from Tables \ref{tab:size_p=25} and \ref{tab:size_p=50} that
	%Represents the moderate-dimensional case.
	most tests have stable estimated sizes (corresponds to the case when $\delta = 0$). While the power function generally increases with $\delta$ for all methods, 
	differences among some methods are visible too.
	Classical and newer KCDG tests (especially, the sKCDG2025 test) show slightly higher power as $\delta$ increases. The power functions of tests like BS1996 and CQ2010 rise at a slower rate in models 2.ii and 3.ii.
	Overall, sKCDG2025 appears to be strong consistently across all the models considered in the paper.
	
	%Represents the higher-dimensional case (larger $p$).
	The overall power performance diverges more clearly among methods.
	KCDG2025 and its version based on the spatial sign (sKCDG2025) retain high power and robustness across all models. The power functions of classical tests like BS1996, CQ2010 and SD2008 tend to grow more slowly( with $\delta$) in complex model structures like Model 3.
	CLX2014 and CLZ2014 perform moderately, but still below KCDG variants for large $\delta$.
	The performance gap between these methods in fact widens with higher values of $p$.
	%compared to the $p=25$ case.
	
	When the dimension is moderate ($p = 25$), the overall power performance of the competing tests is fairly similar across most models. As the signal strength $\delta$ increases, the power of all methods rises smoothly from near zero to close to one. The differences between methods are noticeable, but not very large. The proposed tests KCDG2025 and sKCDG2025 tend to achieve slightly higher power than the older classical tests 
	%such as BS1996, CQ2010, and SD2008 
	(especially, under Models 2 and 3). Nevertheless, even the traditional methods maintain reasonable power in this moderate-dimensional scenario (see Figure \ref{fig:fig_p=25}).
	
	In contrast, when the dimension increases to $p = 50$, the differences between the tests become much more clear. The KCDG2025 tests (including both its original and kernel variants) continues to show strong and stable performance, maintaining high power across all models as $\delta$ grows. Meanwhile, the classical tests, particularly, BS1996, CQ2010 and SD2008 exhibit a substantial loss of power, indicating their reduced effectiveness in higher dimensions. Methods such as CLX2014 and CLZ2014 perform moderately, but still lag behind the KCDG based tests for large values of $\delta$~(see Figure~\ref{fig:fig_p=50}).
	
	To summarize, this comparison shows that while most tests perform comparably in moderate dimensions, the advantage of the sKCDG2025 method becomes clearly evident in higher dimensions. The proposed tests display better scalability and robustness, retaining high power across different model configurations, whereas several traditional multivariate tests degrade significantly as the data dimension increases.
	
	\begin{figure}[!h]
		\centering
		\includegraphics[width=0.95\textwidth]{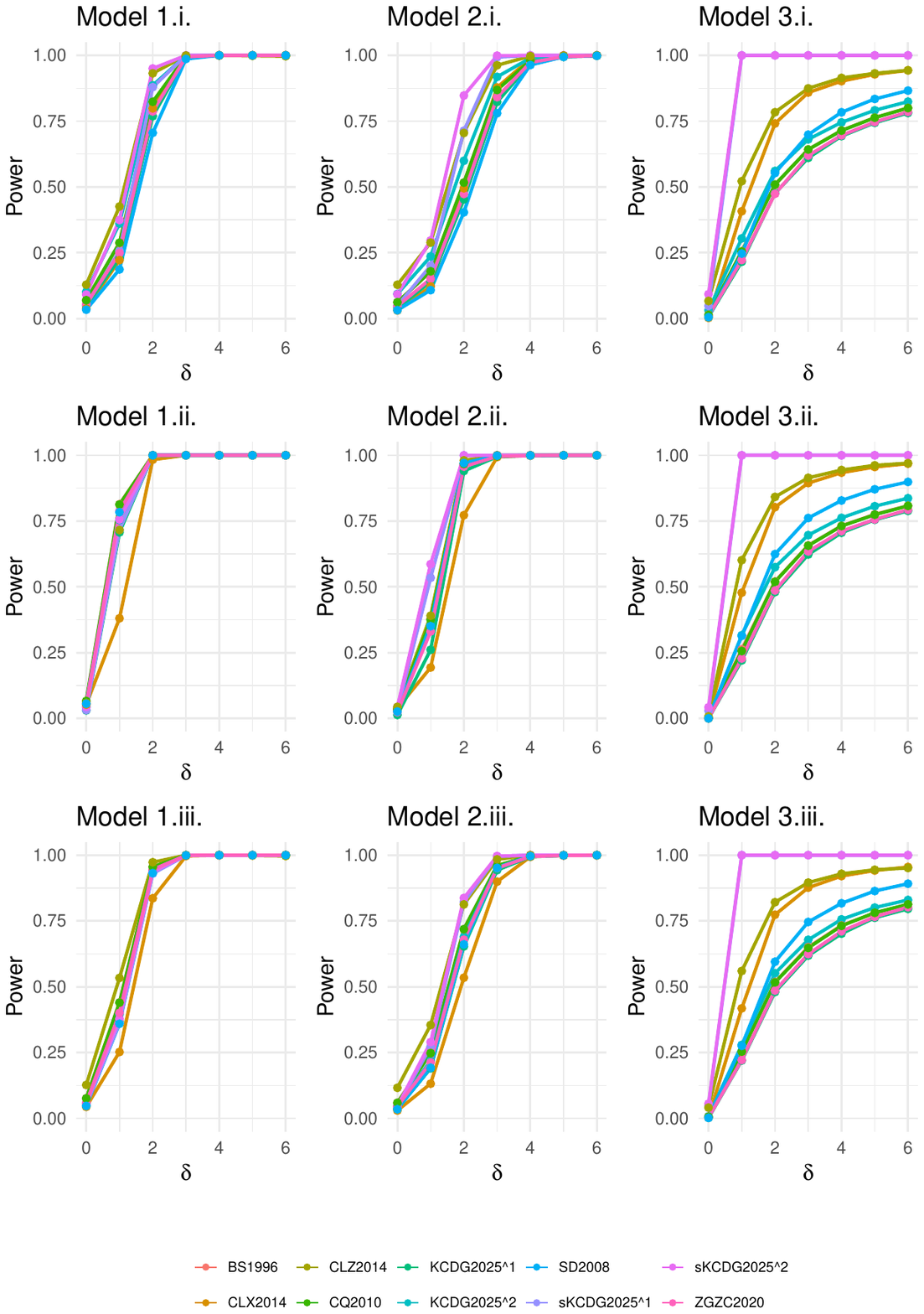}
		\caption{Estimated powers at nominal level 5\% based on $10000$ independent replications at $p=25$ for $n_1 = 40$ and $n_2 = 50$.}
		\label{fig:fig_p=25}
	\end{figure}
	
	\begin{figure}[!h]
		\centering
		\includegraphics[width=0.95\textwidth]{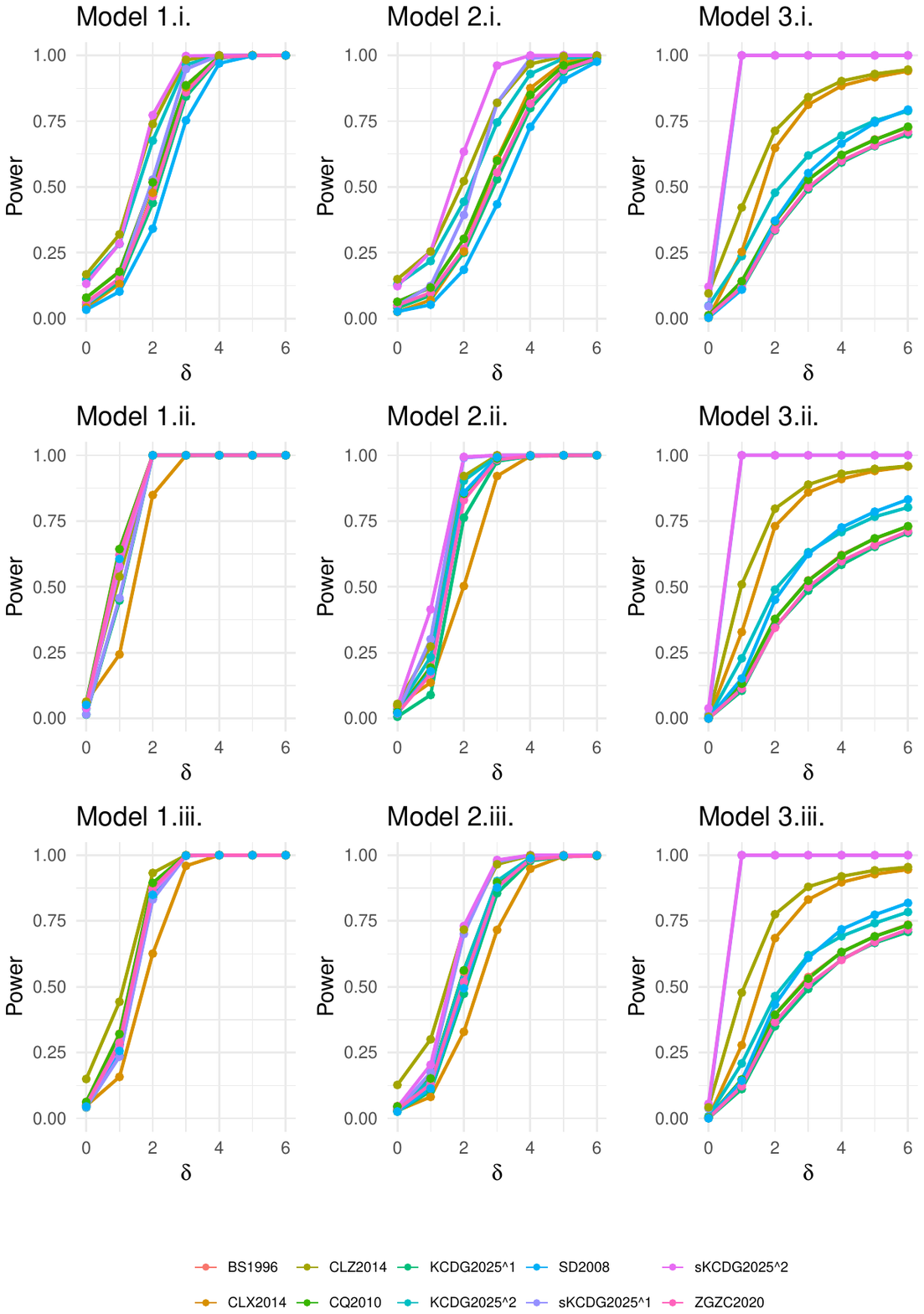}
		\caption{Estimated powers at nominal level 5\% based on $10000$ independent replications at $p=50$ for $n_1 = 40$ and $n_2 = 50$.}
		\label{fig:fig_p=50}
	\end{figure}

\clearpage
\small
%\scriptsize
\bibliographystyle{apalike}
\bibliography{bibliography}

\end{document}